\let\shortversion\a
\let\longversion\x

\documentclass[conference,letterpaper]{IEEEtran}
\usepackage[utf8]{inputenc}
\usepackage{amsmath,amssymb,amsthm,amsfonts}
\usepackage{accents}
\usepackage{mathtools,thmtools}
\usepackage{tikz,pgfplots}
\usepackage{bm, bbm}
\usepackage{color,xcolor,colortbl}
\usepackage{watermark}

\ifx\longversion\x
\fi

\pgfplotsset{compat = 1.7}
\usetikzlibrary{positioning,shapes,arrows}
\usetikzlibrary{decorations.markings}

\renewcommand{\mathbf}[1]{{\bm{#1}}}

\newcommand{\eg}{\emph{e.g.}}
\newcommand{\etal}{\emph{et al.}}
\newcommand{\ie}{\emph{i.e.}}
\newcommand{\Etc}{\emph{Etc.}}

\renewcommand{\leq}{\leqslant}
\renewcommand{\geq}{\geqslant}

\newcommand{\set}[1]{\mathcal{#1}}
\DeclareMathOperator{\Tr}{Tr}
\newcommand{\defeq}{\triangleq}

\newcommand{\setS}{\set{S}}
\newcommand{\setU}{\set{U}}
\newcommand{\setX}{\set{X}}
\newcommand{\setY}{\set{Y}}

\newcommand{\matr}[1]{#1}
\newcommand{\vect}[1]{\mathbf{#1}}

\newcommand{\vs}{\vect{s}}

\newcommand{\tvs}{\tilde{\vect{s}}}
\newcommand{\vu}{\vect{u}}
\newcommand{\hvu}{\hat{\vect{u}}}
\newcommand{\vX}{\vect{X}}
\newcommand{\vx}{\vect{x}}
\newcommand{\cx}{\check{x}}
\newcommand{\cvx}{\check{\vect{x}}}
\newcommand{\vY}{\vect{Y}}
\newcommand{\vy}{\vect{y}}
\newcommand{\cy}{\check{y}}
\newcommand{\cvy}{\check{\vect{y}}}

\newcommand{\muY}{\mu^{\sY}}
\newcommand{\muXY}{\mu^{\sXY}}

\newcommand{\bmuY}{\bar \mu^{\sY}}

\newcommand{\sigmaY}{\sigma^{\sY}}
\newcommand{\sigmaXY}{\sigma^{\sXY}}

\newcommand{\bsigmaY}{\bar \sigma^{\sY}}

\newcommand{\lambdaY}{\lambda^{\sY}}

\newcommand{\HAt}[1]{\mathcal{H}^{\mathrm{A}_{#1}}}

\newcommand{\HBt}[1]{\mathcal{H}^{\mathrm{B}_{#1}}}

\newcommand{\HSt}[1]{\mathcal{H}^{\mathrm{S}_{#1}}}
\newcommand{\DensOp}{\mathcal{D}}

\newcommand{\rhoA}{\rho^{\mathrm{A}}}
\newcommand{\rhoAxell}{\rhoA_{x_{\ell}}}

\newcommand{\rhoAt}[1]{\rho^{\mathrm{A}_{#1}}}
\newcommand{\rhoBt}[1]{\rho^{\mathrm{B}_{#1}}}
\newcommand{\rhoSt}[1]{\rho^{\mathrm{S}_{#1}}}
\newcommand{\rhoAtStp}[1]{\rho^{\mathrm{A}_{#1}\mathrm{S}_{#1-1}}}
\newcommand{\rhoBtSt}[1]{\rho^{\mathrm{B}_{#1}\mathrm{S}_{#1}}}

\newcommand{\sigmaSt}[1]{\sigma^{\mathrm{S}_{#1}}}

\newcommand{\TrBt}[1]{\Tr_{\mathrm{B}_{#1}}}
\newcommand{\TrSt}[1]{\Tr_{\mathrm{S}_{#1}}}
\newcommand{\ISt}[1]{I^{\mathrm{S}_{#1}}}

\newcommand{\sX}{\mathsf{X}}
\newcommand{\sY}{\mathsf{Y}}
\newcommand{\sXY}{\mathsf{X,Y}}

\newcommand{\Herm}{\mathsf{H}}

\newcommand{\pgood}{p_{\mathrm{g}}}
\newcommand{\pbad}{p_{\mathrm{b}}}

\DeclareMathOperator{\tr}{tr}

\newtheorem{Lemma}{Lemma}

\newtheorem{Example}[Lemma]{Example}

\newcounter{mytempeqcounter}

\newcommand{\bigformulatop}[2]{%
  \begin{figure*}[!t]
    \normalsize
    \setcounter{mytempeqcounter}{\value{equation}}
    \setcounter{equation}{#1}
    #2

    \setcounter{equation}{\value{mytempeqcounter}}
    \hrulefill
    \vspace*{4pt}
  \end{figure*}
}

\newcommand{\onestareq}{\overset{\text{(a)}}{=}}
\newcommand{\twostarseq}{\overset{\text{(b)}}{=}}

\newcommand{\onestar}{\text{(a)}}
\newcommand{\twostars}{\text{(b)}}

\declaretheorem[style=plain]{theorem}

\declaretheorem[style=plain,sibling=theorem]{lemma}

\declaretheorem[style=definition,sibling=theorem,
                                 qed=$\blacksquare$]{definition}

\pgfdeclarelayer{bg}    
\pgfdeclarelayer{fg}    
\pgfsetlayers{bg,main,fg}  

\tikzset{
    -*/.style={
        shorten >=#1,
        decoration={
            markings,
            mark={
                at position 1
                with {
                    \draw[fill] circle [radius=#1];
                }
            }
        },
        postaction=decorate
    },
    -*/.default=1.3pt
}

\tikzset{
    *-/.style={
        shorten <=#1,
        decoration={
            markings,
            mark={
                at position 0
                with {
                    \draw[fill] circle [radius=#1];
                }
            }
        },
        postaction=decorate
    },
    *-/.default=1.3pt
}

\tikzset{
    *-*/.style={
        shorten <=#1,
        shorten >=#1,
        decoration={
            markings,
            mark={
                at position 0
                with {
                    \draw[fill] circle [radius=#1];
                }
            },
            mark={
                at position 1
                with {
                    \draw[fill] circle [radius=#1];
                }
            },
        },
        postaction=decorate
    },
    *-*/.default=1.3pt
}

\newcommand{\figCFSMChighlevel}{
  \resizebox{\columnwidth}{!}{
  \begin{tikzpicture}[node/.style={draw=none},
      factor/.style={rectangle, minimum width=1cm, minimum height=.7cm, 
                     draw},
      sfactor/.style={rectangle, minimum size=.4cm, draw}]
  \node[sfactor] (S) {$p_{\tilde S_0}$};
  \node[factor] (E1) [right=.7cm of S] {$W$};
  \draw (S) -- (E1) node[above=-.05cm,midway] {$\tilde s_0$};
  \node[sfactor] (X1) [above=1.1cm of E1] {$p_X$};
  \draw (X1) -- (E1) node[right=-.05cm, midway] {$x_1$};
  \draw (E1.south) -- ([yshift=-.6cm]E1.south) node[right=-.05cm] {$y_1$};
  
  \node[factor] (E2) [right=1cm of E1] {$W$};
  \draw (E1) -- (E2) node[above=-.05cm,midway] {$\tilde s_1$};
  \node[sfactor] (X2) [above=1.1cm of E2] {$p_X$};
  \draw (X2) -- (E2) node[right=-.05cm, midway] {$x_2$};
  \draw (E2.south) -- ([yshift=-.6cm]E2.south) node[right=-.05cm] {$y_2$};
  
  \node[factor,draw=none] (Edummy) [right=1cm of E2] {$\cdots$};

  \node[sfactor,draw=none] (Xdummy) [above=1.1cm of Edummy] {$\cdots$};
  
  \node[factor] (En) [right=1cm of Edummy] {};
  \draw (E2) -- (Edummy) node[above=-.05cm,midway] {$\tilde s_2$};

  \draw (Edummy) -- (En) 
  node[above=-.05cm,midway] {$\tilde s_{n-1}$};
  \node[sfactor] (Xn) [above=1.1cm of En] {$p_X$};
  \draw (Xn) -- (En) node[right=-.05cm, midway] {$x_n$};
  \draw (En.south) -- ([yshift=-.6cm]En.south) node[right=-.05cm] {$y_n$};
  
  \node[sfactor, inner sep=0pt] (ee) [right=.5cm of En] {$1$};
  \draw (En) -- (ee) node[above=-.05cm, midway] {$\tilde s_n$};

  \begin{pgfonlayer}{bg}

    \draw[dashed, blue, line width=1.5pt, fill=blue!10] 
                       ([xshift=-.8cm,yshift=0.45cm]X1)
                               rectangle 
                       ([xshift=.7cm,yshift=-0.45cm]Xn);
    \node[blue] at ([xshift=1.2cm,yshift=0cm]Xn) {$Q(\vx)$};

    \draw[dashed, red, line width=1.5pt, fill=red!10] 
                       ([xshift=-.7cm,yshift=0.55cm]S)
                               rectangle 
                       ([xshift=.6cm,yshift=-0.55cm]ee);
    \node[red] at ([xshift=.5cm,yshift=-.85cm]ee) {$W(\vy|\vx)$};

  \end{pgfonlayer}

  \end{tikzpicture}
  }
}

\newcommand{\figQFSMChighlevel}{
  \resizebox{\columnwidth}{!}{
  \begin{tikzpicture}[node/.style={draw=none},
      factor/.style={rectangle, minimum width=1cm, minimum height=1.7cm, 
                     draw},
      sfactor/.style={rectangle, minimum size=.4cm, draw}]
  \node[sfactor] (S) {$\rhoSt{0}$};
  \node[factor] (E1) [right=.7cm of S] {$W$};
  \draw (S.north) |- ([yshift=.6cm]E1.west) node[above=-0.05cm,pos=.75] 
                    {$s_0$};
  \draw (S.south) |- ([yshift=-.6cm]E1.west) node[below,pos=.75] 
                    {$s_0'$};
  \node[sfactor] (X1) [above=1.1cm of E1] {$p_X$};
  \draw (X1) -- (E1) node[right, midway] {$x_1$};
  \draw (E1.south) -- ([yshift=-.8cm]E1.south) node[right] {$y_1$};
  
  \node[factor] (E2) [right=1cm of E1] {$W$};
  \draw ([yshift=.6cm]E1.east) |- ([yshift=.6cm]E2.west) 
          node[above=-0.05cm,pos=.75] {$s_1$};
  \draw ([yshift=-.6cm]E1.east) |- ([yshift=-.6cm]E2.west) 
          node[below,pos=.75] {$s_1'$};
  \node[sfactor] (X2) [above=1.1cm of E2] {$p_X$};
  \draw (X2) -- (E2) node[right, midway] {$x_2$};
  \draw (E2.south) -- ([yshift=-.8cm]E2.south) node[right] {$y_2$};

  \node[factor, draw=none] (Edummy) [right=1cm of E2] {$\cdots$};

  \draw ([yshift=.6cm]E2.east) |- ([yshift=.6cm]Edummy.west) 
          node[above=-0.05cm,pos=.75] {$s_2$};
  \draw ([yshift=-.6cm]E2.east) |- ([yshift=-.6cm]Edummy.west) 
          node[below,pos=.75] {$s_2'$};

  \node[factor] (En) [right=1cm of Edummy] {$W$};
  \draw ([yshift=.6cm]Edummy.east) |- ([yshift=.6cm]En.west) 
          node[above=-0.05cm,pos=.75] {$s_{n-1}$};
  \draw ([yshift=-.6cm]Edummy.east) |- ([yshift=-.6cm]En.west) 
          node[below,pos=.75] {$s_{n-1}'$};
  \node[sfactor] (Xn) [above=1.1cm of En] {$p_X$};
  \draw (Xn) -- (En) node[right, midway] {$x_n$};
  \draw (En.south) -- ([yshift=-.8cm]En.south) node[right] {$y_n$};

  \node[factor, draw=none] (Edummy)  [right=1cm of E2] {$\cdots$};
  \node[sfactor, draw=none] (Edummy1) [above=-0.50cm of Edummy] {$\cdots$};
  \node[sfactor, draw=none] (Edummy2) [below=-0.45cm of Edummy] {$\cdots$};

  \node[sfactor,draw=none] (Xdummy) [above=1.1cm of Edummy] {$\cdots$};
  
  \node[sfactor, inner sep=0pt] (ee) [right=.5cm of En] {$=$};
  \draw ([yshift=.6cm]En.east) -| (ee.north) 
          node[above=-0.05cm,pos=.25] {$s_n$};
  \draw ([yshift=-.6cm]En.east) -| (ee.south) 
          node[below,pos=.25] {$s_n'$};

  \begin{pgfonlayer}{bg}

  \draw[dashed, blue, line width=1.5pt, fill=blue!10] 
                     ([xshift=-.8cm,yshift=0.45cm]X1)
                             rectangle 
                     ([xshift=.7cm,yshift=-0.45cm]Xn);
  \node[blue] at ([xshift=1.2cm,yshift=0cm]Xn) {$Q(\vx)$};

  \draw[dashed, red, line width=1.5pt, fill=red!10] 
                     ([xshift=-.7cm,yshift=1.15cm]S)
                             rectangle 
                     ([xshift=.6cm,yshift=-1.25cm]ee);
  \node[red] at ([xshift=.5cm,yshift=-1.55cm]ee) {$W(\vy|\vx)$};

  \end{pgfonlayer}

  \end{tikzpicture}
  }
}

\newcommand{\figQFSMCcondition}{
  \resizebox{0.80\columnwidth}{!}{
  \begin{tikzpicture}[node/.style={draw=none},
      factor/.style={rectangle, minimum width=1cm, minimum height=1.7cm, draw},
      sfactor/.style={rectangle, minimum size=.4cm, draw}]
  \node[factor] (W) {$W$};
  \node[above left=-.5cm and 1.2cm of W] (sp) {};
  \node[below left=-.5cm and 1.2cm of W] (sp') {};
  \node[above right=-.5cm and 1.2cm of W] (s) {};
  \node[below right=-.5cm and 1.2cm of W] (s') {};
  \node[above=.7cm of W] (X) {$x_{\ell}$};
  \node[below=.7cm of W] (Y) {$y_{\ell}$};
  \node[sfactor, inner sep=0pt, midway] (EW) [right=1.6cm of W] {$=$};

  \draw (sp.east) -- (sp.east-|W.west) node[above=-0.05cm,pos=0.25] {$s_{\ell-1}$};
  \draw (sp'.east) -- (sp'.east-|W.west) node[below,pos=0.25] {$s_{\ell-1}'$};
  \draw (sp.east-|W.east) -| (EW) node[above=-0.05cm,pos=0.25] {$s_{\ell}$};
  \draw (sp'.east-|W.east) -| (EW) node[below,pos=0.25] {$s_{\ell}'$};
  \draw (X) -- (W);
  \draw (Y) -- (W);
  
  \node[right=2.3cm of W, font=\huge] {$\hat{=}$};
  
  \node[sfactor] (E) [right=5cm of W] {$=$};
  \node[xshift=-1.5cm] at (E|-sp) (tsp) {};
  \node[xshift=-1.5cm] at (E|-sp') (tsp') {};
  \draw (tsp.east) -| (E) node[above=-0.05cm,pos=.1] {$s_{\ell-1}$};
  \draw (tsp'.east) -| (E) node[below,pos=.1] {$s_{\ell-1}'$};

  \begin{pgfonlayer}{bg}
  
  \draw[dashed, green, line width=1.5pt, fill=green!10]
                     ([xshift=-.2cm, yshift=.3cm]W.north west) rectangle 
                     ([xshift=1.65cm, yshift=-1.15cm]W.south east);

  \end{pgfonlayer}

  \end{tikzpicture}
  }
}

\newcommand{\figQGEchannel}{
  \resizebox{0.90\columnwidth}{!}{
  \begin{tikzpicture}[node/.style={draw=none},
      factor/.style={rectangle, minimum width=1cm, minimum height=.7cm, draw},
      sfactor/.style={rectangle, minimum size=.4cm, draw, inner sep=1pt}]
  \node[sfactor, minimum size=.7cm] (rhoX) {$\rhoAxell$};
  \node[factor] (E1) [above right=.4cm and 0.75cm of rhoX] {$E_{k_{\ell}}$};
  \node[factor] (E2) [below right=.4cm and 0.75cm of rhoX] {$E_{k_{\ell}}^{\Herm}$};
  \draw (E1) -- (E2) node[right=-0.05cm, midway] {$k_{\ell}$};
  \draw[*-] (rhoX.north) |- ([yshift=-.15cm]E1.west); 
  \draw[-*] (rhoX.south) |- ([yshift=.15cm]E2.west); 
  \node at ([xshift=-.5cm,yshift=2.3cm]rhoX) (X) {$x_{\ell}$};
  \draw (X) |- (rhoX);
  
  \draw ([xshift=-3.1cm,yshift=.15cm]E1.west) -- ([yshift=.15cm]E1.west)
          node[above=-0.05cm,pos=.1] {$s_{\ell-1}$};
  \draw[-*] ([xshift=-3.1cm,yshift=-.15cm]E2.west) -- ([yshift=-.15cm]E2.west)
          node[below,pos=.1] {$s_{\ell-1}'$};
  
  \node[sfactor, font=\scriptsize, minimum height=.4cm, minimum width=.6cm] (M1) 
    [right =.4cm of E1, yshift=-.15cm] {$M_{y_{\ell}}$};
  \draw[*-] (E1.east|-M1.west) -- (M1.west);
  
  \node[sfactor, font=\scriptsize, minimum height=.4cm, minimum width=.6cm] (M2) 
    [right =.4cm of E2, yshift=.15cm] {$M_{y_{\ell}}^{\Herm}$};
  \draw[-*] (E2.east|-M2.west) -- (M2.west);
  
  \node[factor] (U1) [right=2.5cm of E1] {$U$};
  \node[factor] (U2) [right=2.5cm of E2] {$U^{\Herm}$};
  \draw[*-] ([yshift=.15cm]E1.east) -- ([yshift=.15cm]U1.west) 
      node[above=-0.05cm,pos=.85] {$\bar{s}_{\ell}$};
  \draw[-*] ([yshift=-.15cm]E2.east) -- ([yshift=-.15cm]U2.west) 
      node[below,pos=.85] {$\bar{s}'_{\ell}$};
  \draw[*-] ([yshift=.15cm]U1.east) -- ([xshift=.9cm, yshift=.15cm]U1.east) 
      node[above=-0.05cm,pos=0.85] {$s_{\ell}$};
  \draw ([yshift=-.15cm]U2.east) -- ([xshift=.9cm, yshift=-.15cm]U2.east) 
      node[below,pos=0.85] {$s'_{\ell}$};
  
  \node[sfactor, inner sep=0pt] at (rhoX-|M1) (ey) {$=$};
  \draw (ey) -- (M1);
  \draw (ey) -- (M2);
  \node at ([xshift=.5cm,yshift=-2.3cm]ey) (Y) {$y_{\ell}$};
  \draw (ey) -| (Y);
  
  \node[sfactor, inner sep=0pt] (eey) [right=3.6cm of rhoX] {$=$};
  \draw[-*] (eey) |- (M1);
  \draw (eey) |- (M2);

  \begin{pgfonlayer}{bg}

  \draw[dashed, red, line width=1.5pt, fill=red!10]
                             ([xshift=-.8cm,yshift=1.8cm]rhoX)
                             rectangle 
                             ([xshift=.8cm,yshift=-.75cm]U2);
  \node[red] at ([xshift=-4.2cm,yshift=-1.1cm]U2) 
                {$W^{(y_{\ell}|x_{\ell})}(s_{\ell-1},s_{\ell}; s'_{\ell-1},s'_{\ell})$};

  \end{pgfonlayer}

  \end{tikzpicture}
  }
}

\newcommand{\figMemorylessQuantumChannel}{
  \begin{tikzpicture}[node/.style={draw=none},
      factor/.style={rectangle, minimum width=1cm, minimum height=.7cm, draw},
      sfactor/.style={rectangle, minimum size=.4cm, draw, inner sep=1pt}]
  \node[factor, minimum width=0.7cm, minimum height=0.7cm] (rhoX) {$\rhoAxell$};
  \node[factor] (E1) [above right=.4cm and 1cm of rhoX] {$E_{k_{\ell}}$};
  \node[factor] (E2) [below right=.4cm and 1cm of rhoX] {$E_{k_{\ell}}^{\Herm}$};
  \draw (E1) -- (E2) node[right=-0.05cm, midway] {$k_{\ell}$};
  \draw[*-] (rhoX.north) |- ([yshift=0cm]E1.west)
    node[above, pos=0.75] {$a_{\ell}$}; 
  \draw[-*] (rhoX.south) |- ([yshift=0cm]E2.west)
    node[below, pos=0.75] {$a'_{\ell}$}; 
  \node at ([xshift=-.75cm,yshift=2.3cm]rhoX) (X) {$x_{\ell}$};
  \draw (X) |- (rhoX);
  
  \node[factor, minimum width=1cm, minimum height=0.7cm] (M1) 
    [right =1.0cm of E1, yshift=0cm] {$M_{y_{\ell}}$};
  \draw[*-] (E1.east|-M1.west) -- (M1.west)
    node[above, pos=0.50] {$b_{\ell}$};
  
  \node[factor, minimum width=1cm, minimum height=0.7cm] (M2) 
    [right =1.0cm of E2, yshift=0cm] {$M_{y_{\ell}}^{\Herm}$};
  \draw[-*] (E2.east|-M2.west) -- (M2.west)
    node[below, pos=0.50] {$b'_{\ell}$};
  
  \node[sfactor, inner sep=0pt] at (rhoX-|M1) (ey) {$=$};
  \draw (ey) -- (M1);
  \draw (ey) -- (M2);
  \node at ([xshift=1.0cm,yshift=-2.3cm]ey) (Y) {$y_{\ell}$};
  \draw (ey) -| (Y);
  
  \node[sfactor, inner sep=0pt] (eey) [right=5.25cm of rhoX] {$=$};
  \draw[-*] (eey) |- (M1)
    node[above, pos=0.60] {$c_{\ell}$};
  \draw (eey) |- (M2)
    node[below, pos=0.60] {$c'_{\ell}$};

  \begin{pgfonlayer}{bg}

  \draw[dashed, red, line width=1.5pt, fill=red!10]
                             ([xshift=-5.15cm,yshift=1.75cm]ey)
                             rectangle 
                             ([xshift=2.65cm,yshift=-1.75cm]ey);

  \node[red] at ([xshift=-2.3cm,yshift=-3.2cm]E1) 
                {$W(y_{\ell}|x_{\ell})$};

  \end{pgfonlayer}

  \end{tikzpicture}
}

\newcommand{\figQuantumStateChannel}{
  \resizebox{0.90\columnwidth}{!}{
  \begin{tikzpicture}[node/.style={draw=none},
      factor/.style={rectangle, minimum width=1cm, minimum height=.7cm, draw},
      sfactor/.style={rectangle, minimum size=.4cm, draw, inner sep=1pt}]
  \node[sfactor, minimum size=.7cm] (rhoX) {$\rhoAxell$};
  \node[factor] (E1) [above right=.4cm and 1cm of rhoX] {$E_{k_{\ell}}$};
  \node[factor] (E2) [below right=.4cm and 1cm of rhoX] {$E_{k_{\ell}}^{\Herm}$};
  \draw (E1) -- (E2) node[right=-0.05cm, midway] {$k_{\ell}$};
  \draw[*-] (rhoX.north) |- ([yshift=-.15cm]E1.west); 
  \draw[-*] (rhoX.south) |- ([yshift=.15cm]E2.west); 
  \node at ([xshift=-.5cm,yshift=2.3cm]rhoX) (X) {$x_{\ell}$};
  \draw (X) |- (rhoX);
  
  \draw ([xshift=-3.50cm,yshift=.15cm]E1.west) -- ([yshift=.15cm]E1.west)
          node[above=-0.05cm,pos=.1] {$s_{\ell-1}$};
  \draw[-*] ([xshift=-3.50cm,yshift=-.15cm]E2.west) -- ([yshift=-.15cm]E2.west)
          node[below,pos=.1] {$s_{\ell-1}'$};
  
  \node[sfactor, font=\scriptsize, minimum height=.4cm, minimum width=.6cm] (M1) 
    [right =.4cm of E1, yshift=-.15cm] {$M_{y_{\ell}}$};
  \draw[*-] (E1.east|-M1.west) -- (M1.west);
  
  \node[sfactor, font=\scriptsize, minimum height=.4cm, minimum width=.6cm] (M2) 
    [right =.4cm of E2, yshift=.15cm] {$M_{y_{\ell}}^{\Herm}$};
  \draw[-*] (E2.east|-M2.west) -- (M2.west);
  
  \draw[*-] ([yshift=.15cm]E1.east) -- ([xshift=3.5cm, yshift=.15cm]E1.east) 
      node[above=-0.05cm,pos=.95] {$s_{\ell}$};
  \draw[-] ([yshift=-.15cm]E2.east) -- ([xshift=3.5cm, yshift=-.15cm]E2.east) 
      node[below,pos=.95] {$s'_{\ell}$};
  
  \node[sfactor, inner sep=0pt] at (rhoX-|M1) (ey) {$=$};
  \draw (ey) -- (M1);
  \draw (ey) -- (M2);
  \node at ([xshift=.5cm,yshift=-2.3cm]ey) (Y) {$y_{\ell}$};
  \draw (ey) -| (Y);
  
  \node[sfactor, inner sep=0pt] (eey) [right=3.6cm of rhoX] {$=$};
  \draw[-*] (eey) |- (M1);
  \draw (eey) |- (M2);

  \begin{pgfonlayer}{bg}

  \draw[dashed, red, line width=1.5pt, fill=red!10]
                             ([xshift=-1.0cm,yshift=1.7cm]rhoX)
                             rectangle 
                             ([xshift=-0.5cm,yshift=-.75cm]U2);
  \node[red] at ([xshift=-4.2cm,yshift=-1.1cm]U2) 
                {$W^{(y_{\ell}|x_{\ell})}(s_{\ell-1},s_{\ell}; s'_{\ell-1},s'_{\ell})$};

  \end{pgfonlayer}

  \end{tikzpicture}
  }
}

\newcommand{\figScalingSimulation}{115pt}

\newcommand{\figSimulationPlotOne}{
  \resizebox{!}{\figScalingSimulation}{
  \begin{tikzpicture}[every axis/.append style={font=\footnotesize},
  					   every mark/.append style={scale=1.5}]
    \begin{axis}[xlabel={$\pbad$}, 
                     ylabel={bits per channel use},
                     ylabel shift=-4 pt,
                     xmin=0, xmax=1, ymin=0,  ymax=1,
                     legend style={font=\scriptsize}, 
                     mark repeat={1}, 
                     width=300pt, height=200pt,
                     y tick label style={
                         /pgf/number format/fixed,
                         /pgf/number format/fixed zerofill,
                         /pgf/number format/precision=1},
                     ytick={0.1,0.2,0.3,0.4,0.5,0.6,0.7,0.8,0.9},
                     extra y ticks={0.0,1.0},
                     extra y tick style={grid=none},
                     xtick={0.2,0.4,0.6,0.8},
                     extra x ticks={0.0,1.0},
                     extra x tick style={grid=none},
                     x tick label style={
                         /pgf/number format/fixed,
                         /pgf/number format/fixed zerofill,
                         /pgf/number format/precision=1},
                     grid=both,
                     grid style={line width=0.1pt, draw=black, dashed},
                     legend cell align={left}]
    \addplot[solid,mark=o, red] table[x=p_b,y=Cap] {GEs1aux24_alpha1.data};
    \addlegendentry{IR estimate}
    \addplot[solid,mark=+, blue] table[x=p_b,y=LU_4] {GEs1aux24_alpha1.data};
    \addlegendentry{IR lower bound estimate after $50$ updates of a $4$-state auxiliary channel}
    \addplot[solid,mark=diamond, black] table[x=p_b,y=LU_2] {GEs1aux24_alpha1.data};
    \addlegendentry{IR lower bound estimate after $50$ updates of a $2$-state auxiliary channel}
    \end{axis}
  \end{tikzpicture}
  }
}

\newcommand{\figSimulationPlotTwo}{
  \resizebox{!}{\figScalingSimulation}{
  \begin{tikzpicture}[every axis/.append style={font=\footnotesize},
  					   every mark/.append style={scale=1.5}]
    \begin{axis}[xlabel={$\pbad$}, 
                     ylabel={bits per channel use},
                     ylabel shift=-4 pt,
                     xmin=0, xmax=1, ymin=0,  ymax=1,
                     legend style={font=\scriptsize}, 
                     mark repeat={1}, 
                     width=300pt, height=200pt,
                     y tick label style={
                         /pgf/number format/fixed,
                         /pgf/number format/fixed zerofill,
                         /pgf/number format/precision=1},
                     ytick={0.1,0.2,0.3,0.4,0.5,0.6,0.7,0.8,0.9},
                     extra y ticks={0.0,1.0},
                     extra y tick style={grid=none},
                     xtick={0.2,0.4,0.6,0.8},
                     extra x ticks={0.0,1.0},
                     extra x tick style={grid=none},
                     x tick label style={
                         /pgf/number format/fixed,
                         /pgf/number format/fixed zerofill,
                         /pgf/number format/precision=1},
                     grid=both,
                     grid style={line width=0.1pt, draw=black, dashed},
                     legend cell align={left}]
    \addplot[solid,mark=o, red] table[x=p_b,y=Cap] {GEs2aux24QR.data};
    \addlegendentry{IR estimate}
    \addplot[solid,mark=+, blue] table[x=p_b,y=LU_4] {GEs2aux24QR.data};
    \addlegendentry{IR lower bound estimate after $50$ updates of a $4$-state auxiliary channel}
    \addplot[solid,mark=diamond, black] table[x=p_b,y=LU_2] {GEs2aux24QR.data};
    \addlegendentry{IR lower bound estimate after $50$ updates of a $2$-state auxiliary channel}
    \addplot[solid,mark=triangle, black!60!green] table[x=p_b,y=QL] {GEs2aux24QR.data};
	\addlegendentry{IR lower bound estimate with $1$-qubit-state auxiliary channel}
    \end{axis}
  \end{tikzpicture}
  }
}

\newcommand{\figSimulationPlotThree}{
  \resizebox{!}{\figScalingSimulation}{
  \begin{tikzpicture}[every axis/.append style={font=\footnotesize},
                      every mark/.append style={scale=1.5}]
    \begin{axis}[xlabel={$\alpha$}, 
                     ylabel={bits per channel use},
                     ylabel shift=-4 pt,
                     xmin=-1.5, xmax=1.5, ymin=0,  ymax=0.9,
                     legend style={font=\scriptsize}, 
                     mark repeat={1}, 
                     width=300pt, height=200pt,
                     y tick label style={
                         /pgf/number format/fixed,
                         /pgf/number format/fixed zerofill,
                         /pgf/number format/precision=1},
                     ytick={0.1,0.2,0.3,0.4,0.5,0.6,0.7,0.8},
                     extra y ticks={0,0.9},
                     extra y tick style={grid=none},
                     xtick={-1.0,-0.5,0.0,0.5,1.0},
                     extra x ticks={-1.5,1.5},
                     extra x tick style={grid=none},
                     x tick label style={
                         /pgf/number format/fixed,
                         /pgf/number format/fixed zerofill,
                         /pgf/number format/precision=1},
                     grid=both,
                     grid style={line width=0.1pt, draw=black, dashed},
                     legend cell align={left}]
    \addplot[solid,mark=o, red] table[x=alphaA,y=CapA] {GEs1aux24ALPHA.data};
    \addlegendentry{IR estimate}
    \addplot[solid,mark=+, blue] table[x=alphaA,y=LU_4A] {GEs1aux24ALPHA.data};
    \addlegendentry{IR lower bound estimate after $50$ updates of a $4$-state the auxiliary channel}
    \addplot[solid,mark=diamond, black] table[x=alphaA,y=LU_2A] {GEs1aux24ALPHA.data};
    \addlegendentry{IR lower bound estimate after $50$ updates of a $2$-state auxiliary channel}
    \addplot[solid,mark=o, red] table[x=alphaB,y=CapB] {GEs1aux24ALPHA.data};
    \addplot[solid,mark=+, blue] table[x=alphaB,y=LU_4B] {GEs1aux24ALPHA.data};
    \addplot[solid,mark=diamond, black] table[x=alphaB,y=LU_2B] {GEs1aux24ALPHA.data};
    \end{axis}
  \end{tikzpicture}
  }
}

\newcommand{\figSimulationPlotFour}{
  \resizebox{!}{\figScalingSimulation}{
  \begin{tikzpicture}[every axis/.append style={font=\footnotesize},
  					   every mark/.append style={scale=1.5}]
    \begin{axis}[xlabel={$\alpha$}, 
                     ylabel={bits per channel use},
                     ylabel shift=-4 pt,
                     xmin=-1.5, xmax=1.5, ymin=0,  ymax=0.9,
                     legend style={font=\scriptsize}, 
                     mark repeat={1}, 
                     width=300pt, height=200pt,
                     y tick label style={
                         /pgf/number format/fixed,
                         /pgf/number format/fixed zerofill,
                         /pgf/number format/precision=1},
                     ytick={0.1,0.2,0.3,0.4,0.5,0.6,0.7,0.8},
                     extra y ticks={0,0.9},
                     extra y tick style={grid=none},
                     xtick={-1.0,-0.5,0.0,0.5,1.0},
                     extra x ticks={-1.5,1.5},
                     extra x tick style={grid=none},
                     x tick label style={
                         /pgf/number format/fixed,
                         /pgf/number format/fixed zerofill,
                         /pgf/number format/precision=1},
                     grid=both,
                     grid style={line width=0.1pt, draw=black, dashed},
                     legend cell align={left}]
    \addplot[solid,mark=o, red] table[x=alphaA,y=CapA] {GEs2aux24ALPHA.data};
    \addlegendentry{IR estimate}
    \addplot[solid,mark=+, blue] table[x=alphaA,y=LU_4A] {GEs2aux24ALPHA.data};
    \addlegendentry{IR lower bound estimate after $50$ updates of a $4$-state auxiliary channel}
    \addplot[solid,mark=diamond, black] table[x=alphaA,y=LU_2A] {GEs2aux24ALPHA.data};
    \addlegendentry{IR lower bound estimate after $50$ updates of a $2$-state auxiliary channel}
    \addplot[solid,mark=o, red] table[x=alphaB,y=CapB] {GEs2aux24ALPHA.data};
    \addplot[solid,mark=+, blue] table[x=alphaB,y=LU_4B] {GEs2aux24ALPHA.data};
    \addplot[solid,mark=diamond, black] table[x=alphaB,y=LU_2B] {GEs2aux24ALPHA.data};
    \end{axis}
  \end{tikzpicture}
  }
}

\newcommand{\figQFSMestimatehY}{
\resizebox{1.00\linewidth}{!}{
\begin{tikzpicture}[node/.style={draw=none},
	factor/.style={rectangle, minimum width=1cm, minimum height=1.7cm, draw},
	sfactor/.style={rectangle, minimum size=.4cm, draw},
	darksolid/.style={rectangle, minimum size=.15cm, draw,fill = black, inner sep=0pt, outer sep = 0pt},
	label/.style={magenta,anchor=north east,xshift = .1cm}]
	\node[sfactor] (S) {$\rhoSt{0}$};
	
	\node[factor] (E1) [right=.7cm of S] {$W$};
	\draw (S.north) |- ([yshift=.6cm]E1.west) node[above=-0.05cm,pos=.75] {$s_0$};
	\draw (S.south) |- ([yshift=-.6cm]E1.west) node[below,pos=.75] {$s_0'$};
	\node[sfactor] (X1) [above=.7cm of E1] {$p_X$};
	\draw (X1) -- (E1) node[right, midway] {$x_1$};
	\draw (E1.south) -- ([yshift=-.7cm]E1.south) node[darksolid]{} node (Y1) [right] {$\cy_1$};
	
	\node[factor] (E2) [right=1cm of E1] {$W$};
	\draw ([yshift=.6cm]E1.east) |- ([yshift=.6cm]E2.west) node[above=-0.05cm,pos=.75] {$s_1$};
	\draw ([yshift=-.6cm]E1.east) |- ([yshift=-.6cm]E2.west) node[below,pos=.75] {$s_1'$};
	\node[sfactor] (X2) [above=.7cm of E2] {$p_X$};
	\draw (X2) -- (E2) node[right, midway] {$x_2$};
	\draw (E2.south) -- ([yshift=-.7cm]E2.south) node[darksolid]{} node[right] {$\cy_2$};

	\node[factor, draw=none] (Edummy1) [right=1cm of E2] {};
	\node at ([yshift=.6cm]E2.east-|Edummy1) {$\cdots$};
	\node at ([yshift=-.6cm]E2.east-|Edummy1) {$\cdots$};
	\draw ([yshift=.6cm]E2.east) |- ([yshift=.6cm]Edummy1.west) 
		node[above=-0.05cm,pos=.75] {$s_2$};
	\draw ([yshift=-.6cm]E2.east) |- ([yshift=-.6cm]Edummy1.west) node[below,pos=.75] {$s_2'$};
	\node at (X1-|Edummy1) {$\cdots$};
	\node at (Y1-|Edummy1) {$\cdots$};
	
	\node[factor] (El) [right=1cm of Edummy1] {$W$};
	\draw ([yshift=.6cm]Edummy1.east) |- ([yshift=.6cm]El.west) 
		node[above=-0.05cm,pos=.75] {$s_{\ell-1}$};
	\draw ([yshift=-.6cm]Edummy1.east) |- ([yshift=-.6cm]El.west)  node[below,pos=.75] {$s_{\ell-1}'$};
	\node[sfactor] (Xl) [above=.7cm of El] {$p_X$};
	\draw (Xl) -- (El) node[right, midway] {$x_\ell$};
	\draw (El.south) -- ([yshift=-.7cm]El.south) node[darksolid]{} node[right] {$\cy_\ell$};
	
	\node[factor] (El2) [right=1cm of El] {$W$};
	\draw ([yshift=.6cm]El.east) |- ([yshift=.6cm]El2.west) node[above=-0.05cm,pos=.75] {$s_\ell$};
	\draw ([yshift=-.6cm]El.east) |- ([yshift=-.6cm]El2.west) node[below,pos=.75] {$s_\ell'$};
	\node[sfactor] (Xl2) [above=.7cm of El2] {$p_X$};
	\draw (Xl2) -- (El2) node[right, midway] {$x_{\ell\!+\!1}$};
	\draw (El2.south) -- ([yshift=-.7cm]El2.south) node[darksolid]{} node[right] {$\cy_{\ell\!+\!1}$};
	
	\node[factor, draw=none] (Edummy2) [right=1cm of El2] {};
	\node at ([yshift=.6cm]El2.east-|Edummy2) {$\cdots$};
	\node at ([yshift=-.6cm]El2.east-|Edummy2) {$\cdots$};
	\draw ([yshift=.6cm]El2.east) |- ([yshift=.6cm]Edummy2.west) 
		node[above=-0.05cm,pos=.85] {$s_{\ell\!+\!1}$};
	\draw ([yshift=-.6cm]El2.east) |- ([yshift=-.6cm]Edummy2.west) 
		node[below,pos=.85] {$s_{\ell\!+\!1}'$};
	\node at (X1-|Edummy2) {$\cdots$};
	\node at (Y1-|Edummy2) {$\cdots$};
	
	\node[factor] (En) [right=1cm of Edummy2] {$W$};
	\draw ([yshift=.6cm]Edummy2.east) |- ([yshift=.6cm]En.west) 
		node[above=-0.05cm,pos=.75] {$s_{n\!-\!1}$};
	\draw ([yshift=-.6cm]Edummy2.east) |- ([yshift=-.6cm]En.west)  
		node[below,pos=.75] {$s_{n\!-\!1}'$};
	\node[sfactor] (Xn) [above=.7cm of En] {$p_X$};
	\draw (Xn) -- (En) node[right, midway] {$x_n$};
	\draw (En.south) -- ([yshift=-.7cm]En.south) node[darksolid]{} node[right] {$\cy_n$};
	
	\node[sfactor, inner sep=0pt] (ee) [right=1cm of En] {$=$};
	\draw ([yshift=.6cm]En.east) -| (ee.north) node[above=-0.05cm,pos=.25] {$s_n$};
	\draw ([yshift=-.6cm]En.east) -| (ee.south) node[below,pos=.25] {$s_n'$};

	\begin{pgfonlayer}{bg}
	\draw[dashed, magenta, line width=1.5pt, fill=magenta!8] ([xshift=-1.4cm,yshift=2.9cm]S)
		rectangle ([xshift=.7cm,yshift=-4.2cm]En);
	\node[label] at ([xshift=.7cm,yshift=-4.2cm]En) {$\sigmaY_n$};
	\draw[dashed, magenta, line width=1.5pt, fill=magenta!16] ([xshift=-1.2cm,yshift=2.7cm]S)
		rectangle ([xshift=.8cm,yshift=-3.6cm]El2);
	\node[label] at ([xshift=.7cm,yshift=-3.6cm]El2) {$\sigmaY_{\ell+1}$};
	\draw[dashed, magenta, line width=1.5pt, fill=magenta!24] ([xshift=-1cm,yshift=2.5cm]S)
		rectangle ([xshift=.7cm,yshift=-3cm]El);
	\node[label] at ([xshift=.7cm,yshift=-3cm]El) {$\sigmaY_{\ell}$};
	\draw[dashed, magenta, line width=1.5pt, fill=magenta!32] ([xshift=-.8cm,yshift=2.3cm]S)
		rectangle ([xshift=.7cm,yshift=-2.4cm]E2);
	\node[label] at ([xshift=.7cm,yshift=-2.4cm]E2) {$\sigmaY_2$};
	\draw[dashed, magenta, line width=1.5pt, fill=magenta!40] ([xshift=-.6cm,yshift=2.1cm]S)
		rectangle ([xshift=.7cm,yshift=-1.8cm]E1);
	\node[label] at ([xshift=.7cm,yshift=-1.8cm]E1) {$\sigmaY_1$};
	\end{pgfonlayer}
\end{tikzpicture}
}
}

\newcommand{\figQFSMestimatehXY}{
\resizebox{1.00\linewidth}{!}{
\begin{tikzpicture}[node/.style={draw=none},
	factor/.style={rectangle, minimum width=1cm, minimum height=1.7cm, draw},
	sfactor/.style={rectangle, minimum size=.4cm, draw},
	darksolid/.style={rectangle, minimum size=.15cm, draw,fill = black, inner sep=0pt, outer sep = 0pt},
	label/.style={magenta,anchor=north east,xshift = .1cm}]
	\node[sfactor] (S) {$\rhoSt{0}$};
	
	\node[factor] (E1) [right=.7cm of S] {$W$};
	\draw (S.north) |- ([yshift=.6cm]E1.west) node[above=-0.05cm,pos=.75] {$s_0$};
	\draw (S.south) |- ([yshift=-.6cm]E1.west) node[below,pos=.75] {$s_0'$};
	\node[darksolid] (X1) [above=.3cm of E1.north,anchor = south] {};
	\draw (X1) -- (E1) node[right, pos = 0] {$\cx_1$};
	\draw (X1) -- ([yshift=.3cm]X1.north) node (pX1) [sfactor, anchor = south, pos=1] {$p_X$};
	\draw (E1.south) -- ([yshift=-.7cm]E1.south) node[darksolid]{} node (Y1) [right] {$\cy_1$};
	
	\node[factor] (E2) [right=1cm of E1] {$W$};
	\draw ([yshift=.6cm]E1.east) |- ([yshift=.6cm]E2.west) node[above=-0.05cm,pos=.75] {$s_1$};
	\draw ([yshift=-.6cm]E1.east) |- ([yshift=-.6cm]E2.west) node[below,pos=.75] {$s_1'$};
	\node[darksolid] (X2) [above=.3cm of E2.north,anchor = south] {};
	\draw (X2) -- (E2) node[right, pos = 0] {$\cx_2$};
	\draw (X2) -- ([yshift=.3cm]X2.north) node[sfactor, anchor = south, pos=1] {$p_X$};
	\draw (E2.south) -- ([yshift=-.7cm]E2.south) node[darksolid]{} node[right] {$\cy_2$};

	\node[factor, draw=none] (Edummy1) [right=1cm of E2] {};
	\node at ([yshift=.6cm]E2.east-|Edummy1) {$\cdots$};
	\node at ([yshift=-.6cm]E2.east-|Edummy1) {$\cdots$};
	\draw ([yshift=.6cm]E2.east) |- ([yshift=.6cm]Edummy1.west) 
		node[above=-0.05cm,pos=.75] {$s_2$};
	\draw ([yshift=-.6cm]E2.east) |- ([yshift=-.6cm]Edummy1.west) node[below,pos=.75] {$s_2'$};
	\node at (pX1-|Edummy1) {$\cdots$};
	\node at (Y1-|Edummy1) {$\cdots$};
	
	\node[factor] (El) [right=1cm of Edummy1] {$W$};
	\draw ([yshift=.6cm]Edummy1.east) |- ([yshift=.6cm]El.west) 
		node[above=-0.05cm,pos=.75] {$s_{\ell-1}$};
	\draw ([yshift=-.6cm]Edummy1.east) |- ([yshift=-.6cm]El.west)  node[below,pos=.75] {$s_{\ell-1}'$};
	\node[darksolid] (Xl) [above=.3cm of El.north,anchor = south] {};
	\draw (Xl) -- (El) node[right, pos = 0] {$\cx_\ell$};
	\draw (Xl) -- ([yshift=.3cm]Xl.north) node[sfactor, anchor = south, pos=1] {$p_X$};
	\draw (El.south) -- ([yshift=-.7cm]El.south) node[darksolid]{} node[right] {$\cy_\ell$};
	
	\node[factor] (El2) [right=1cm of El] {$W$};
	\draw ([yshift=.6cm]El.east) |- ([yshift=.6cm]El2.west) node[above=-0.05cm,pos=.75] {$s_\ell$};
	\draw ([yshift=-.6cm]El.east) |- ([yshift=-.6cm]El2.west) node[below,pos=.75] {$s_\ell'$};
	\node[darksolid] (Xl2) [above=.3cm of El2.north,anchor = south] {};
	\draw (Xl2) -- (El2) node[right, pos = 0] {$\cx_{\ell\!+\!1}$};
	\draw (Xl2) -- ([yshift=.3cm]Xl2.north) node[sfactor, anchor = south, pos=1] {$p_X$};
	\draw (El2.south) -- ([yshift=-.7cm]El2.south) node[darksolid]{} node[right] {$\cy_{\ell\!+\!1}$};
	
	\node[factor, draw=none] (Edummy2) [right=1cm of El2] {};
	\node at ([yshift=.6cm]El2.east-|Edummy2) {$\cdots$};
	\node at ([yshift=-.6cm]El2.east-|Edummy2) {$\cdots$};
	\draw ([yshift=.6cm]El2.east) |- ([yshift=.6cm]Edummy2.west) 
		node[above=-0.05cm,pos=.85] {$s_{\ell\!+\!1}$};
	\draw ([yshift=-.6cm]El2.east) |- ([yshift=-.6cm]Edummy2.west) 
		node[below,pos=.85] {$s_{\!\ell+\!1}'$};
	\node at (pX1-|Edummy2) {$\cdots$};
	\node at (Y1-|Edummy2) {$\cdots$};
	
	\node[factor] (En) [right=1cm of Edummy2] {$W$};
	\draw ([yshift=.6cm]Edummy2.east) |- ([yshift=.6cm]En.west) 
		node[above=-0.05cm,pos=.75] {$s_{n\!-\!1}$};
	\draw ([yshift=-.6cm]Edummy2.east) |- ([yshift=-.6cm]En.west)
		node[below,pos=.75] {$s_{n\!-\!1}'$};
	\node[darksolid] (Xn) [above=.3cm of En.north,anchor = south] {};
	\draw (Xn) -- (En) node[right, pos = 0] {$\cx_n$};
	\draw (Xn) -- ([yshift=.3cm]Xn.north) node[sfactor, anchor = south, pos=1] {$p_X$};
	\draw (En.south) -- ([yshift=-.7cm]En.south) node[darksolid]{} node[right] {$\cy_n$};
	
	\node[sfactor, inner sep=0pt] (ee) [right=1cm of En] {$=$};
	\draw ([yshift=.6cm]En.east) -| (ee.north) node[above=-0.05cm,pos=.25] {$s_n$};
	\draw ([yshift=-.6cm]En.east) -| (ee.south) node[below,pos=.25] {$s_n'$};

	\begin{pgfonlayer}{bg}
	\draw[dashed, magenta, line width=1.5pt, fill=magenta!8] ([xshift=-1.4cm,yshift=3cm]S)
		rectangle ([xshift=.7cm,yshift=-4.2cm]En);
	\node[label] at ([xshift=.7cm,yshift=-4.2cm]En) {$\sigmaXY_n$};
	\draw[dashed, magenta, line width=1.5pt, fill=magenta!16] ([xshift=-1.2cm,yshift=2.8cm]S)
		rectangle ([xshift=.8cm,yshift=-3.6cm]El2);
	\node[label] at ([xshift=.8cm,yshift=-3.6cm]El2) {$\sigmaXY_{\ell+1}$};
	\draw[dashed, magenta, line width=1.5pt, fill=magenta!24] ([xshift=-1cm,yshift=2.6cm]S)
		rectangle ([xshift=.7cm,yshift=-3cm]El);
	\node[label] at ([xshift=.7cm,yshift=-3cm]El) {$\sigmaXY_{\ell}$};
	\draw[dashed, magenta, line width=1.5pt, fill=magenta!32] ([xshift=-.8cm,yshift=2.4cm]S)
		rectangle ([xshift=.7cm,yshift=-2.4cm]E2);
	\node[label] at ([xshift=.7cm,yshift=-2.4cm]E2) {$\sigmaXY_2$};
	\draw[dashed, magenta, line width=1.5pt, fill=magenta!40] ([xshift=-.6cm,yshift=2.2cm]S)
		rectangle ([xshift=.7cm,yshift=-1.8cm]E1);
	\node[label] at ([xshift=.7cm,yshift=-1.8cm]E1) {$\sigmaXY_1$};
	\end{pgfonlayer}
\end{tikzpicture}
}
}

\newcommand{\figQFSMclosingthebox}{
\resizebox{1.00\linewidth}{!}{
\begin{tikzpicture}[node/.style={draw=none},
	factor/.style={rectangle, minimum width=1cm, minimum height=1.7cm, draw},
	sfactor/.style={rectangle, minimum size=.4cm, draw},
	darksolid/.style={rectangle, minimum size=.15cm, draw,fill = black, inner sep=0pt, outer sep = 0pt},
	label/.style={magenta,anchor=north east,xshift = .1cm}]
	\node[sfactor] (S) {$\rhoSt{0}$};
	
	\node[factor] (E1) [right=.7cm of S] {$W$};
	\draw (S.north) |- ([yshift=.6cm]E1.west) node[above=-0.05cm,pos=.75] {$s_0$};
	\draw (S.south) |- ([yshift=-.6cm]E1.west) node[below,pos=.75] {$s_0'$};
	\node[sfactor] (X1) [above=.7cm of E1] {$p_X$};
	\draw (X1) -- (E1) node[right, midway] {$x_1$};
	\draw (E1.south) -- ([yshift=-.7cm]E1.south) node (Y1) [right] {$y_1$};
	
	\node[factor] (E2) [right=1cm of E1] {$W$};
	\draw ([yshift=.6cm]E1.east) |- ([yshift=.6cm]E2.west) node[above=-0.05cm,pos=.75] {$s_1$};
	\draw ([yshift=-.6cm]E1.east) |- ([yshift=-.6cm]E2.west) node[below,pos=.75] {$s_1'$};
	\node[sfactor] (X2) [above=.7cm of E2] {$p_X$};
	\draw (X2) -- (E2) node[right, midway] {$x_2$};
	\draw (E2.south) -- ([yshift=-.7cm]E2.south) node[right] {$y_2$};

	\node[factor, draw=none] (Edummy1) [right=1cm of E2] {};
	\node at ([yshift=.6cm]E2.east-|Edummy1) {$\cdots$};
	\node at ([yshift=-.6cm]E2.east-|Edummy1) {$\cdots$};
	\draw ([yshift=.6cm]E2.east) |- ([yshift=.6cm]Edummy1.west) 
		node[above=-0.05cm,pos=.75] {$s_2$};
	\draw ([yshift=-.6cm]E2.east) |- ([yshift=-.6cm]Edummy1.west) node[below,pos=.75] {$s_2'$};
	\node at (X1-|Edummy1) {$\cdots$};
	\node at (Y1-|Edummy1) {$\cdots$};
	
	\node[factor] (El) [right=1cm of Edummy1] {$W$};
	\draw ([yshift=.6cm]Edummy1.east) |- ([yshift=.6cm]El.west) 
		node[above=-0.05cm,pos=.75] {$s_{\ell-1}$};
	\draw ([yshift=-.6cm]Edummy1.east) |- ([yshift=-.6cm]El.west)  node[below,pos=.75] {$s_{\ell-1}'$};
	\node[sfactor] (Xl) [above=.7cm of El] {$p_X$};
	\draw (Xl) -- (El) node[right, midway] {$x_\ell$};
	\draw (El.south) -- ([yshift=-.7cm]El.south) node[right] {$y_\ell$};
	
	\node[factor] (El2) [right=1cm of El] {$W$};
	\draw ([yshift=.6cm]El.east) |- ([yshift=.6cm]El2.west) node[above=-0.05cm,pos=.75] {$s_\ell$};
	\draw ([yshift=-.6cm]El.east) |- ([yshift=-.6cm]El2.west) node[below,pos=.75] {$s_\ell'$};
	\node[sfactor] (Xl2) [above=.7cm of El2] {$p_X$};
	\draw (Xl2) -- (El2) node[right, midway] {$x_{\ell\!+\!1}$};
	\draw (El2.south) -- ([yshift=-.7cm]El2.south) node[right] {$y_{\ell\!+\!1}$};
	
	\node[factor, draw=none] (Edummy2) [right=1cm of El2] {};
	\node at ([yshift=.6cm]El2.east-|Edummy2) {$\cdots$};
	\node at ([yshift=-.6cm]El2.east-|Edummy2) {$\cdots$};
	\draw ([yshift=.6cm]El2.east) |- ([yshift=.6cm]Edummy2.west) 
		node[above=-0.05cm,pos=.75] {$s_{\ell\!+\!1}$};
	\draw ([yshift=-.6cm]El2.east) |- ([yshift=-.6cm]Edummy2.west) 
		node[below,pos=.75] {$s_{\ell\!+\!1}'$};
	\node at (X1-|Edummy2) {$\cdots$};
	\node at (Y1-|Edummy2) {$\cdots$};
	
	\node[factor] (En) [right=1cm of Edummy2] {$W$};
	\draw ([yshift=.6cm]Edummy2.east) |- ([yshift=.6cm]En.west) 
		node[above=-0.05cm,pos=.75] {$s_{n\!-\!1}$};
	\draw ([yshift=-.6cm]Edummy2.east) |- ([yshift=-.6cm]En.west)  
		node[below,pos=.75] {$s_{n\!-\!1}'$};
	\node[sfactor] (Xn) [above=.7cm of En] {$p_X$};
	\draw (Xn) -- (En) node[right, midway] {$x_n$};
	\draw (En.south) -- ([yshift=-.7cm]En.south) node[right] {$y_n$};
	
	\node[sfactor, inner sep=0pt] (ee) [right=.5cm of En] {$=$};
	\draw ([yshift=.6cm]En.east) -| (ee.north) node[above=-0.05cm,pos=.25] {$s_n$};
	\draw ([yshift=-.6cm]En.east) -| (ee.south) node[below,pos=.25] {$s_n'$};

	\begin{pgfonlayer}{bg}
	\draw[dashed, green, line width=1.5pt, fill=green!10] ([xshift=-2.2cm,yshift=3.2cm]E1)
		rectangle ([xshift=1.3cm,yshift=-2.8cm]ee);
	\draw[dashed, green, line width=1.5pt, fill=green!20] ([xshift=-.7cm,yshift=3cm]E1)
		rectangle ([xshift=1.1cm,yshift=-2.6cm]ee);
	\draw[dashed, green, line width=1.5pt, fill=green!30] ([xshift=-.7cm,yshift=2.8cm]E2)
		rectangle ([xshift=.9cm,yshift=-2.4cm]ee);
	\draw[dashed, green, line width=1.5pt, fill=green!40] ([xshift=-.7cm,yshift=2.6cm]El)
		rectangle ([xshift=.7cm,yshift=-2.2cm]ee);
	\draw[dashed, green, line width=1.5pt, fill=green!50] ([xshift=-.7cm,yshift=2.4cm]El2)
		rectangle ([xshift=.5cm,yshift=-2cm]ee);
	\draw[dashed, green, line width=1.5pt, fill=green!60] ([xshift=-.7cm,yshift=2.2cm]En)
		rectangle ([xshift=.3cm,yshift=-1.8cm]ee);
	\end{pgfonlayer}
\end{tikzpicture}
}
}

\newcommand{\figCFSMestimatehY}{
\begin{tikzpicture}[node/.style={draw=none},
	factor/.style={rectangle, minimum width=1cm, minimum height=.7cm, draw},
	sfactor/.style={rectangle, minimum size=.4cm, draw},
	darksolid/.style={rectangle, minimum size=.15cm, draw,fill = black, inner sep=0pt, outer sep = 0pt},
	label/.style={magenta,anchor=north east,xshift = .1cm}]
	\node[sfactor] (S) {$p_{\tilde S_0}$};
	
	\node[factor] (E1) [right=.7cm of S] {$W$};
	\draw (S) -- (E1) node[above,midway] {$\tilde{s}_0$};
	\node[sfactor] (X1) [above=.8cm of E1] {$p_X$};
	\draw (X1) -- (E1) node[right, midway] {$x_1$};
	\draw (E1.south) -- ([yshift=-.8cm]E1.south) node (Y1) [darksolid]{} node[right] {$\cy_1$};
	
	\node[factor] (E2) [right=1cm of E1] {$W$};
	\draw (E1) -- (E2) node[above,midway] {$\tilde{s}_1$};
	\node[sfactor] (X2) [above=.8cm of E2] {$p_X$};
	\draw (X2) -- (E2) node[right, midway] {$x_2$};
	\draw (E2.south) -- ([yshift=-.8cm]E2.south) node[darksolid]{} node[right] {$\cy_2$};

	\node[factor, draw=none] (Edummy1) [right=1cm of E2] {$\cdots$};
	\draw (E2) -- (Edummy1) node[above,midway] {$\tilde{s}_2$};
	\node at (X1-|Edummy1) {$\cdots$};
	\node at (Y1-|Edummy1) {$\cdots$};
	
	\node[factor] (El) [right=1cm of Edummy1] {$W$};
	\draw (Edummy1) -- (El) node[above,midway] {$\tilde{s}_{\ell-1}$};
	\node[sfactor] (Xl) [above=.8cm of El] {$p_X$};
	\draw (Xl) -- (El) node[right, midway] {$x_\ell$};
	\draw (El.south) -- ([yshift=-.8cm]El.south) node[darksolid]{} node[right] {$\cy_\ell$};
	
	\node[factor] (El2) [right=1cm of El] {$W$};
	\draw (El) -- (El2) node[above,midway] {$\tilde{s}_\ell$};
	\node[sfactor] (Xl2) [above=.8cm of El2] {$p_X$};
	\draw (Xl2) -- (El2) node[right, midway] {$x_{\ell\!+\!1}$};
	\draw (El2.south) -- ([yshift=-.8cm]El2.south) node[darksolid]{} node[right] {$\cy_{\ell\!+\!1}$};
	
	\node[factor, draw=none] (Edummy2) [right=1cm of El2] {$\cdots$};
	\draw (El2) -- (Edummy2) node[above,pos=.7] {$\tilde{s}_{\ell\!+\!1}$};
	\node at (X1-|Edummy2) {$\cdots$};
	\node at (Y1-|Edummy2) {$\cdots$};
	
	\node[factor] (En) [right=1cm of Edummy2] {$W$};
	\draw (Edummy2) -- (En) node[above,midway] {$\tilde{s}_{n\!-\!1}$};
	\node[sfactor] (Xn) [above=.8cm of En] {$p_X$};
	\draw (Xn) -- (En) node[right, midway] {$x_n$};
	\draw (En.south) -- ([yshift=-.8cm]En.south) node[darksolid]{} node[right] {$\cy_{n}$};
	
	\node[sfactor, inner sep=0pt] (ee) [right=1cm of En] {$\mathbf{1}$};
	\draw (En) -- (ee) node[above,midway] {$\tilde{s}_n$};

	\begin{pgfonlayer}{bg}
	\draw[dashed, magenta, line width=1.5pt, fill=magenta!8] ([xshift=-1.4cm,yshift=2.5cm]S)
		rectangle ([xshift=.7cm,yshift=-3.8cm]En);
	\node[label] at ([xshift=.7cm,yshift=-3.8cm]En) {$\muY_{n}$};
	\draw[dashed, magenta, line width=1.5pt, fill=magenta!16] ([xshift=-1.2cm,yshift=2.3cm]S)
		rectangle ([xshift=.8cm,yshift=-3.2cm]El2);
	\node[label] at ([xshift=.8cm,yshift=-3.2cm]El2) {$\muY_{\ell+1}$};
	\draw[dashed, magenta, line width=1.5pt, fill=magenta!24] ([xshift=-1cm,yshift=2.1cm]S)
		rectangle ([xshift=.7cm,yshift=-2.6cm]El);
	\node[label] at ([xshift=.7cm,yshift=-2.6cm]El) {$\muY_{\ell}$};
	\draw[dashed, magenta, line width=1.5pt, fill=magenta!32] ([xshift=-.8cm,yshift=1.9cm]S)
		rectangle ([xshift=.7cm,yshift=-2cm]E2);
	\node[label] at ([xshift=.7cm,yshift=-2cm]E2) {$\muY_2$};
	\draw[dashed, magenta, line width=1.5pt, fill=magenta!40] ([xshift=-.6cm,yshift=1.7cm]S)
		rectangle ([xshift=.7cm,yshift=-1.4cm]E1);
	\node[label] at ([xshift=.7cm,yshift=-1.4cm]E1) {$\muY_1$};
	\end{pgfonlayer}
\end{tikzpicture}
}

\newcommand{\figCFSMestimatehXY}{
\begin{tikzpicture}[node/.style={draw=none},
	factor/.style={rectangle, minimum width=1cm, minimum height=.7cm, draw},
	sfactor/.style={rectangle, minimum size=.4cm, draw},
	darksolid/.style={rectangle, minimum size=.15cm, draw,fill = black, inner sep=0pt, outer sep = 0pt},
	label/.style={magenta,anchor=north east,xshift = .1cm}]
	\node[sfactor] (S) {$p_{\tilde S_0}$};
	
	\node[factor] (E1) [right=.7cm of S] {$W$};
	\draw (S) -- (E1) node[above,midway] {$\tilde{s}_0$};
	\node[darksolid] (X1) [above=.3cm of E1.north,anchor = south] {};
	\draw (X1) -- (E1) node[right, pos = 0] {$\cx_1$};
	\draw (X1) -- ([yshift=.3cm]X1.north) node (pX1) [sfactor, anchor = south, pos=1] {$p_X$};
	\draw (E1.south) -- ([yshift=-.8cm]E1.south) node (Y1) [darksolid]{} node[right] {$\cy_1$};
	
	\node[factor] (E2) [right=1cm of E1] {$W$};
	\draw (E1) -- (E2) node[above,midway] {$\tilde{s}_1$};
	\node[darksolid] (X2) [above=.3cm of E2.north,anchor = south] {};
	\draw (X2) -- (E2) node[right, pos = 0] {$\cx_2$};
	\draw (X2) -- ([yshift=.3cm]X2.north) node[sfactor, anchor = south, pos=1] {$p_X$};
	\draw (E2.south) -- ([yshift=-.8cm]E2.south) node[darksolid]{} node[right] {$\cy_2$};

	\node[factor, draw=none] (Edummy1) [right=1cm of E2] {$\cdots$};
	\draw (E2) -- (Edummy1) node[above,midway] {$s_2$};
	\node at (pX1-|Edummy1) {$\cdots$};
	\node at (Y1-|Edummy1) {$\cdots$};
	
	\node[factor] (El) [right=1cm of Edummy1] {$W$};
	\draw (Edummy1) -- (El) node[above,midway] {$\tilde{s}_{\ell-1}$};
	\node[darksolid] (Xl) [above=.3cm of El.north,anchor = south] {};
	\draw (Xl) -- (El) node[right, pos = 0] {$\cx_\ell$};
	\draw (Xl) -- ([yshift=.3cm]Xl.north) node[sfactor, anchor = south, pos=1] {$p_X$};
	\draw (El.south) -- ([yshift=-.8cm]El.south) node[darksolid]{} node[right] {$\cy_\ell$};
	
	\node[factor] (El2) [right=1cm of El] {$W$};
	\draw (El) -- (El2) node[above,midway] {$\tilde{s}_{\ell}$};
	\node[darksolid] (Xl2) [above=.3cm of El2.north,anchor = south] {};
	\draw (Xl2) -- (El2) node[right, pos = 0] {$\cx_{\ell\!+\!1}$};
	\draw (Xl2) -- ([yshift=.3cm]Xl2.north) node[sfactor, anchor = south, pos=1] {$p_X$};
	\draw (El2.south) -- ([yshift=-.8cm]El2.south) node[darksolid]{} node[right] {$\cy_{\ell\!+\!1}$};
	
	\node[factor, draw=none] (Edummy2) [right=1cm of El2] {$\cdots$};
	\draw (El2) -- (Edummy2) node[above,pos=.7] {$\tilde{s}_{\ell\!+\!1}$};
	\node at (pX1-|Edummy2) {$\cdots$};
	\node at (Y1-|Edummy2) {$\cdots$};
	
	\node[factor] (En) [right=1cm of Edummy2] {$W$};
	\draw (Edummy2) -- (En) node[above,midway] {$\tilde{s}_{n\!-\!1}$};
	\node[darksolid] (Xn) [above=.3cm of En.north,anchor = south] {};
	\draw (Xn) -- (En) node[right, pos = 0] {$\cx_n$};
	\draw (Xn) -- ([yshift=.3cm]Xn.north) node[sfactor, anchor = south, pos=1] {$p_X$};
	\draw (En.south) -- ([yshift=-.8cm]En.south) node[darksolid]{} node[right] {$\cy_{n}$};
	
	\node[sfactor, inner sep=0pt] (ee) [right=1cm of En] {$\mathbf{1}$};
	\draw (En) -- (ee) node[above,midway] {$\tilde{s}_n$};

	\begin{pgfonlayer}{bg}
	\draw[dashed, magenta, line width=1.5pt, fill=magenta!8] ([xshift=-1.4cm,yshift=2.5cm]S)
		rectangle ([xshift=.7cm,yshift=-3.8cm]En);
	\node[label] at ([xshift=.7cm,yshift=-3.8cm]En) {$\muXY_{n}$};
	\draw[dashed, magenta, line width=1.5pt, fill=magenta!16] ([xshift=-1.2cm,yshift=2.3cm]S)
		rectangle ([xshift=.8cm,yshift=-3.2cm]El2);
	\node[label] at ([xshift=.8cm,yshift=-3.2cm]El2) {$\muXY_{\ell+1}$};
	\draw[dashed, magenta, line width=1.5pt, fill=magenta!24] ([xshift=-1cm,yshift=2.1cm]S)
		rectangle ([xshift=.7cm,yshift=-2.6cm]El);
	\node[label] at ([xshift=.7cm,yshift=-2.6cm]El) {$\muXY_{\ell}$};
	\draw[dashed, magenta, line width=1.5pt, fill=magenta!32] ([xshift=-.8cm,yshift=1.9cm]S)
		rectangle ([xshift=.7cm,yshift=-2cm]E2);
	\node[label] at ([xshift=.7cm,yshift=-2cm]E2) {$\muXY_2$};
	\draw[dashed, magenta, line width=1.5pt, fill=magenta!40] ([xshift=-.6cm,yshift=1.7cm]S)
		rectangle ([xshift=.7cm,yshift=-1.4cm]E1);
	\node[label] at ([xshift=.7cm,yshift=-1.4cm]E1) {$\muXY_1$};
	\end{pgfonlayer}
\end{tikzpicture}
}

\newcommand{\figCFSMclosingthebox}{
\begin{tikzpicture}[node/.style={draw=none},
	factor/.style={rectangle, minimum width=1cm, minimum height=.7cm, draw},
	sfactor/.style={rectangle, minimum size=.4cm, draw},
	darksolid/.style={rectangle, minimum size=.15cm, draw,fill = black, inner sep=0pt, outer sep = 0pt},
	label/.style={magenta,anchor=north east,xshift = .1cm}]
	\node[sfactor] (S) {$p_{\tilde S_0}$};
	
	\node[factor] (E1) [right=.7cm of S] {$W$};
	\draw (S) -- (E1) node[above,midway] {$\tilde{s}_0$};
	\node[sfactor] (X1) [above=.8cm of E1] {$p_X$};
	\draw (X1) -- (E1) node[right, midway] {$x_1$};
	\draw (E1.south) -- ([yshift=-.8cm]E1.south) node (Y1) [right] {$y_1$};
	
	\node[factor] (E2) [right=1cm of E1] {$W$};
	\draw (E1) -- (E2) node[above,midway] {$\tilde{s}_1$};
	\node[sfactor] (X2) [above=.8cm of E2] {$p_X$};
	\draw (X2) -- (E2) node[right, midway] {$x_2$};
	\draw (E2.south) -- ([yshift=-.8cm]E2.south) node[right] {$y_2$};

	\node[factor, draw=none] (Edummy1) [right=1cm of E2] {$\cdots$};
	\draw (E2) -- (Edummy1) node[above,midway] {$s_2$};
	\node at (X1-|Edummy1) {$\cdots$};
	\node at (Y1-|Edummy1) {$\cdots$};
	
	\node[factor] (El) [right=1cm of Edummy1] {$W$};
	\draw (Edummy1) -- (El) node[above,midway] {$\tilde{s}_{\ell-1}$};
	\node[sfactor] (Xl) [above=.8cm of El] {$p_X$};
	\draw (Xl) -- (El) node[right, midway] {$x_\ell$};
	\draw (El.south) -- ([yshift=-.8cm]El.south) node[right] {$y_\ell$};
	
	\node[factor] (El2) [right=1cm of El] {$W$};
	\draw (El) -- (El2) node[above,midway] {$\tilde{s}_\ell$};
	\node[sfactor] (Xl2) [above=.8cm of El2] {$p_X$};
	\draw (Xl2) -- (El2) node[right, midway] {$x_{\ell\!+\!1}$};
	\draw (El2.south) -- ([yshift=-.8cm]El2.south) node[right] {$y_{\ell\!+\!1}$};
	
	\node[factor, draw=none] (Edummy2) [right=1cm of El2] {$\cdots$};
	\draw (El2) -- (Edummy2) node[above,midway] {$s_{\ell\!+\!1}$};
	\node at (X1-|Edummy2) {$\cdots$};
	\node at (Y1-|Edummy2) {$\cdots$};
	
	\node[factor] (En) [right=1cm of Edummy2] {$W$};
	\draw (Edummy2) -- (En) node[above,pos=0.4] {$\tilde{s}_{n\!-\!1}$};
	\node[sfactor] (Xn) [above=.8cm of En] {$p_X$};
	\draw (Xn) -- (En) node[right, midway] {$x_n$};
	\draw (En.south) -- ([yshift=-.8cm]En.south) node[right] {$y_{n}$};
	
	\node[sfactor, inner sep=0pt] (ee) [right=.5cm of En] {$\mathbf{1}$};
	\draw (En) -- (ee) node[above,midway] {$\tilde{s}_n$};

	\begin{pgfonlayer}{bg}
	\draw[dashed, green, line width=1.5pt, fill=green!10] ([xshift=-2.1cm,yshift=2.8cm]E1)
		rectangle ([xshift=1.5cm,yshift=-2.4cm]ee);
	\draw[dashed, green, line width=1.5pt, fill=green!20] ([xshift=-.6cm,yshift=2.6cm]E1)
		rectangle ([xshift=1.3cm,yshift=-2.2cm]ee);
	\draw[dashed, green, line width=1.5pt, fill=green!30] ([xshift=-.7cm,yshift=2.4cm]E2)
		rectangle ([xshift=1.1cm,yshift=-2cm]ee);
	\draw[dashed, green, line width=1.5pt, fill=green!40] ([xshift=-.6cm,yshift=2.2cm]El)
		rectangle ([xshift=.9cm,yshift=-1.8cm]ee);
	\draw[dashed, green, line width=1.5pt, fill=green!50] ([xshift=-.7cm,yshift=2cm]El2)
		rectangle ([xshift=.7cm,yshift=-1.6cm]ee);
	\draw[dashed, green, line width=1.5pt, fill=green!60] ([xshift=-.7cm,yshift=1.8cm]En)
		rectangle ([xshift=.5cm,yshift=-1.4cm]ee);
	\end{pgfonlayer}
\end{tikzpicture}
}

\newcommand{\figQFSMchannelsimulationY}{
\begin{tikzpicture}[node/.style={draw=none},
	factor/.style={rectangle, minimum width=1cm, minimum height=1.7cm, draw},
	sfactor/.style={rectangle, minimum size=.4cm, draw},
	darksolid/.style={rectangle, minimum size=.15cm, draw,fill = black, inner sep=0pt, outer sep = 0pt}]
	\node[sfactor] (S) {$\rhoSt{0}$};
	
	\node[factor] (E1) [right=.7cm of S] {$W$};
	\draw (S.north) |- ([yshift=.6cm]E1.west) node[above=-0.05cm,pos=.75] {$s_0$};
	\draw (S.south) |- ([yshift=-.6cm]E1.west) node[below,pos=.75] {$s_0'$};
	\node[sfactor] (X1) [above=.7cm of E1] {$p_X$};
	\draw (X1) -- (E1) node[right, midway] {$x_1$};
	\draw (E1.south) -- ([yshift=-.8cm]E1.south) node[darksolid]{} node (Y1) [right] {$\cy_1$};
	
	\node[factor] (E2) [right=1cm of E1] {$W$};
	\draw ([yshift=.6cm]E1.east) |- ([yshift=.6cm]E2.west) node[above=-0.05cm,pos=.75] {$s_1$};
	\draw ([yshift=-.6cm]E1.east) |- ([yshift=-.6cm]E2.west) node[below,pos=.75] {$s_1'$};
	\node[sfactor] (X2) [above=.7cm of E2] {$p_X$};
	\draw (X2) -- (E2) node[right, midway] {$x_2$};
	\draw (E2.south) -- ([yshift=-.8cm]E2.south) node[darksolid]{} node[right] {$\cy_2$};

	\node[factor, draw=none] (Edummy1) [right=1cm of E2] {};
	\node at ([yshift=.6cm]E2.east-|Edummy1) {$\cdots$};
	\node at ([yshift=-.6cm]E2.east-|Edummy1) {$\cdots$};
	\draw ([yshift=.6cm]E2.east) |- ([yshift=.6cm]Edummy1.west) 
		node[above=-0.05cm,pos=.75] {$s_2$};
	\draw ([yshift=-.6cm]E2.east) |- ([yshift=-.6cm]Edummy1.west) node[below,pos=.75] {$s_2'$};
	\node at (X1-|Edummy1) {$\cdots$};
	\node at (Y1-|Edummy1) {$\cdots$};
	
	\node[factor] (El0) [right=1cm of Edummy1] {$W$};
	\draw ([yshift=.6cm]Edummy1.east) |- ([yshift=.6cm]El0.west) 
		node[above=-0.05cm,pos=.75] {$s_{\ell-2}$};
	\draw ([yshift=-.6cm]Edummy1.east) |- ([yshift=-.6cm]El0.west)
		node[below,pos=.75] {$s_{\ell-2}'$};
	\node[sfactor] (Xl0) [above=.7cm of El0] {$p_X$};
	\draw (Xl0) -- (El0) node[right, midway] {$x_{\ell\!-\!1}$};
	\draw (El0.south) -- ([yshift=-.8cm]El0.south) node[darksolid]{} node[right] {$\cy_{\ell\!-\!1}$};
	
	\node[factor] (El) [right=1cm of El0] {$W$};
	\draw ([yshift=.6cm]El0.east) |- ([yshift=.6cm]El.west) 
		node[above=-0.05cm,pos=.85] {$s_{\ell\!-\!1}$};
	\draw ([yshift=-.6cm]El0.east) |- ([yshift=-.6cm]El.west)
		node[below,pos=.85] {$s_{\ell\!-\!1}'$};

	\node[darksolid] (Xl) [above=.3cm of El.north,anchor = south] {};
	\draw (Xl) -- (El) node[right, pos = 0] {$\cx_\ell$};
	\draw (Xl) -- ([yshift=.3cm]Xl.north) node (pXl) [sfactor, anchor = south, pos=1] {$p_X$};

	\draw (El.south) -- ([yshift=-2.2cm]El.south) node[right] {$y_\ell$};
	
	\node[factor] (El2) [right=1cm of El] {$W$};
	\draw ([yshift=.6cm]El.east) |- ([yshift=.6cm]El2.west) node[above=-0.05cm,pos=.75] {$s_\ell$};
	\draw ([yshift=-.6cm]El.east) |- ([yshift=-.6cm]El2.west) node[below,pos=.75] {$s_\ell'$};
	\node[sfactor] (Xl2) [above=.7cm of El2] {$p_X$};
	\draw (Xl2) -- (El2) node[right, midway] {$x_{\ell\!+\!1}$};
	\draw (El2.south) -- ([yshift=-.8cm]El2.south) node[right] {$y_{\ell\!+\!1}$};
	
	\node[factor, draw=none] (Edummy2) [right=1cm of El2] {};
	\node at ([yshift=.6cm]El2.east-|Edummy2) {$\cdots$};
	\node at ([yshift=-.6cm]El2.east-|Edummy2) {$\cdots$};
	\draw ([yshift=.6cm]El2.east) |- ([yshift=.6cm]Edummy2.west) 
		node[above=-0.05cm,pos=.75] {$s_{\ell\!+\!1}$};
	\draw ([yshift=-.6cm]El2.east) |- ([yshift=-.6cm]Edummy2.west)
		node[below,pos=.75] {$s_{\ell\!+\!1}'$};
	\node at (X1-|Edummy2) {$\cdots$};
	\node at (Y1-|Edummy2) {$\cdots$};
	
	\node[factor] (En) [right=1cm of Edummy2] {$W$};
	\draw ([yshift=.6cm]Edummy2.east) |- ([yshift=.6cm]En.west) 
		node[above=-0.05cm,pos=.75] {$s_{n\!-\!1}$};
	\draw ([yshift=-.6cm]Edummy2.east) |- ([yshift=-.6cm]En.west)
		node[below,pos=.75] {$s_{n\!-\!1}'$};
	\node[sfactor] (Xn) [above=.7cm of En] {$p_X$};
	\draw (Xn) -- (En) node[right, midway] {$x_n$};
	\draw (En.south) -- ([yshift=-.8cm]En.south) node[right] {$y_n$};
	
	\node[sfactor, inner sep=0pt] (ee) [right=.5cm of En] {$=$};
	\draw ([yshift=.6cm]En.east) -| (ee.north) node[above=-0.05cm,pos=.25] {$s_n$};
	\draw ([yshift=-.6cm]En.east) -| (ee.south) node[below,pos=.25] {$s_n'$};

	\begin{pgfonlayer}{bg}
	\draw[dashed, black, line width=1.5pt, fill=yellow!10] ([xshift=-.8cm,yshift=2.4cm]S)
		rectangle ([xshift=.8cm,yshift=-2.3cm]ee);
	\node[anchor=north west] at ([xshift=-.8cm,yshift=-2.3cm]S|-ee)
		{$p_{X_{\ell}, \vY_1^{\ell-1}, Y_{\ell}}
                    (\cx_{\ell}, \cvy_1^{\ell-1}, y_{\ell})$};
	
	\draw[dashed, green, line width=1.5pt, fill=green!20] ([xshift=-.7cm,yshift=2.2cm]El2)
		rectangle ([xshift=.6cm,yshift=-2.1cm]ee);
	\node[anchor=south east, fill=green!50] at ([xshift=.6cm,yshift=-2.1cm]ee)
		{$\delta$};

	\draw[dashed, magenta, line width=1.5pt, fill=magenta!20] ([xshift=-.6cm,yshift=2.2cm]S)
		rectangle ([xshift=.8cm,yshift=-2.1cm]El0);
	\node[anchor=south east,fill=magenta!50] at ([xshift=.6cm,yshift=-2.1cm]S|-El0)
		{$\sigmaSt{\ell-1}_{|\cvy_{1}^{\ell-1}}$};
	\end{pgfonlayer}
\end{tikzpicture}
}

\newcommand{\figCFSMchannelsimulationY}{
\begin{tikzpicture}[node/.style={draw=none},
	factor/.style={rectangle, minimum width=1cm, minimum height=.7cm, draw},
	sfactor/.style={rectangle, minimum size=.4cm, draw},
	darksolid/.style={rectangle, minimum size=.15cm, draw,fill = black, inner sep=0pt, outer sep = 0pt}]
	\node[sfactor] (S) {$p_{\tilde S_0}$};
	
	\node[factor] (E1) [right=.7cm of S] {$W$};
	\draw (S) -- (E1) node[above,midway] {$\tilde{s}_0$};
	\node[sfactor] (X1) [above=.8cm of E1] {$p_X$};
	\draw (X1) -- (E1) node[right, midway] {$x_1$};
	\draw (E1.south) -- ([yshift=-.8cm]E1.south) node (Y1) [darksolid]{} node[right] {$\cy_1$};
	
	\node[factor] (E2) [right=1cm of E1] {$W$};
	\draw (E1) -- (E2) node[above,midway] {$\tilde{s}_1$};
	\node[sfactor] (X2) [above=.8cm of E2] {$p_X$};
	\draw (X2) -- (E2) node[right, midway] {$x_2$};
	\draw (E2.south) -- ([yshift=-.8cm]E2.south) node[darksolid]{} node[right] {$\cy_2$};

	\node[factor, draw=none] (Edummy1) [right=1cm of E2] {$\cdots$};
	\draw (E2) -- (Edummy1) node[above,midway] {$\tilde{s}_2$};
	\node at (X1-|Edummy1) {$\cdots$};
	\node at (Y1-|Edummy1) {$\cdots$};
	
	\node[factor] (El0) [right=1cm of Edummy1] {$W$};
	\draw (Edummy1) -- (El0) node[above,midway] {$\tilde{s}_{\ell-2}$};
	\node[sfactor] (Xl0) [above=.8cm of El0] {$p_X$};
	\draw (Xl0) -- (El0) node[right, midway] {$x_{\ell\!-\!1}$};
	\draw (El0.south) -- ([yshift=-.8cm]El0.south) node[darksolid]{} node[right] {$\cy_{\ell\!-\!1}$};
	
	\node[factor] (El) [right=1cm of El0] {$W$};
	\draw (El0) -- (El) node[above,pos=.65] {$\tilde{s}_{\ell\!-\!1}$};

	\node[darksolid] (Xl) [above=.3cm of El.north,anchor = south] {};
	\draw (Xl) -- (El) node[right, pos = 0] {$\cx_\ell$};
	\draw (Xl) -- ([yshift=.3cm]Xl.north) node (pXl) [sfactor, anchor = south, pos=1] {$p_X$};
	\draw (El.south) -- ([yshift=-2.2cm]El.south) node[right] {$y_\ell$};

	\node[factor] (El2) [right=1cm of El] {$W$};
	\draw (El) -- (El2) node[above,midway] {$\tilde{s}_\ell$};
	\node[sfactor] (Xl2) [above=.8cm of El2] {$p_X$};
	\draw (Xl2) -- (El2) node[right, midway] {$x_{\ell+1}$};
	\draw (El2.south) -- ([yshift=-.8cm]El2.south) node[right] {$y_{\ell+1}$};
	
	\node[factor, draw=none] (Edummy2) [right=1cm of El2] {$\cdots$};
	\draw (El2) -- (Edummy2) node[above,midway] {$\tilde{s}_{\ell+1}$};
	\node at (X1-|Edummy2) {$\cdots$};
	\node at (Y1-|Edummy2) {$\cdots$};
	
	\node[factor] (En) [right=1cm of Edummy2] {$W$};
	\draw (Edummy2) -- (En) node[above,midway] {$\tilde{s}_{n\!-\!1}$};
	\node[sfactor] (Xn) [above=.8cm of En] {$p_X$};
	\draw (Xn) -- (En) node[right, midway] {$x_n$};
	\draw (En.south) -- ([yshift=-.8cm]En.south) node[right] {$y_n$};
	
	\node[sfactor, inner sep=0pt] (ee) [right=.5cm of En] {$\mathbf{1}$};
	\draw (En) -- (ee) node[above,midway] {$\tilde{s}_n$};

	\begin{pgfonlayer}{bg}
	\draw[dashed, black, line width=1.5pt, fill=yellow!10] ([xshift=-.8cm,yshift=1.95cm]S)
		rectangle ([xshift=.6cm,yshift=-1.6cm]ee);
	\node[anchor=north west] at ([xshift=-.6cm,yshift=-1.6cm]S|-ee)
		{$p_{X_{\ell}, \vY_1^{\ell-1}, Y_{\ell}}
                    (\cx_{\ell}, \cvy_1^{\ell-1}, y_{\ell})$};
	
	\draw[dashed, green, line width=1.5pt, fill=green!20] ([xshift=-.7cm,yshift=1.75cm]El2)
		rectangle ([xshift=.4cm,yshift=-1.4cm]ee);
	\node[anchor=south east, fill=green!50] at ([xshift=.4cm,yshift=-1.4cm]ee)
		{$\mathbf{1}$};

	\draw[dashed, magenta, line width=1.5pt, fill=magenta!20] ([xshift=-.6cm,yshift=1.75cm]S)
		rectangle ([xshift=.8cm,yshift=-1.4cm]El0);
	\node[anchor=south east,fill=magenta!50] at ([xshift=.6cm,yshift=-1.4cm]S|-El0)
		{$\mu^{\ell-1}_{|\cvy_{1}^{\ell-1}}$};
	\end{pgfonlayer}
\end{tikzpicture}
}

\begin{document}

\title{Estimating the Information Rate of a Channel with \\
       Classical Input and Output and a Quantum State}

\author{
    \IEEEauthorblockN{Michael~X.~Cao and Pascal~O.~Vontobel}
    \IEEEauthorblockA{Department of Information Engineering \\
                      The Chinese University of Hong Kong \\
                      \{m.x.cao, pascal.vontobel\}@ieee.org\\[-0.15cm]}
}

\maketitle

\ifx\longversion\x
\watermark{\put(0, 12){\texttt{This is an extended version of a paper 
                         that appears in}}
            \put(0,  0){\texttt{Proc.\ 2017 IEEE International Symposium on 
                          Information Theory, Aachen, Germany.}}}
\fi

\begin{abstract}
  We consider the problem of transmitting classical information over a
  time-invariant channel with memory. A popular class of time-invariant
  channels with memory are finite-state-machine channels, where a
  \emph{classical} state evolves over time and governs the relationship
  between the classical input and the classical output of the channel. For
  such channels, various techniques have been developed for estimating and
  bounding the information rate.

  In this paper we consider a class of time-invariant channels where a
  \emph{quantum} state evolves over time and governs the relation\-ship
  between the classical input and the classical output of the channel. We
  propose algorithms for estimating and bounding the information rate of such
  channels. In particular, we discuss suitable graphical models for doing the
  relevant computations.
\end{abstract}

\section{Introduction}
\label{sec:introduction}

In this section, we first review some results about classical channels, in
particular channels with an evolving classical state. Afterwards, we discuss
channels with an evolving quantum state. Finally, we highlight the
contributions of this paper.

\subsection{Information Rates of Classical Channels}

The information rate of a classical point-to-point channel characterizes the
amount of classical information per channel use that can be transmitted
reliably with the help of this channel. A particularly interesting class of
channels are discrete memoryless channels (DMCs). A DMC is characterized by a
channel input alphabet $\set{X}$, a channel output alphabet $\set{Y}$, and a
channel law $W(y|x)$, where the latter equals the probability of receiving $y$
upon sending $x$. (Here and in the following, we assume that $\set{X}$ and
$\set{Y}$ are finite sets.) As is well known~\cite{Gallager:68}, the
information rate $I(Q,W)$ of a DMC is given by
\begin{align}
  I(Q,W)
    &= \sum_{x \in \set{X}}
       \sum_{y \in \set{Y}}
         Q(x) W(y|x)
         \log
           \left(
             \frac{W(y|x)}{(QW)(y)}
           \right) \ ,
             \label{eq:dmc:information:rate:1}
\end{align}
where $Q$ is some probability mass function (pmf) on $\set{X}$ and $(QW)(y)
\defeq \sum_{x \in \set{X}} Q(x) W(y|x)$. Recall that in order to achieve this
information rate, one needs to design a suitable encoder and a suitable
decoder for some suitably chosen codebook where the distribution of the
entries of the codewords equals $Q$. Because of the simplicity of the
expression in~\eqref{eq:dmc:information:rate:1}, the information $I(Q,W)$ can
be efficiently computed for any given $Q$.\footnote{Note that even the
  maximization of $I(Q,W)$ over all pmfs over $\set{X}$ can be done
  efficiently~\cite{Arimoto:72, Blahut:72}. The maximal information rate is
  known as the capacity of the DMC and the maximizing $Q$ is known as the
  capacity-achieving input distribution.}

\begin{Example}
  For any $0 \leq p \leq 1$, the binary symmetric channel with cross-over
  probability $p$, henceforth called BSC($p$), is a DMC with $\set{X}
  \defeq \{ 0, 1 \}$, $\set{Y} = \{ 0, 1 \}$, $W(0|0) = 1 - p$, $W(1|0) =
  p$, $W(0|1) = p$, and $W(1|1) = 1 - p$. If $Q(0) = Q(1) = 1/2$,
  then its information rate is $I(Q,W) = 1 - h_2(p)$ bits per channel use,
  where $h_2$ is the binary entropy function.
\end{Example}

We proceed to channels with memory, in particular to stationary ergodic
channels with input alphabet $\set{X}$ and output alphabet $\set{Y}$. Let $W$
denote the channel law of such a channel. Under suitable
conditions~\cite{Gallager:68}, the information rate is given by $I(Q,W) =
\lim_{n \to \infty} \frac{1}{n} I(X_1, X_2, \ldots, X_n; Y_1, Y_2, \ldots
Y_n)$, where $\sX \defeq (X_1, X_2, \ldots)$ is the channel input process
characterized by some stationary ergodic law $Q$, and where $\sY \defeq (Y_1,
Y_2, \ldots)$ is the channel output process.

For such channels, computing the information rate, let alone the capacity, is
much more challenging than for DMCs. Namely, except for very special cases,
there are no single-letter or other simple expressions for information rates
available, and so, most of the time, one needs to rely on upper and lower
bounds and/or on stochastic techniques for estimating the information rate.

Notably, in the case of finite-state-machine channels
(FSMCs)~\cite{Gallager:68}, \ie, channels with a finite classical state,
efficient stochastic techniques have been developed for estimating the
information rate~\cite{Arnold:Loeliger:Vontobel:Kavcic:Zeng:06:1,
  Sharma:Singh:01:1, Pfister:Soriaga:Siegel:01:1}. (For these techniques,
under mild conditions, the numerical estimate of the information rate
converges with probability one to the true value when the length of the
channel input sequence goes to infinity.) However, even for FSMCs, maximizing
the information rate is much more challenging than maximizing the information
rate of DMCs~\cite{Vontobel:Kavcic:Arnold:Loeliger:08:1}.

\begin{Example}
  A notable example of an FSMC is the Gilbert--Elliott
  channel~\cite{Mushkin:BarDavid:89:1}, which can be either in the so-called
  ``good'' state or in the so-called ``bad'' state. If the channel is in the
  ``good'' state, then it behaves like a BSC($\pgood$), but if the channel is
  in the ``bad'' state, then it behaves like a BSC($\pbad$), where usually
  $\bigl| \pbad \! - \! \frac{1}{2} \bigr| < \bigl| \pgood \! - \!
  \frac{1}{2} \bigr|$. The state process itself is a first-order stationary
  ergodic Markov process which is independent of the input
  process.\footnote{The independence of the state process on the input process
    is a particular feature of the Gilbert--Elliott channel. In general, the
    state process of a finite-state channel can depend on the input process.}
  (For more details, see, \eg, the discussions
  in~\cite{Vontobel:Kavcic:Arnold:Loeliger:08:1,
    Sadeghi:Vontobel:Shams:09:1}.)
\end{Example}

For FSMCs with large state spaces, the above-mentioned information rate
estimation techniques can be time-consuming and so stochastic techniques to
estimate upper and lower bounds have proven
useful~\cite{Arnold:Loeliger:Vontobel:Kavcic:Zeng:06:1,
  Sadeghi:Vontobel:Shams:09:1}. These bounding techniques are based on a
so-called auxiliary FSMC, which is a low-complexity approximation of the true
FSMC. Interestingly enough, the lower bounds represent achievable rates under
mismatched decoding, where the decoder bases its computations not on the true
FSMC but on the auxiliary FSMC~\cite{Ganti:Lapidoth:Telatar:00:1}. (See the
paper~\cite{Sadeghi:Vontobel:Shams:09:1} for a more detailed discussion of
this topic and for further references.)

\subsection{Information Rates of Channels with a Quantum State
                      --- Paper Overview}
\label{sec:information:rate:channels:with:quantum:state:1}

In this paper we consider the problem of transmitting classical information
over a channel with an evolving \emph{quantum} state. A particular instance of
such a channel is as follows:
\begin{itemize}

\item The state is given by some quantum system, called the state quantum
  system, whose position in space does not change and which, if left by
  itself, evolves according to some Hamiltonian $\mathcal{H}_{\mathrm{s}}$.

\item Alice wants to transmit some classical information to Bob. To this end,
  she uses a classical code to encode her information word $\vu = (u_1, u_2,
  \ldots, u_k) \in \setU^k$ into a codeword $\vx = (x_1, x_2, \ldots, x_n) \in
  \setX^n$.

\item At time instance $\ell$, Alice encodes $x_{\ell} \in \setX$ as a
  particular state of some quantum system, called the $\ell$-th transmit
  quantum system, which she sends to Bob.

\item On the way to Bob, the $\ell$-th transmit quantum system interacts with
  the state quantum system.

\item Bob receives the $\ell$-th transmit quantum system and performs a
  quantum measurement resulting in some value $y_{\ell} \in \setY$.

\item After receiving $\vy = (y_1, y_2, \ldots, y_n) \in \setY^n$, Bob decodes
  $\vy$ toward obtaining an estimate $\hvu$ of $\vu$.

\end{itemize}
This setup is vaguely inspired by the setup in Fig.~4 of
\cite{Haroche:Wineland:Nobelprize:2012:1}. Note that the setup therein was
\emph{not} for data communication, but for manipulating and measuring what we
call here the state quantum system.

In this paper, we discuss algorithms for estimating and lower bounding the
information rate of such channels with an evolving quantum state (see
Section~\ref{sec:information:rate:estimation:1}). Toward this end, we
introduce suitable graphical models for visualizing and doing the relevant
computations (see Section~\ref{sec:graphical:models:1}). Finally, we present
some numerical results (see Section~\ref{sec:numerical:examples:1}).

\subsection{References with Background Information}
\label{sec:background:1}

\ifx\shortversion\x
Due to space constraints, we limit ourselves to conveying the main ideas
behind our results. (The details can be found in~\cite{Cao:Vontobel:17:2}.)
\fi

In the following, we assume that the reader is familiar with the very basics
of quantum information processing (see, \eg, the excellent book Nielsen and
Chuang~\cite{Nielsen:Chuang:00:1} for an introduction).  For a general
introduction to quantum channels with memory, we refer to the survey papers by
Kretschmann and Werner~\cite{Kretschmann:Werner:05:1} and by Caruso
\etal~\cite{Caruso:Giovannetti:Lupo:Mancini:14:1}.

Moreover, some familiarity with graphical models (like factor
graphs)~\cite{Kschischang:Frey:Loeliger:01, Forney:01:1, Loeliger:04:1} and
with techniques for estimating the information rate of a classical FSMC as
presented in~\cite{Arnold:Loeliger:Vontobel:Kavcic:Zeng:06:1,
  Sadeghi:Vontobel:Shams:09:1} will be beneficial. Recall that graphical
models are a popular approach for representing multivariate functions with
non-trivial factorizations and for doing computations like
marginalization~\cite{Kschischang:Frey:Loeliger:01, Forney:01:1,
  Loeliger:04:1}. In particular, graphical models can be used to represent
joint probability mass functions (pmfs) / probability density functions
(pdfs). In the present paper we will heavily rely on the
papers~\cite{Loeliger:Vontobel:12:1, Loeliger:Vontobel:15:1:subm}, which
discussed an approach for using normal factor graphs (NFGs) for representing
functions that typically appear when doing computations w.r.t.\ some quantum
systems. Probabilities of interest are then obtained by suitably applying the
sum-product algorithm / applying the closing-the-box operation.

\section{Channels with Classical or Quantum States 
               and their Graphical Models}
\label{sec:graphical:models:1}

We first review NFGs that were used
in~\cite{Arnold:Loeliger:Vontobel:Kavcic:Zeng:06:1} in the context of
estimating the information rate of channels with an evolving classical
state. Afterwards, we will show NFGs that we can use for estimating the
information rate of channels with an evolving quantum state.

\subsection{Channels with a Classical State}
\label{sec:channel:with:classical:state:1}

Fig.~\ref{fig:FMSC:high:level:1} shows the NFG that was used
in~\cite{Arnold:Loeliger:Vontobel:Kavcic:Zeng:06:1} in the context of
estimating the information rate of channels with a classical state. Let
$g(\vx,\vy,\tvs)$ denote the global function of this NFG (\ie, the
multivariate function represented by this NFG), where $\vx \defeq (x_1,
\ldots, x_n)$, $\vy \defeq (y_1, \ldots, y_n)$, and $\tvs \defeq (\tilde s_0,
\tilde s_1, \ldots, \tilde s_n)$. Some comments:
\begin{itemize}

\item The part of the NFG inside the blue box represents the input process
  $Q(\vx)$. Here, for simplicity, the input process is an i.i.d.\ process
  characterized by the pmf $p_X$, \ie, $Q(\vx) = \prod_{\ell=1}^{n}
  p_X(x_{\ell})$.

\item The part of the NFG inside the red box represents $W(\vy, \tvs | \vx)$,
  \ie, the probability of $\vy$ and $\tvs$ given $\vx$. After applying the
  closing-the-box operation, \ie, after summing over all the variables
  associated with edges completely inside the red box, we obtain the channel
  law $W(\vy | \vx) \defeq \sum_{\tvs} W(\vy, \tvs | \vx)$.

\item The function $W(\vy, \tvs | \vx)$ decomposes as $W(\vy, \tvs | \vx) =
  p_{\tilde S_0}(\tilde s_0) \cdot \prod_{\ell=1}^{n} W(\tilde s_{\ell},
  y_{\ell} | \tilde s_{\ell-1}, x_{\ell})$. Here, $p_{\tilde S_0}$ is a pmf
  and $W(\tilde s_{\ell}, y_{\ell} | \tilde s_{\ell-1}, x_{\ell})$ is assumed
  to be a conditional pmf:
  \begin{align}
    \forall \ \tilde s_{\ell}, y_{\ell}, \tilde s_{\ell-1}, x_{\ell}: & \ \ 
      W(\tilde s_{\ell}, y_{\ell} | \tilde s_{\ell-1}, x_{\ell}) \geq 0 \ ,
     \label{eq:classical:W:condition:1} \\
    \forall \ \tilde s_{\ell-1}, x_{\ell}: & \ \ 
      \sum_{\tilde s_{\ell}, \, y_{\ell}}
        W(\tilde s_{\ell}, y_{\ell} | \tilde s_{\ell-1}, x_{\ell}) = 1 \ .
      \label{eq:classical:W:condition:2}
  \end{align}

\end{itemize}
With this, one can verify that the NFG in Fig.~\ref{fig:FMSC:high:level:1} has
the following properties. (Most of these properties are in contrast to the
properties of the upcoming NFG that we will use for channels with a quantum
state.)
\begin{itemize}

\item The global function $g(\vx, \vy, \tvs)$ is a pmf over $\vx$, $\vy$, and
  $\tvs$.

\item The function $g(\vx, \vy) \defeq \sum_{\tvs} g(\vx, \vy,\tvs)$, which is
  obtained by summing the global function over $\tvs$, represents the
  corresponding marginal pmf over $\vx$ and $\vy$. The function $g(\tvs)
  \defeq \sum_{\vx, \, \vy} g(\vx, \vy,\tvs)$, which is obtained by summing
  the global function over $\vx$ and $\vy$, represents the corresponding
  marginal pmf over $\tvs$. \Etc

\end{itemize}

\ifx\longversion\x

\noindent
\begin{center}
  \textcolor{black}{--- See Appendix~\ref{sec:app:1} for additional
    comments. ---}
\end{center}

\fi

\subsection{Channels with a Quantum State}
\label{sec:channel:with:quantum:state:1}

We now turn our attention to channels with an evolving quantum state. In this
case, it is in general \emph{not} possible to come up with an NFG that has a
``nice'' factorization and that has the properties listed at the end of
Section~\ref{sec:channel:with:classical:state:1}. However, note that we
``only'' need an NFG with a global function $g(\vx, \vy, \ldots)$ which has
the property that if we sum over all variables except $\vx$ and $\vy$, then we
obtain a pmf over $\vx$ and $\vy$. In particular, we do \emph{not} need
$g(\vx, \vy, \ldots)$ to have the property that if we sum over $\vx$ and $\vy$
then the resulting function is a pmf over the remaining variables.

Consider an NFG with global function $g(\vx,\vy,\vs,\vs')$, where $\vx \defeq
(x_1, \ldots, x_n)$, $\vy \defeq (y_1, \ldots, y_n)$, $\vs \defeq (s_0, s_1,
\ldots, s_n)$, and $\vs' \defeq (s'_0, s'_1, \ldots, s'_n)$. Define
$g(\vx,\vy) \defeq \sum_{\vs, \, \vs'} g(\vx,\vy,\vs,\vs')$. The
above-mentioned conditions mean that $g(\vx,\vy)$ must be a pmf over $\vx$ and
$\vy$, but $g(\vx,\vy,\vs,\vs')$ need \emph{not} be a pmf over $\vx$, $\vy$,
$\vs$, and $\vs'$. In particular, $g(\vs,\vs') \defeq \sum_{\vx, \, \vy}
g(\vx,\vy,\vs,\vs')$ need \emph{not} be a pmf over $\vs$ and $\vs'$.

As it happens to be, considering NFGs whose global function
$g(\vx,\vy,\vs,\vs')$ satisfies
\begin{align}
  \forall \ \vx, \vy, \vs, \vs' : & \ \  
    g(\vx,\vy,\vs,\vs') \in \mathbb{C} \ , 
      \label{eq:quantum:channel:nfg:property:1} \\
  \forall \ \vx, \vy, \vs, \vs' : & \ \ 
    g(\vx,\vy,\vs,\vs') = \overline{g(\vx,\vy,\vs',\vs)} \ , 
      \label{eq:quantum:channel:nfg:property:2} \\
                                  & \ \ 
    \sum_{\vx, \, \vy, \, \vs, \, \vs'} 
      g(\vx,\vy,\vs,\vs') = 1 \ , 
      \label{eq:quantum:channel:nfg:property:3} \\
  \forall \ \vx, \vy : & \ \ 
    g(\vx,\vy) \in \mathbb{R}_{\geq 0} \ ,
      \label{eq:quantum:channel:nfg:property:4} \\
                       & \ \
    \sum_{\vx, \, \vy} 
      g(\vx,\vy) = 1 \ , 
      \label{eq:quantum:channel:nfg:property:5} 
\end{align}
is general enough in order to capture quantum phenomena and to represent the
associated computations with the help of NFGs that have a ``nice''
factorization~\cite{Loeliger:Vontobel:15:1:subm}.\footnote{The over-line
  in~\eqref{eq:quantum:channel:nfg:property:2} denotes complex
  conjugation. Note that condition~\eqref{eq:quantum:channel:nfg:property:5}
  is redundant given condition~\eqref{eq:quantum:channel:nfg:property:3}, but
  we display it because of its importance.}

With suitably chosen local function nodes, the NFG in
Fig.~\ref{fig:QFMSC:high:level:1} is an NFG that
satisfies~\eqref{eq:quantum:channel:nfg:property:1}--\eqref{eq:quantum:channel:nfg:property:5}.
Specifically, it suffices to impose the following requirements on the local
function nodes:
\begin{itemize}

\item The input process is an i.i.d.\ process characterized by the pmf
  $p_X$. (This is for simplicity only; more complicated processes could be
  used.)

\item In order to show the constraints on the function $W$, it is beneficial
  to write its arguments as follows: $W^{(y_{\ell}|x_{\ell})}(s_{\ell-1},
  s_{\ell}; s'_{\ell-1}, s'_{\ell})$. Moreover, for any fixed $x_{\ell}$ and
  $y_{\ell}$, we denote by $\matr{W}^{(y_{\ell}|x_{\ell})}$ the matrix with
  row labels $(s_{\ell-1}, s_{\ell})$, with column labels $(s'_{\ell-1},
  s'_{\ell})$, and with entries $W^{(y_{\ell}|x_{\ell})}(s_{\ell-1},
  s_{\ell}; s'_{\ell-1}, s'_{\ell})$. With this, $W$ has to satisfy
  \begin{align}
    &
    \hspace{-0.2cm}
    \forall \ x_{\ell}, y_{\ell} : \ \ 
    \matr{W}^{(y_{\ell}|x_{\ell})} 
      \ \text{is a p.s.d. matrix (over $\mathbb{C}$)} \ ,
              \label{eq:quantum:channel:nfg:property:11}  \\
    &
    \hspace{-0.2cm}
    \forall \ x_{\ell}, s_{\ell-1}, s'_{\ell-1} \! : \!\!\!\!\!
    \sum_{s_{\ell}, \, s'_{\ell}, \, y_{\ell}} \!\!\!\!\!
      W^{(y_{\ell}|x_{\ell})}(s_{\ell-1}, s_{\ell}; s'_{\ell-1}, s'_{\ell})
      \! \cdot \!
      \delta(s'_{\ell},s_{\ell}) \nonumber \\[-0.25cm]
    &\hspace{4.5cm}
       = \delta(s'_{\ell-1},s_{\ell-1}) \ ,
                 \label{eq:quantum:channel:nfg:property:12} 
  \end{align}
  where p.s.d.\ stands for positive semi-definite and where $\delta$ is the
  Kronecker-delta function. Note that
  condition~\eqref{eq:quantum:channel:nfg:property:12} can be visualized as
  shown in Fig.~\ref{fig:quantum:channel:condition:1}, where applying a
  closing-the-box operation~\cite{Loeliger:Vontobel:15:1:subm} to the NFG on
  the left-hand side results in the NFG on the right-hand side. (On the side,
  we note that with the above ordering of the entries, for every $x_{\ell}$
  the matrix $\sum_{y_{\ell}} W^{(y_{\ell}|x_{\ell})}$ is known to be in
  Choi-matrix-representation form or in dynamical-matrix-representation
  form~\cite{Wood:Biamonte:Cory:15:1}.)

\item The initial quantum state $(s_0,s'_0)$ is characterized by the
  complex-valued function $\rhoSt{0}$, which, when written as a matrix, is
  p.s.d.\ (over $\mathbb{C}$) and has trace one.

\end{itemize}
One can verify that these constraints on the local functions of the NFG in
Fig.~\ref{fig:QFMSC:high:level:1} lead to a global function $g(\vx, \vy, \vs,
\vs')$ which
satisfies~\eqref{eq:quantum:channel:nfg:property:1}--\eqref{eq:quantum:channel:nfg:property:5}.

\ifx\longversion\x

\noindent
\begin{center}
  \textcolor{black}{--- See Appendix~\ref{sec:app:2} for additional
    comments. ---}
\end{center}

\fi

\begin{Example}
  \label{example:QGEChannel:1}

  As a particular example of a channel with a quantum state, we propose a
  quantum version of the classical Gilbert--Elliott channel, henceforth called
  the Quantum Gilbert--Elliott channel. We define this channel by specifying
  the NFG as in Fig.~\ref{fig:quantum:GE:channel:1}, which, upon
  closing-the-box results in a function node that can be used as
  $W^{(y_{\ell}|x_{\ell})}(s_{\ell-1}, s_{\ell}; s'_{\ell-1}, s'_{\ell})$ in
  Fig.~\ref{fig:QFMSC:high:level:1}. The NFG in
  Fig.~\ref{fig:quantum:GE:channel:1} stems from the following considerations.
  (Recall the communication setup from
  Section~\ref{sec:information:rate:channels:with:quantum:state:1}.)
  \begin{itemize}
    
  \item $\setX \defeq \{ 0, 1 \}$, $\setY \defeq \{ 0, 1 \}$, $\setS \defeq
    \setS' \defeq \{ 0, 1 \}$.

  \item The state quantum system is some qubit.

  \item The $\ell$-th transmit quantum system is some qubit.

  \item At time index $\ell$, Alice encodes $x_{\ell}$ into state $\rhoAxell$
    of the $\ell$-th transmit quantum system, where the matrix version of
    $\rhoAxell$ is a p.s.d.\ matrix with trace one. Specifically, for the
    communication setups in Section~\ref{sec:numerical:examples:1} we choose
    $\rhoA_0 = \bigl( \begin{smallmatrix} 1 & 0 \\ 0 & 0 \end{smallmatrix}
    \bigr)$ and $\rhoA_1 = \bigl( \begin{smallmatrix} 0 & 0 \\ 0 &
      1 \end{smallmatrix} \bigr)$.

  \item Alice sends the $\ell$-th transmit quantum system to Bob. On its way
    it interacts with the state quantum system. This interaction is described
    in terms of an operator-sum representation
    \cite[Chap.~8]{Nielsen:Chuang:00:1} based on matrices $E_{k_{\ell}}$,
    $k_{\ell} \in \{ 0, 1 \}$, where
    \begin{align*}
      E_0 
        &\! \defeq \! 
           \left[
           \begin{smallmatrix}
             \!\!\sqrt{1\!-\!\pgood}\!\! & 0                & 0   & 0 \\
             0                & \!\!\sqrt{1\!-\!\pgood}\!\! & 0   & 0 \\
             0                & 0                & \!\!\sqrt{1\!-\!\pbad}\!\! & 0 \\
             0                & 0                & 0   & \!\!\sqrt{1\!-\!\pbad}\!\!
           \end{smallmatrix}
           \right] \!\! , \  
      E_1 
         \! \defeq \! 
           \left[
           \begin{smallmatrix}
             0                  & \sqrt{\pgood} & 0     & 0 \\
             \sqrt{\pgood} & 0            & 0           & 0 \\
             0                  & 0            & 0   & \sqrt{\pbad} \\
             0                  & 0            & \sqrt{\pbad} & 0
           \end{smallmatrix}
           \right] \!\! .
    \end{align*}

  \item Bob performs a quantum measurement~\cite[Chap.~2]{Nielsen:Chuang:00:1}
    defined by measurement operators $\{ M_{y_{\ell}} \}_{y_{\ell} \in
      \setY}$ on the $\ell$-th transmit quantum system. Specifically, for the
    communication setups in Section~\ref{sec:numerical:examples:1} we choose
    $M_0 = \bigl( \begin{smallmatrix} 1 & 0 \\ 0 & 0 \end{smallmatrix} \bigr)$
    and $M_1 = \bigl( \begin{smallmatrix} 0 & 0 \\ 0 & 1 \end{smallmatrix}
    \bigr)$.

  \item Between two transmissions, the evolution of the state quantum system
    is described by a unitary matrix $U$ that is derived from the Hamiltonian
    $\mathcal{H}_{\mathrm{s}}$ and the time difference between two
    transmissions.

  \end{itemize}
  Note that for function nodes in the NFG in
  Fig.~\ref{fig:quantum:GE:channel:1} that were specified in terms of a
  matrix, we use a dot to denote the variable that corresponds to the row
  index of the matrix. (In the case of the function nodes $E_{k_{\ell}}$ and
  $E_{k_{\ell}}^{\Herm}$, two variables jointly correspond to the row index of
  the matrix.)

  \ifx\longversion\x

  \noindent
  \begin{center}
    \textcolor{black}{--- See Appendix~\ref{sec:app:3} for additional
      comments. ---}
  \end{center}

  \fi

\end{Example}

\newpage

\begin{figure}[!h]
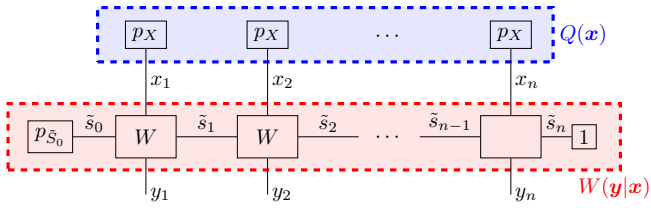

  \centering
  \figCFSMChighlevel
  \caption{Channel with a classical state: joint NFG for input process $Q(\vx)$
    (blue box) and channel law $W(\vy|\vx)$ (red box), after closing-the-box.}
  \label{fig:FMSC:high:level:1}
\end{figure}

\begin{figure}[!h]
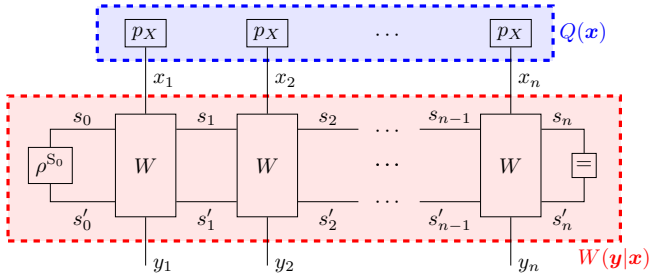

  \centering
  \figQFSMChighlevel
  \caption{Channel with a quantum state: joint NFG for input process $Q(\vx)$
    (blue box) and channel law $W(\vy|\vx)$ (red box), after closing-the-box.}
  \label{fig:QFMSC:high:level:1}
\end{figure}

\begin{figure}[!h]
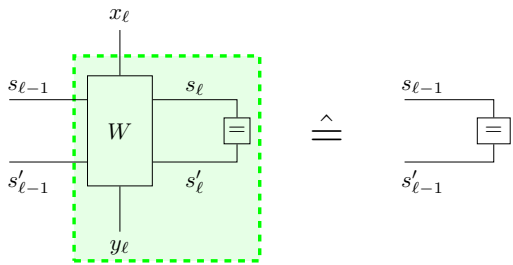

  \centering
  \figQFSMCcondition
  \caption{NFGs visualizing the
    constraint~\eqref{eq:quantum:channel:nfg:property:12} that $W$ has to
    satisfy. Namely, after applying the closing-the-box operation to the NFG
    on the left-hand side, one has to obtain the NFG on the right-hand
    side. (Note that on the right-hand side, the edge corresponding to
    $x_{\ell}$ has been omitted because it is not connected to any function
    node.)}
  \label{fig:quantum:channel:condition:1}
\end{figure}

\begin{figure}[!h]
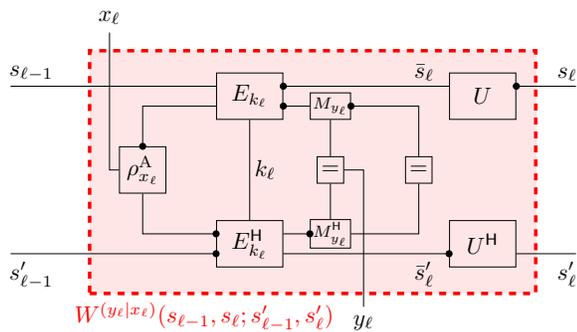

  \centering
  \figQGEchannel
  \caption{Internal details of $W^{(y_{\ell}|x_{\ell})}(s_{\ell-1}, s_{\ell}; 
    s'_{\ell-1}, s'_{\ell})$ for the Quantum Gilbert--Elliott Channel.}
  \label{fig:quantum:GE:channel:1}
\end{figure}

\newpage

\section{Information Rate Estimation}
\label{sec:information:rate:estimation:1}

Recall that the approach of~\cite{Arnold:Loeliger:Vontobel:Kavcic:Zeng:06:1}
for estimating information rates of FSMCs is based on the
Shannon--McMillan--Breiman theorem (see, \eg, \cite{Cover:Thomas:91}) and
suitable generalizations. Namely, the
information rate
\begin{align*}
  I(\sX;\sY) 
    &= H(\sX) + H(\sY) - H(\sXY) \\[-0.60cm]
\end{align*}
of a channel with a classical state can be estimated as follows:
\begin{enumerate}

\item Randomly generate a channel input sequence $\cvx = (\cx_1, \cx_2,
  \ldots, \cx_n)$ according to the law $Q$.

\item Based on this channel input sequence, randomly generate a channel output
  sequence $\cvy = (\cy_1, \cy_2, \ldots, \cy_n)$.

\item Estimate $H(\sX)$, $H(\sY)$, $H(\sXY)$ by computing $-\frac{1}{n}
  \log\bigl( g(\cvx) \bigr)$, $-\frac{1}{n} \log\bigl( g(\cvy) \bigr)$,
  $-\frac{1}{n} \log \bigl( g(\cvx,\cvy) \bigr)$, where $g(\cvx) =
  \sum_{\vy, \, \tvs} g(\cvx,\vy,\tvs)$, $g(\cvy) = \sum_{\vx, \, \tvs}
  g(\vx,\cvy,\tvs)$, $g(\cvx,\cvy) = \sum_{\tvs} g(\cvx,\cvy,\tvs)$.

\item Combine the above estimates to obtain an estimate of $I(\sX;\sY)$.

\end{enumerate}
Thanks to the close relationship between the NFG in
Fig.~\ref{fig:FMSC:high:level:1} and the NFG in
Fig.~\ref{fig:QFMSC:high:level:1}, it is \emph{formally} straightforward to
generalize the above procedure to channels with a quantum state. Namely, one
simply has to replace Step~3) by Step~3'), where
\begin{enumerate}

\item[3')] Estimate $H(\sX)$, $H(\sY)$, $H(\sXY)$ by computing $-\frac{1}{n}
  \log\bigl( g(\cvx) \bigr)$, $-\frac{1}{n} \log\bigl( g(\cvy) \bigr)$,
  $-\frac{1}{n} \log \bigl( g(\cvx,\cvy) \bigr)$, where $g(\cvx) =
  \sum_{\vy, \, \vs, \, \vs'} g(\cvx,\vy,\vs,\vs')$, $g(\cvy) = \sum_{\vx, \,
    \vs, \, \vs'} g(\vx,\cvy,\vs,\vs')$, $g(\cvx,\cvy) = \sum_{\vs, \, \vs'}
  g(\cvx,\cvy,\vs,\vs')$.

\end{enumerate}
In order to efficiently compute all the relevant quantities, one can apply
suitable closing-the-box operations as in~\cite{Loeliger:Vontobel:15:1:subm},
in particular as in Section~IV of~\cite{Loeliger:Vontobel:15:1:subm}. This is
equivalent to applying the sum-product algorithm on a modified version of the
underlying NFG, where edges are suitably merged so that the modified NFG does
not contain cycles and so that the computed marginals are exact.

\ifx\longversion\x

\vspace{-1.5mm}

\noindent
\begin{center}
  \textcolor{black}{--- See Appendix~\ref{sec:app:4} for additional
    comments. ---}
\end{center}

\vspace{-3.0mm}

\fi

\section{Numerical Examples}
\label{sec:numerical:examples:1}

In Figs.~\ref{fig:QGEC:plot:1}--\ref{fig:QGEC:plot:4}, we present some
numerical information rate (IR) estimates for various setups of the Quantum
Gilbert--Elliott channel where the channel input process is an i.i.d.\ process
with $p_X(0) = p_X(1) = 1/2$. (See the figure captions for further details.)
In Figs.~\ref{fig:QGEC:plot:1}--\ref{fig:QGEC:plot:4}, we also show some
auxiliary-channel-based information rate lower bound estimates that are based
on auxiliary channels with a classical
state~\cite{Arnold:Loeliger:Vontobel:Kavcic:Zeng:06:1}. These auxiliary
channels were optimized with the help of the techniques
in~\cite{Sadeghi:Vontobel:Shams:09:1}. Finally, Fig.~\ref{fig:QGEC:plot:2}
includes an auxiliary-channel-based information rate lower bound estimate that
is based on an auxiliary channel with a quantum state. As already emphasized
beforehand, these lower bounds represent rates that are achievable with the
help of a mismatched decoder~\cite{Ganti:Lapidoth:Telatar:00:1}.

\ifx\longversion\x

\vspace{-1.5mm}

\noindent
\begin{center}
  \textcolor{black}{--- See Appendix~\ref{sec:app:5} for additional
    comments. ---}
\end{center}

\vspace{-3.0mm}

\fi

\section*{Acknowledgment}
\label{sec:ack:1}

It is a great pleasure to acknowledge discussions on topics related to this
paper with Andi Loeliger.

\newpage

\begin{figure}[t]
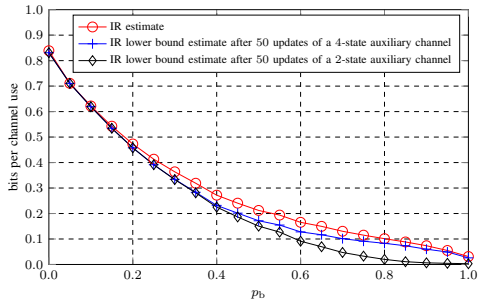

  \centering
  \figSimulationPlotOne
  \caption{Quantum Gilbert–Elliott Channel: $\pgood = 0.05$ is fixed; $\pbad$
    varies from $0$ to $1$; $U = \exp(-\i \alpha H)$, where $H$ is some fixed
    Hermitian matrix and where $\alpha = 1$ is fixed; $n = 10^5$.}
  \label{fig:QGEC:plot:1}
  \vspace{-0.25cm}
\end{figure}

\begin{figure}
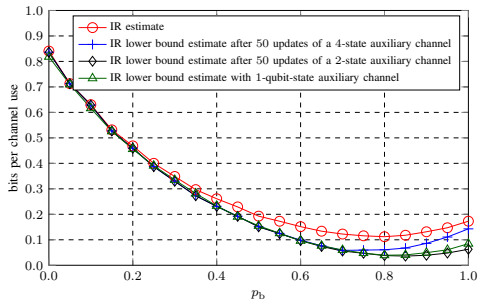

  \centering
  \vspace{0.6cm}
  \figSimulationPlotTwo
  \caption{Variant of the Quantum Gilbert–Elliott Channel where the state
    quantum system consists of two qubits whose evolution is described by $U$,
    but where only one of the qubits interacts directly with the transmit
    quantum system: $\pgood = 0.05$ is fixed; $\pbad$ varies from $0$ to $1$;
    $U = \exp(-\i \alpha H)$, where $H$ is some fixed Hermitian matrix and
    where $\alpha = 1.2$ is fixed; $n = 10^5$.}
  \label{fig:QGEC:plot:2}
  \vspace{-0.25cm}
\end{figure}

\begin{figure}[t]
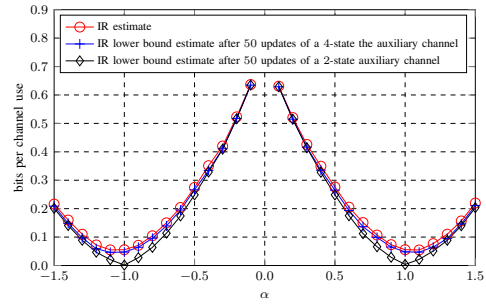

  \centering
  \figSimulationPlotThree
  \caption{Quantum Gilbert–Elliott Channel: $\pgood = 0.05$ is fixed; $\pbad =
    0.95$ is fixed; $U = \exp(-\i \alpha H)$, where $H$ is the same Hermitian
    matrix as in Fig.~\ref{fig:QGEC:plot:1} and where $\alpha$ varies from
    $-1.5$ to $+1.5$; $n = 10^5$. (No information rate estimates are included
    for $\alpha$ around $0$ because of slow mixing of the channel.)}
  \label{fig:QGEC:plot:3}
  \vspace{-0.25cm}
\end{figure}

\begin{figure}
  \centering
  \figSimulationPlotFour
  \caption{Same variant of the Quantum Gilbert–Elliott Channel as in
    Fig.~\ref{fig:QGEC:plot:2}: $\pgood = 0.05$ is fixed; $\pbad = 0.95$ is
    fixed; $U = \exp(-\i \alpha H)$, where $H$ is the same Hermitian matrix as
    in Fig.~\ref{fig:QGEC:plot:2} and where $\alpha$ varies from $-1.5$ to
    $+1.5$; $n = 10^5$. (No information rate estimates are included for
    $\alpha$ around $0$ because of slow mixing of the channel.)}
  \label{fig:QGEC:plot:4}
  \vspace{-0.25cm}
\end{figure}

\bibliographystyle{IEEEtran}
\bibliography{/home/vontobel/references/references}

\begin{thebibliography}{10}
\providecommand{\url}[1]{#1}
\csname url@samestyle\endcsname
\providecommand{\newblock}{\relax}
\providecommand{\bibinfo}[2]{#2}
\providecommand{\BIBentrySTDinterwordspacing}{\spaceskip=0pt\relax}
\providecommand{\BIBentryALTinterwordstretchfactor}{4}
\providecommand{\BIBentryALTinterwordspacing}{\spaceskip=\fontdimen2\font plus
\BIBentryALTinterwordstretchfactor\fontdimen3\font minus
  \fontdimen4\font\relax}
\providecommand{\BIBforeignlanguage}[2]{{%
\expandafter\ifx\csname l@#1\endcsname\relax
\typeout{** WARNING: IEEEtran.bst: No hyphenation pattern has been}%
\typeout{** loaded for the language `#1'. Using the pattern for}%
\typeout{** the default language instead.}%
\else
\language=\csname l@#1\endcsname
\fi
#2}}
\providecommand{\BIBdecl}{\relax}
\BIBdecl

\bibitem{Gallager:68}
R.~G. Gallager, \emph{{I}nformation {T}heory and {R}eliable
  {C}ommunication}.\hskip 1em plus 0.5em minus 0.4em\relax New York: Wiley,
  1968.

\bibitem{Arimoto:72}
S.~Arimoto, ``An algorithm for computing the capacity of arbitrary memoryless
  channels,'' \emph{IEEE Trans.\ Inf.\ Theory}, vol.~18, no.~1, pp. 14--20,
  Jan. 1972.

\bibitem{Blahut:72}
R.~E. Blahut, ``Computation of channel capacity and rate distortion
  functions,'' \emph{IEEE Trans.\ Inf.\ Theory}, vol.~18, no.~4, pp. 460--473,
  Jul. 1972.

\bibitem{Arnold:Loeliger:Vontobel:Kavcic:Zeng:06:1}
D.~M. Arnold, H.-A. Loeliger, P.~O. Vontobel, A.~Kav\v{c}i\'c, and W.~Zeng,
  ``Simulation-based computation of information rates for channels with
  memory,'' \emph{IEEE Trans.\ Inf.\ Theory}, vol.~52, no.~8, pp. 3498--3508,
  Aug. 2006.

\bibitem{Sharma:Singh:01:1}
V.~Sharma and S.~K. Singh, ``Entropy and channel capacity in the regenerative
  setup with applications to {M}arkov channels,'' in \emph{Proc.\ IEEE Int.\
  Symp.\ Inf.\ Theory}, Washington, D.C., June 24--29 2001, p. 283.

\bibitem{Pfister:Soriaga:Siegel:01:1}
H.~D. Pfister, J.~B. Soriaga, and P.~H. Siegel, ``On the achievable information
  rates of finite-state {ISI} channels,'' in \emph{Proc.\ IEEE Global
  Communications Conference}, San Antonio, TX, Nov. 2001, pp. 2992--2996.

\bibitem{Vontobel:Kavcic:Arnold:Loeliger:08:1}
P.~O. Vontobel, A.~Kav{\v c}i{\'c}, D.~M. Arnold, and H.-A. Loeliger, ``A
  generalization of the {B}lahut-{A}rimoto algorithm to finite-state
  channels,'' \emph{IEEE Trans.\ Inf.\ Theory}, vol.~54, no.~5, pp. 1887--1918,
  May 2008.

\bibitem{Mushkin:BarDavid:89:1}
M.~Mushkin and I.~Bar-David, ``Capacity and coding for the {G}ilbert-{E}lliott
  channel,'' \emph{IEEE Trans.\ Inf.\ Theory}, vol.~35, no.~6, pp. 1277--1290,
  Nov. 1989.

\bibitem{Sadeghi:Vontobel:Shams:09:1}
P.~Sadeghi, P.~O. Vontobel, and R.~Shams, ``Optimization of information rate
  upper and lower bounds for channels with memory,'' \emph{IEEE Trans.\ Inf.\
  Theory}, vol.~55, no.~2, pp. 663--688, Feb. 2009.

\bibitem{Ganti:Lapidoth:Telatar:00:1}
A.~Ganti, A.~Lapidoth, and {\. I}.~E. Telatar, ``Mismatched decoding revisited:
  general alphabets, channels with memory, and the wide-band limit,''
  \emph{IEEE Trans.\ Inf.\ Theory}, vol.~46, no.~7, pp. 2315--2328, 2000.

\bibitem{Haroche:Wineland:Nobelprize:2012:1}
{Royal Swedish Academy of Sciences}, ``Measuring and manipulating individual
  quantum systems,'' \emph{Scientific Background on the Nobel Prize in Physics
  awarded to Serge Haroche and David J.\ Wineland}, 2012,
  https://www.nobelprize.org/nobel\_prizes/physics/laureates/2012/
  advanced.html.

\bibitem{Nielsen:Chuang:00:1}
M.~A. Nielsen and I.~L. Chuang, \emph{Quantum Computation and Quantum
  Information}.\hskip 1em plus 0.5em minus 0.4em\relax Cambridge, UK: Cambridge
  University Press, 2000.

\bibitem{Kretschmann:Werner:05:1}
D.~Kretschmann and R.~F. Werner, ``Quantum channels with memory,'' \emph{Phys.\
  Rev.\ A}, vol.~72, no. 062323, pp. 1--19, 2005.

\bibitem{Caruso:Giovannetti:Lupo:Mancini:14:1}
F.~Caruso, V.~Giovannetti, C.~Lupo, and S.~Mancini, ``Quantum channels and
  memory effects,'' \emph{Rev.\ Mod.\ Phys.}, vol.~86, pp. 1203--1259,
  Oct.--Dec. 2014.

\bibitem{Kschischang:Frey:Loeliger:01}
F.~R. Kschischang, B.~J. Frey, and H.-A. Loeliger, ``Factor graphs and the
  sum-product algorithm,'' \emph{IEEE Trans.\ Inf.\ Theory}, vol.~47, no.~2,
  pp. 498--519, Feb. 2001.

\bibitem{Forney:01:1}
G.~D. {Forney, Jr.}, ``Codes on graphs: normal realizations,'' \emph{IEEE
  Trans.\ Inf.\ Theory}, vol.~47, no.~2, pp. 520--548, Feb. 2001.

\bibitem{Loeliger:04:1}
H.-A. Loeliger, ``An introduction to factor graphs,'' \emph{IEEE Sig.\ Proc.\
  Mag.}, vol.~21, no.~1, pp. 28--41, Jan. 2004.

\bibitem{Loeliger:Vontobel:12:1}
H.-A. Loeliger and P.~O. Vontobel, ``A factor-graph representation of
  probabilities in quantum mechanics,'' in \emph{Proc.\ IEEE Int.\ Symp.\ Inf.\
  Theory}, Cambridge, MA, USA, Jul.~1--6 2012, pp. 656--660.

\bibitem{Loeliger:Vontobel:15:1:subm}
------, ``Factor graphs for quantum probabilities,'' \emph{submitted to IEEE
  Trans.\ Inf.\ Theory, available online under \emph{\texttt{http://arxiv.org/
  abs/1508.00689}}}, Aug. 2015.

\bibitem{Wood:Biamonte:Cory:15:1}
C.~J. Wood, J.~D. Biamonte, and D.~G. Cory, ``Tensor networks and graphical
  calculus for open quantum systems,'' \emph{Quantum Inf.\ and Comp.}, vol.~15,
  no. 9--10, pp. 759--811, 2015.

\bibitem{Cover:Thomas:91}
T.~M. Cover and J.~A. Thomas, \emph{Elements of Information Theory}.\hskip 1em
  plus 0.5em minus 0.4em\relax New York: John Wiley \& Sons Inc., 1991.

\bibitem{Ephraim:Merhav:02:1}
Y.~Ephraim and N.~Merhav, ``Hidden {M}arkov processes,'' \emph{IEEE Trans.\
  Inf.\ Theory}, vol.~48, no.~6, pp. 1518--1569, Jun. 2002.

\end{thebibliography}

\ifx\longversion\x

\newpage

\appendices


%
%

%

%
%

\section*{Notation}
\label{sec:app:0}

In these appendices, we will use $\vx_{\ell_1}^{\ell_2}$ to denote the vector
$(x_{\ell_1}, \ldots, x_{\ell_2})$, $\vy_{\ell_1}^{\ell_2}$ to denote the
vector $(y_{\ell_1}, \ldots, y_{\ell_2})$, etc., where $\ell_1$ and $\ell_2$
are integers satisfying $\ell_1 \leq \ell_2$.

\section{Supplementary Notes for
               Section~\ref{sec:channel:with:classical:state:1}}
\label{sec:app:1}

The main purpose of this appendix is to prove
Lemma~\ref{lemma:classical:channel:valid:pmf:1} (see below) about the global
function of the NFG in Fig.~\ref{fig:FMSC:high:level:1}, which is associated
with a classical channel with memory.

Let $g(\vx_1^n, \vy_1^n, \tvs_0^n)$ be the global function of the NFG in
Fig.~\ref{fig:FMSC:high:level:1}, \ie,
\begin{align}
  g(\vx_1^n, \vy_1^n, \tvs_0^n)
    &\defeq
       p_{\tilde S_0}(\tilde s_0) 
       \cdot 
       \prod_{\ell=1}^{n}
         \Bigl(
           p_X(x_{\ell})
           \cdot
           W(\tilde s_{\ell}, y_{\ell} | \tilde s_{\ell-1}, x_{\ell})
         \Bigr) \ .
           \label{eq:classical:state:channel:global:function:1}
\end{align}
Moreover, let $Q(\vx_1^n)$ be obtained by a suitable closing-the-box operation
around parts of the NFG in Fig.~\ref{fig:FMSC:high:level:1} (see the blue box
in Fig.~\ref{fig:FMSC:high:level:1}), \ie,
\begin{align}
  Q(\vx_1^n)
    &\defeq 
       \prod_{\ell=1}^{n}
         p_X(x_{\ell}) \ ,
\end{align}
and let $W(\vy_1^n | \vx_1^n)$ be obtained by a suitable closing-the-box
operation around parts of the NFG in Fig.~\ref{fig:FMSC:high:level:1} (see the
red box in Fig.~\ref{fig:FMSC:high:level:1}), \ie,
\begin{align}
  W(\vy_1^n | \vx_1^n)
    &\defeq
       \sum_{\tvs_0^n}
         p_{\tilde S_0}(\tilde s_0) 
         \cdot 
         \prod_{\ell=1}^{n} 
           W(\tilde s_{\ell}, y_{\ell} | \tilde s_{\ell-1}, x_{\ell}) \ .
\end{align}
(Note that here the closing-the-box operation for the blue box is trivial in
the sense that there are no edges completely inside the blue box, and so there
are no variables to sum over.)

Recall that the factorization of $g(\vx_1^n, \vy_1^n, \tvs_0^n)$
in~\eqref{eq:classical:state:channel:global:function:1} follows from the
following assumptions on our source/channel model:
\begin{itemize}
  
\item The input process is an i.i.d.\ process.

\item Conditioned on $\tilde s_{\ell-1}$ and $x_{\ell}$, the channel state
  $\tilde s_{\ell}$ and channel output $y_{\ell}$ are conditionally
  independent of $\tvs_0^{\ell-2}$, $\vx_1^{\ell-1}$, and $\vy_1^{\ell-1}$.

\end{itemize}
For more context and further details, we refer to~\cite{Gallager:68}. (Note
that what we call finite-state-machine channels (FSMCs) are called
finite-state channels in~\cite{Gallager:68}.)

\begin{lemma}
  \label{lemma:classical:channel:valid:pmf:1}

  Assume that the channel model $W$ is such
  that~\eqref{eq:classical:W:condition:1}
  and~\eqref{eq:classical:W:condition:2} hold. Then the function $g(\vx_1^n,
  \vy_1^n, \tvs_0^n)$ is a pmf over $\vx_1^n$, $\vy_1^n$, and $\tvs_0^n$, \ie,
  \begin{align}
    \forall \ \vx_1^n, \vy_1^n, \tvs_0^n : \ \ 
      g(\vx_1^n, \vy_1^n, \tvs_0^n) &\geq 0 \ ,
        \label{eq:classical:state:channel:global:function:property:1} \\
      \sum_{\vx_1^n, \, \vy_1^n, \, \tvs_0^n} 
        g(\vx_1^n, \vy_1^n, \tvs_0^n) &= 1 \ .
        \label{eq:classical:state:channel:global:function:property:2}
  \end{align}
  Moreover, the function $Q(\vx_1^n)$ is a pmf over $\vx_1^n$, and the
  function $W(\vy_1^n | \vx_1^n)$ is a conditional pmf over $\vy_1^n$ given
  $\vx_1^n$.
\end{lemma}

\begin{proof}
  The fact that $g(\vx_1^n, \vy_1^n, \tvs_0^n) \geq 0$ for all $(\vx_1^n,
  \vy_1^n, \tvs_0^n)$ follows immediately
  from~\eqref{eq:classical:state:channel:global:function:1}
  and~\eqref{eq:classical:W:condition:1}. On the other hand, the fact that
  $\sum_{\vx_1^n, \, \vy_1^n, \, \tvs_0^n} g(\vx_1^n, \vy_1^n, \tvs_0^n) = 1$
  can be shown by using~\eqref{eq:classical:W:condition:2} repeatedly. Namely,
  \begin{align}
    &
    \sum_{\vx_1^n, \, \vy_1^n, \, \tvs_0^n}
      g(\vx_1^n, \vy_1^n, \tvs_0^n)
        \nonumber \\
      &\quad\quad
       = \sum_{\vx_1^n, \, \vy_1^n, \, \tvs_0^n}
           p_{\tilde S_0}(\tilde s_0) 
           \cdot 
           \prod_{\ell=1}^{n} 
             \Bigl(
               p_X(x_{\ell})
               \cdot 
               W(\tilde s_{\ell}, y_{\ell} | \tilde s_{\ell-1}, x_{\ell})
             \Bigr)
               \nonumber \\
      &\quad\quad
       = \!\!\!
         \sum_{\vx_1^{n-1}, \, \vy_1^{n-1}, \, \tvs_0^{n-1}}
         \!\!\!\!\!\!\!\!\!\!\!
           p_{\tilde S_0}(\tilde s_0) 
           \cdot 
           \prod_{\ell=1}^{n-1} 
             \Bigl(
               p_X(x_{\ell})
               \cdot
               W(\tilde s_{\ell}, y_{\ell} | \tilde s_{\ell-1}, x_{\ell})
             \Bigr)
               \nonumber \\[0.25cm]
      &\quad\quad\hspace{1.75cm}
           \cdot \,
           \underbrace{
             \sum_{x_n}
               p_X(x_n)
               \cdot
               \underbrace{
               \sum_{y_n}
                 \sum_{\tilde s_n}
                   W(\tilde s_n, y_n | \tilde s_{n-1}, x_n)
               }_{= \, 1}
             }_{= \, 1}
             \nonumber \\
      &\quad\quad
       = \!\!\!
         \sum_{\vx_1^{n-1}, \, \vy_1^{n-1}, \, \tvs_0^{n-1}} 
         \!\!\!\!\!\!\!\!\!\!\!
           p_{\tilde S_0}(\tilde s_0) 
           \cdot 
           \prod_{\ell=1}^{n-1} 
             \Bigl(
               p_X(x_{\ell})
               \cdot
               W(\tilde s_{\ell}, y_{\ell} | \tilde s_{\ell-1}, x_{\ell})
             \Bigr)
               \nonumber \\
      &\quad\quad
       = \ldots 
           \nonumber \\
      &\quad\quad
       = \sum_{\tilde s_0}
           p_{\tilde S_0}(\tilde s_0)
             \nonumber \\
      &\quad\quad
       = 1 \ .
           \label{fig:classical:channel:global:function:sum:1}
  \end{align}
  This computation is visualized in Fig.~\ref{fig:CFSM:closing:the:box} by
  applying suitable closing-the-box operations to the NFG in
  Fig.~\ref{fig:FMSC:high:level:1}.

  Showing that the function $Q(\vx_1^n)$ is a pmf over $\vx_1^n$ is
  straightforward, and showing that the function $W(\vy_1^n | \vx_1^n)$ is a
  conditional pmf over $\vy_1^n$ given $\vx_1^n$ can be done analogously to
  the above proof. We omit the details.
\end{proof}

\section{Supplementary Notes for
               Section~\ref{sec:channel:with:quantum:state:1}}
\label{sec:app:2}

The main purpose of this appendix is to prove
Lemma~\ref{lemma:quantum:channel:valid:g:1} (see below) about properties of
the global function associated with a quantum channel with memory as in
Fig.~\ref{fig:QFMSC:high:level:1}. Connections to expressions involving more
standard quantum information processing notation will be discussed in
Appendix~\ref{sec:app:3}.

Let $g(\vx_1^n, \vy_1^n, \vs_0^n, {\vs'}_0^n)$ be the global function of the
NFG in Fig.~\ref{fig:QFMSC:high:level:1}, \ie,
\begin{align}
  &
  g(\vx_1^n, \vy_1^n, \vs_0^n, {\vs'}_0^n)
    \nonumber \\
    &\quad\quad
     \defeq
        \rhoSt{0}(s_0, s'_0) 
         \nonumber \\
    &\quad\quad\quad\quad  
       \cdot \, 
       \prod_{\ell=1}^{n} 
         \left(
           p_X(x_{\ell})
           \cdot
           W^{(y_{\ell} | x_{\ell})}(s_{\ell-1}, s_{\ell}; s'_{\ell-1},
           s'_{\ell}) 
       \right)
         \nonumber \\
    &\quad\quad\quad\quad
         \cdot \, 
         \delta(s'_n, s_n) \ .
           \label{eq:quantum:state:channel:global:function:1}
\end{align}
Moreover, let $g(\vx_1^n, \vy_1^n)$ be obtained from $g(\vx_1^n, \vy_1^n,
\vs_0^n, {\vs'}_0^n)$ by summing over $\vs_0^n$ and ${\vs'}_0^n$, \ie,
\begin{align}
  g(\vx_1^n, \vy_1^n) 
    &\defeq
       \sum_{\vs_0^n, \, {\vs'}_0^n}
         g(\vx_1^n, \vy_1^n, \vs_0^n, {\vs'}_0^n) \ ,
           \label{eq:quantum:state:channel:global:function:marginal:1}
\end{align}
let $Q(\vx_1^n)$ be obtained by a suitable closing-the-box operation
around parts of the NFG in Fig.~\ref{fig:QFMSC:high:level:1} (see the blue box
in Fig.~\ref{fig:QFMSC:high:level:1}), \ie,
\begin{align}
  Q(\vx_1^n)
    &\defeq
       \prod_{\ell=1}^{n}
         p_X(x_{\ell}) \ ,
\end{align}
and let $W(\vy_1^n | \vx_1^n)$ be obtained by a suitable closing-the-box
operation around parts of the NFG in Fig.~\ref{fig:QFMSC:high:level:1} (see the
red box in Fig.~\ref{fig:QFMSC:high:level:1}), \ie,
\begin{align}
  W(\vy_1^n | \vx_1^n)
    &\defeq
       \sum_{\vs_0^n, \, {\vs'}_0^n}
         \rhoSt{0}(s_0, s'_0)
           \nonumber \\
    &\quad\quad
         \cdot \,
         \prod_{\ell=1}^{n}
           \Bigl(
             p_X(x_{\ell})
             \cdot
             W^{(y_{\ell} | x_{\ell})}(s_{\ell-1}, s_{\ell}; s'_{\ell-1}, s'_{\ell})
           \Bigr)
             \nonumber \\
    &\quad\quad
         \cdot \,
         \delta(s'_n, s_n) \ .
             \label{eq:quantum:state:channel:law:1}
\end{align}
(Note that here the closing-the-box operation for the blue box is trivial in
the sense that there are no full edges completely inside the blue box, and so
there are no variables to sum over.)

\begin{lemma}
  \label{lemma:quantum:channel:valid:g:1}

  Assume that the channel law $W$ is such
  that~\eqref{eq:quantum:channel:nfg:property:11}
  and~\eqref{eq:quantum:channel:nfg:property:12} hold. Then the function
  $g(\vx_1^n, \vy_1^n, \vs_0^n, {\vs'}_0^n)$
  satisfies~\eqref{eq:quantum:channel:nfg:property:1}--\eqref{eq:quantum:channel:nfg:property:5}. Moreover,
  the function $Q(\vx_1^n)$ is a pmf over $\vx_1^n$, and the function
  $W(\vy_1^n | \vx_1^n)$ is a conditional pmf over $\vy_1^n$ given $\vx_1^n$.
\end{lemma}

  \bigformulatop{21}{%
    \begin{align}
      \sum_{\vx_1^n, \, \vy_1^n, \, \vs_0^n, \, {\vs'}_0^n}
        g(\vx_1^n, \vy_1^n, \vs_0^n, {\vs'}_0^n)
        &= \sum_{\vx_1^n, \, \vy_1^n, \, \vs_0^n, \, {\vs'}_0^n}
             \rhoSt{0}(s_0, s'_0) 
             \cdot \, 
             \prod_{\ell=1}^{n} 
               \left(
                 p_X(x_{\ell})
                 \cdot
                 W^{(y_{\ell} | x_{\ell})}(s_{\ell-1}, s_{\ell}; s'_{\ell-1}, s'_{\ell})
              \right)
            \cdot \, 
            \delta(s'_n, s_n)
                \nonumber \\
        &= \sum_{\vx_1^{n-1}, \, \vy_1^{n-1}, \, \vs_0^{n-1}, \, {\vs'}_0^{n-1}}
             \rhoSt{0}(s_0, s'_0) 
             \cdot \, 
             \prod_{\ell=1}^{n-1} 
               \left(
                 p_X(x_{\ell})
                 \cdot
                 W^{(y_{\ell} | x_{\ell})}(s_{\ell-1}, s_{\ell}; s'_{\ell-1}, s'_{\ell})
               \right) 
               \nonumber \\
         &\quad\quad\hspace{2cm}
           \cdot \, 
           \underbrace{
             \sum_{x_n}
               p_X(x_n)
               \cdot \
               \underbrace{
                 \sum_{y_n, \, s_n, \, s'_n}
                   W^{(y_n | x_n)}(s_{n-1}, s_{n}; s'_{n-1}, s'_{n})
                   \cdot
                   \delta(s'_n, s_n)
                }_{= \, \delta(s'_{n-1}, s_{n-1})}
           }_{= \, \delta(s'_{n-1}, s_{n-1})} 
             \nonumber \\
        &= \sum_{\vx_1^{n-1}, \, \vy_1^{n-1}, \, \vs_0^{n-1}, \, {\vs'}_0^{n-1}}
             \!\!\!\!\!\!\!\!
             \rhoSt{0}(s_0, s'_0) 
             \cdot \, 
             \prod_{\ell=1}^{n-1} 
               \left(
                 p_X(x_{\ell})
                 \cdot
                 W^{(y_{\ell} | x_{\ell})}(s_{\ell-1}, s_{\ell}; s'_{\ell-1}, s'_{\ell})
               \right) 
             \cdot
             \delta(s'_{n-1}, s_{n-1})
               \nonumber \\
        &= \ldots
             \nonumber \\
        &= \sum_{s_0, \, s'_0}
             \rhoSt{0}(s_0, s'_0)
             \cdot
             \delta(s'_0, s_0)
               \nonumber \\
        &= 1 \ .
             \label{eq:quantum:channel:g:sum:property:1}
    \end{align}
  }

  \bigformulatop{22}{%
    \begin{align}
      \overline{g(\vx_1^n, \vy_1^n)}
        &= \overline{
             \sum_{\vs_0^n, \, {\vs'}_0^n}
               g(\vx_1^n, \vy_1^n, \vs_0^n, {\vs'}_0^n)
           }
             \nonumber \\
        &= \overline{
             \sum_{\vs_0^n, \, {\vs'}_0^n}
               \rhoSt{0}(s_0, s'_0) 
               \cdot \, 
               \prod_{\ell=1}^{n} 
                 \Bigl(
                   p_X(x_{\ell})
                   \cdot
                   W^{(y_{\ell} | x_{\ell})}(s_{\ell-1}, s_{\ell}; s'_{\ell-1}, s'_{\ell})
                  \Bigr)
                \cdot \, 
                \delta(s'_n, s_n)
           }
             \nonumber \\
        &= \sum_{\vs_0^n, \, {\vs'}_0^n}
             \overline{\rhoSt{0}(s_0, s'_0)}
             \cdot \, 
             \prod_{\ell=1}^{n} 
               \Bigl(
                 \overline{p_X(x_{\ell})}
                 \cdot
                 \overline{
                   W^{(y_{\ell} | x_{\ell})}
                     (s_{\ell-1}, s_{\ell}; s'_{\ell-1}, s'_{\ell})
                 }
               \Bigr)
             \cdot \, 
             \overline{\delta(s'_n, s_n)}
               \nonumber \\
        &\onestareq
           \sum_{\vs_0^n, \, {\vs'}_0^n}
             \rhoSt{0}(s'_0, s_0)
             \cdot \, 
             \prod_{\ell=1}^{n} 
               \Bigl(
                 p_X(x_{\ell})
                 \cdot
                 W^{(y_{\ell} | x_{\ell})}(s'_{\ell-1}, s'_{\ell}; s_{\ell-1}, s_{\ell})
               \Bigr)
             \cdot \, 
             \delta(s _n, s'_n)
               \nonumber \\
        &\twostarseq
           \sum_{\vs_0^n, \, {\vs'}_0^n}
             \rhoSt{0}(s_0, s'_0)
             \cdot \, 
             \prod_{\ell=1}^{n} 
               \Bigl(
                 p_X(x_{\ell})
                 \cdot
                 W^{(y_{\ell} | x_{\ell})}(s_{\ell-1}, s_{\ell}; s'_{\ell-1}, s'_{\ell})
               \Bigr)
             \cdot \, 
             \delta(s'_n, s_n) 
               \nonumber \\
        &= g(\vx_1^n, \vy_1^n) \ .
                \label{eq:quantum:channel:g:sum:marginal:property:1}
    \end{align}
  }

\begin{proof}
  We prove the first claim as follows.
  \begin{itemize}

  \item Property~\eqref{eq:quantum:channel:nfg:property:1} follows immediately
    from~\eqref{eq:quantum:state:channel:global:function:1} and the fact that
    all the factors appearing in this expression are complex-valued.

  \item Property~\eqref{eq:quantum:channel:nfg:property:2} follows
    from~\eqref{eq:quantum:state:channel:global:function:1} and the assumption
    that $W^{(y_{\ell} | x_{\ell})}$ is a p.s.d.~matrix for all $y_{\ell}$ and
    $x_{\ell}$, and with that a Hermitian matrix for all $y_{\ell}$ and
    $x_{\ell}$, \ie,
    \begin{align}
      W^{(y_{\ell} | x_{\ell})}(s_{\ell-1}, s_{\ell}; s'_{\ell-1}, s'_{\ell}) 
        &= \overline{
             W^{(y_{\ell} | x_{\ell})}(s'_{\ell-1}, s'_{\ell}; s_{\ell-1}, s_{\ell})
           }
    \end{align}
    for all $y_{\ell}$, $x_{\ell}$, $s_{\ell-1}$, $s_{\ell}$, $s'_{\ell-1}$,
    and $s'_{\ell}$. Moreover, one uses the real-valuedness of the
    Kronecker-delta function and its symmetry in the arguments.
  
  \item Property~\eqref{eq:quantum:channel:nfg:property:3} can be shown by
    using~\eqref{eq:quantum:channel:nfg:property:12} repeatedly, see the
    derivation in Eq.~\eqref{eq:quantum:channel:g:sum:property:1} at the top
    of the next page. This computation is visualized in
    Fig.~\ref{fig:QFSM:closing:the:box} by applying suitable closing-the-box
    operations to the NFG in Fig.~\ref{fig:QFMSC:high:level:1}.
  
  \setcounter{equation}{22}

  \item Property~\eqref{eq:quantum:channel:nfg:property:4} follows immediately
    from proving $\overline{g(\vx_1^n, \vy_1^n)} = g(\vx_1^n, \vy_1^n)$ for
    all $\vx_1^n$ and $\vy_1^n$, see the derivation in
    Eq.~\eqref{eq:quantum:channel:g:sum:marginal:property:1} in the middle of
    the next page. There, Step~$\onestar$ follows from the p.s.d.\ property
    of $\rhoSt{0}$ and $W^{(y_{\ell} | x_{\ell})}$, along with
    the real-valuedness of the Kronecker-delta function and its symmetry in
    the arguments. Moreover, Step~$\twostars$ follows from relabeling the
    summation variables, \ie, $\vs_0^n$ becomes ${\vs'}_0^n$ and vice-versa.

  \setcounter{equation}{23}

  \item Property~\eqref{eq:quantum:channel:nfg:property:5} follows immediately
    from~\eqref{eq:quantum:state:channel:global:function:marginal:1}
    and~\eqref{eq:quantum:channel:nfg:property:3}.

  \end{itemize}

  Showing that the function $Q(\vx_1^n)$ is a pmf over $\vx_1^n$ is
  straightforward, and showing that the function $W(\vy_1^n | \vx_1^n)$ is a
  conditional pmf over $\vy_1^n$ given $\vx_1^n$ can be done analogously to
  the above proof. We omit the details.
\end{proof}

\begin{figure}[t]
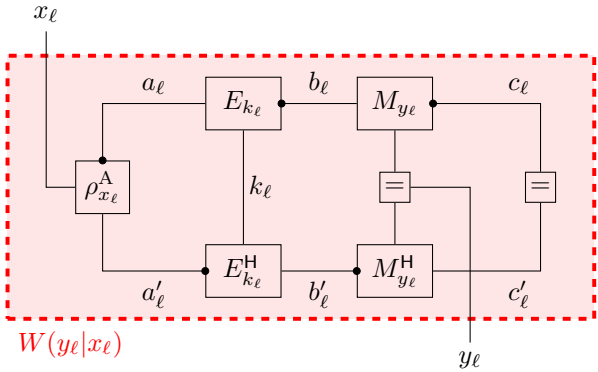

  \centering
  \figMemorylessQuantumChannel
  \caption{Classical communication over a memoryless quantum channel.}
  \label{fig:memoryless:quantum:channel:1}
\end{figure}

\section{Supplementary Notes for
                Example~\ref{example:QGEChannel:1}}
\label{sec:app:3}

The main purpose of this appendix is to give some more details w.r.t.\
Example~\ref{example:QGEChannel:1}. We do this by first discussing quantum
channels without memory and then quantum channels with memory. This appendix
should also help making the transition between standard quantum information
processing notation and our NFG representations.

Let us emphasize that the exact details of Example~\ref{example:QGEChannel:1}
are \emph{not} important. What is important is the framework that allows us to
deal with this type of channels.

\subsection{Classical Communication over a Memoryless Quantum Channel}

Alice wants to communicate some classical information to Bob. To that end, for
time indices $\ell = 1, \ldots, n$, they can use the following quantum channel
characterized by a completely-positive trace-preserving (CPTP) map
\begin{alignat*}{2}
  \Phi_{\ell} : \ \ 
    && \DensOp\bigl( \HAt{\ell} \bigr)
       &\to
       \DensOp\bigl( \HBt{\ell} \bigr) \\
    && \rhoAt{\ell}
       &\mapsto
       \rhoBt{\ell}
\end{alignat*}
The following objects appear in this expression:
\begin{itemize}

\item $\HAt{\ell}$ is a Hilbert space on Alice's side.

\item $\HBt{\ell}$ is a Hilbert space on Bob's side.

\item $\DensOp\bigl( \HAt{\ell} \bigr)$ is the set of density operators defined
  on $\HAt{\ell}$.

\item $\DensOp\bigl( \HBt{\ell} \bigr)$ is the set of density operators defined
  on $\HBt{\ell}$.

\end{itemize}
We make the following assumptions:
\begin{itemize}

\item All Hilbert spaces are finite dimensional.

\item In order to be specific, the mapping $\Phi_{\ell}$ is defined via Kraus
  operators $\{ E_{k_{\ell}} \}_{k_{\ell}}$, \ie,
  \begin{align}
    \Phi_{\ell}\bigl( \rhoAt{\ell} \bigr)
      &\defeq
         \sum_{k_{\ell}}
           E_{k_{\ell}}
           \,
           \rhoAt{\ell}
           \,
           E^\Herm_{k_{\ell}} \ .
  \end{align}
  Note that the operators $\{ E_{k_{\ell}} \}_{k_{\ell}}$ have to satisfy the
  condition $\sum_{k_{\ell}} E^\Herm_{k_{\ell}} \, E_{k_{\ell}} = I$, where
  $I$ is an identity matrix of suitable size.

\item Alice can prepare quantum states in $\HAt{\ell}$ described by density
  operators $\{ \rhoAxell \}_{x_{\ell} \in \set{X}}$.

\item Bob can make a quantum measurement on $\DensOp\bigl( \HBt{\ell} \bigr)$
  described by the measurement operators $\{ M_{y_{\ell}} \}_{y_{\ell} \in
    \set{Y}}$. Specifically, for $\rhoBt{\ell} \in \DensOp\bigl( \HBt{\ell}
  \bigr)$, the measurement outcome is $y_{\ell} \in \set{Y}$ with probability
  \begin{align}
    p_{Y_{\ell}}(y_{\ell})
      &= \Tr\bigl( M_{y_{\ell}} \, \rhoBt{\ell} \, M^\Herm_{y_{\ell}} \bigr) \ .
  \end{align}
  Note that the operators $\{ M_{y} \}_{y_{\ell} \in \set{Y}}$ have to satisfy
  the condition $\sum_{y_{\ell}} M^\Herm_{y_{\ell}} \, M_{y_{\ell}} = I$.

\end{itemize}
For further details about CPTP maps and measurement operators, see, \eg,
\cite{Nielsen:Chuang:00:1}.

Alice and Bob use $n$ independent instantiations of this channel to transmit
classical information as follows.
\begin{itemize}

\item Alice uses a classical code to encode her information word $\vu = (u_1,
  u_2, \ldots, u_k) \in \setU^k$ into a codeword $\vx = (x_1, x_2, \ldots,
  x_n) \in \setX^n$.

\item At time instance $\ell$, Alice transmits $\rhoAt{\ell} =
  \rhoAxell$ via the $\ell$-th instantiation of the memoryless quantum
  channel to Bob.

\item Bob makes a quantum measurement on $\rhoBt{\ell} \defeq
  \Phi_{\ell}(\rhoAt{\ell})$ described by the measurement operators $\{
  M_{y_{\ell}} \}_{y_{\ell} \in \set{Y}}$. The measurement outcome is called
  $y_{\ell}$.

\item Bob decodes $\vy = (y_1, y_2, \ldots, y_n) \in \setY^n$ toward obtaining
  an estimate $\hvu$ of $\vu$.

\end{itemize}
We emphasize that in our setup, the operators $\{ \rhoA_{x_{\ell}}
\}_{x_{\ell} \in \set{X}}$ and $\{ M_{y_{\ell}} \}_{y_{\ell} \in \set{Y}}$ are
given, \ie, they cannot be chosen by Alice and Bob, respectively.

Let $W(y_{\ell}|x_{\ell}) = p_{Y_{\ell}|X_{\ell}}(y_{\ell}|x_{\ell})$ be the
channel law, \ie, the probability of $y_{\ell}$ given $x_{\ell}$. We obtain
\begin{align}
  W(y_{\ell}|x_{\ell})
    &= \Tr
         \left(
           M_{y_{\ell}} 
           \,
           \left(
             \sum_{k_{\ell}}
               E_{k_{\ell}}
               \,
               \rhoAxell
               \,
               E^\Herm_{k_{\ell}}
           \right)
           \, 
           M^\Herm_{y_{\ell}}
         \right) 
           \nonumber \\
    &= \sum_{k_{\ell}}
         \Tr
         \left(
           M_{y_{\ell}}
           \,
           E_{k_{\ell}}
           \,
           \rhoAxell
           \,
           E^\Herm_{k_{\ell}}
           \,
           M^\Herm_{y_{\ell}}
         \right) \ .
\end{align}
Introducing suitable orthonormal bases to express the operators, we can write
this as
\begin{align}
  &
  W(y_{\ell}|x_{\ell}) 
    \nonumber \\
    &\quad
     = \sum_{k_{\ell}}
         \sum_{a_{\ell}, \, a'_{\ell}}
           \sum_{b_{\ell}, \, b'_{\ell}}
             \sum_{c_{\ell}, \, c'_{\ell}}
               M_{y_{\ell}}(c_{\ell}, b_{\ell})
               \cdot
               E_{k_{\ell}}(b_{\ell}, a_{\ell}) 
                 \nonumber \\
    &\quad
     \hspace{1.0cm}
               \cdot
               \rhoAxell(a_{\ell},a'_{\ell})
               \cdot
               E^\Herm_{k_{\ell}}(a'_{\ell},b'_{\ell})
               \cdot
               M^\Herm_{y_{\ell}}(b'_{\ell}, c'_{\ell})
               \cdot
               \delta(c'_{\ell}, c_{\ell}) \ . 
\end{align}

These calculations can be visualized with the help of the NFG in
Fig.~\ref{fig:memoryless:quantum:channel:1}. Namely, the global function is
\begin{align}
  &
  g(x_{\ell}, y_{\ell}, k_{\ell}, 
    a_{\ell}, a'_{\ell}, b_{\ell}, b'_{\ell}, c_{\ell}, c'_{\ell})
    \nonumber \\
    &\quad
     = M_{y_{\ell}}(c_{\ell}, b_{\ell})
               \cdot
               E_{k_{\ell}}(b_{\ell}, a_{\ell}) 
               \cdot
               \rhoAxell(a_{\ell},a'_{\ell})
                 \nonumber \\
    &\quad
     \hspace{1.0cm}
               \cdot
               E^\Herm_{k_{\ell}}(a'_{\ell},b'_{\ell})
               \cdot
               M^\Herm_{y_{\ell}}(b'_{\ell}, c'_{\ell})
               \cdot
               \delta(c'_{\ell}, c_{\ell}) \ ,
\end{align}
and the above-mentioned function $W(y_{\ell}|x_{\ell})$ is obtained by a
suitably closing-the-box operation (see the red box in
Fig.~\ref{fig:memoryless:quantum:channel:1}), where we sum over all variables
associated with edges that are completely inside the box.

Finally, note that the channel law of $n$ independent instantiations of this
channel is given by
\begin{align}
  W(\vy_1^n | \vx_1^n)
    &= \prod_{\ell=1}^{n}
         W(y_{\ell} | x_{\ell}) \ .
\end{align}

\subsection{Classical Communication over a Quantum Channel with Memory}

Having discussed quantum channels without memory, we now turn our attention to
quantum channels with memory. Alice wants again to communicate some classical
information to Bob. For time indices $\ell = 1, \ldots, n$, they can use the
following quantum channel characterized by a CPTP map
\begin{alignat*}{2}
  \Phi_{\ell} : \ \ 
    && \DensOp\bigl( \HAt{\ell} \otimes \HSt{\ell-1} \bigr)
       &\to
       \DensOp\bigl( \HBt{\ell} \otimes \HSt{\ell} \bigr) \\
    && \rhoAtStp{\ell}
       &\mapsto
       \rhoBtSt{\ell}
\end{alignat*}
The following objects appear in this expression:
\begin{itemize}

\item $\HAt{\ell}$ is a Hilbert space on Alice's side.

\item $\HBt{\ell}$ is a Hilbert space on Bob's side.

\item $\HSt{\ell-1}$ is the Hilbert space relevant for the memory of the
  channel at time index $\ell-1$.

\item $\HSt{\ell}$ is the Hilbert space relevant for the memory of the channel
  at time index $\ell$.

\item $\DensOp\bigl( \HAt{\ell} \otimes \HSt{\ell-1} \bigr)$ is the set of
  density operators defined on $\HAt{\ell} \otimes \HSt{\ell-1}$.

\item $\DensOp\bigl( \HBt{\ell} \otimes \HSt{\ell} \bigr)$ is the set of
  density operators defined on $\HBt{\ell} \otimes \HSt{\ell}$.

\end{itemize}

We make the following assumptions:
\begin{itemize}

\item All Hilbert spaces are finite dimensional.

\item In order to be specific, the mapping $\Phi_{\ell}$ is defined via Kraus
  operators $\{ E_{k_{\ell}} \}_{k_{\ell}}$, \ie,
  \begin{align}
    \Phi_{\ell}\bigl( \rhoAtStp{\ell} \bigr)
      &\defeq
         \sum_{k_{\ell}}
           E_{k_{\ell}}
           \,
           \rhoAtStp{\ell}
           \,
           E^\Herm_{k_{\ell}} \ .
  \end{align}
  Note that the operators $\{ E_{k_{\ell}} \}_{k_{\ell}}$ have to satisfy the
  condition $\sum_{k_{\ell}} E^\Herm_{k_{\ell}} \, E_{k_{\ell}} = I$.

\item Alice can prepare quantum states in $\HAt{\ell}$ described by density
  operators $\{ \rhoAxell \}_{x_{\ell} \in \set{X}}$. We assume that,
  given $x_{\ell}$, what Alice does is independent of the state of the channel
  at time index $\ell-1$, \ie, $\rhoAtStp{\ell} = \rhoAxell \otimes
  \rhoSt{\ell-1}$.

\item Bob can make a quantum measurement on $\DensOp\bigl( \HBt{\ell} \bigr)$
  described by the measurement operators $\{ M_{y_{\ell}} \}_{y_{\ell} \in
    \set{Y}}$. Specifically, for $\rhoBt{\ell} \in \DensOp\bigl( \HBt{\ell}
  \bigr)$, the measurement outcome is $y_{\ell} \in \set{Y}$ with probability
  \begin{align}
    p_{Y_{\ell}}(y_{\ell})
      &= \Tr
           \Bigl( 
             M_{y_{\ell}}
             \, 
             \rhoBt{\ell}
             \, 
             M^\Herm_{y_{\ell}}
           \Bigr) \ .
  \end{align}
  Note that the operators $\{ M_{y} \}_{y_{\ell} \in \set{Y}}$ have to satisfy
  the condition $\sum_{y_{\ell}} M^\Herm_{y_{\ell}} \, M_{y_{\ell}} = I$.

\end{itemize}
For further details about quantum channels with memory we refer to the survey
papers by Kretschmann and Werner~\cite{Kretschmann:Werner:05:1} and by Caruso
\etal~\cite{Caruso:Giovannetti:Lupo:Mancini:14:1}.

\begin{figure}
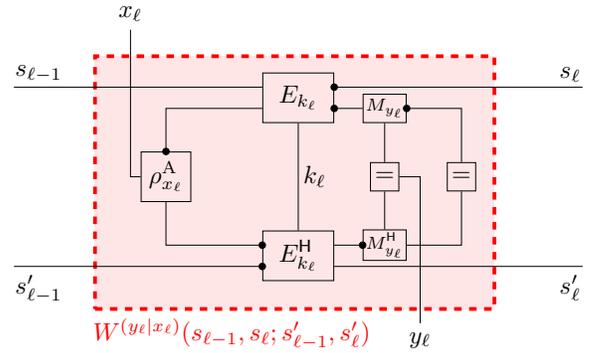

  \centering
  \figQuantumStateChannel
  \caption{Classical communication over a quantum channel with memory:
    internal details of $W^{(y_{\ell}|x_{\ell})}(s_{\ell-1}, s_{\ell};
    s'_{\ell-1}, s'_{\ell})$
    in~\eqref{eq:quantum:state:channel:W:section:1}.}
  \label{fig:quantum:state:channel:1}
\end{figure}

Alice and Bob use $n$ instantiations of this channel to transmit classical
information as follows:
\begin{itemize}

\item Alice uses a classical code to encode her information word $\vu = (u_1,
  u_2, \ldots, u_k) \in \setU^k$ into a codeword $\vx = (x_1, x_2, \ldots,
  x_n) \in \setX^n$.

\item At time instance $\ell$, Alice transmits $\rhoAt{\ell} =
  \rhoAxell$ via the $\ell$-th instantiation of the quantum channel to
  Bob.

\item Bob makes a quantum measurement on
  \begin{align}
    \rhoBt{\ell}
      &\defeq 
         \TrSt{\ell}
           \bigl( 
             \rhoBtSt{\ell}
           \bigr) 
       \defeq
         \Tr_{\mathrm{S}_{\ell}} 
          \bigl(
            \Phi_{\ell}(\rhoAtStp{\ell})
          \bigr)
  \end{align}
  described by the measurement operators $\{ M_{y_{\ell}} \}_{y_{\ell} \in
    \set{Y}}$. The measurement outcome is called $y_{\ell}$.

\item Bob decodes $\vy = (y_1, y_2, \ldots, y_n) \in \setY^n$ toward obtaining
  an estimate $\hvu$ of $\vu$.

\end{itemize}
We emphasize that in our setup, the operators $\{ \rhoA_{x_{\ell}}
\}_{x_{\ell} \in \set{X}}$ and $\{ M_{y_{\ell}} \}_{y_{\ell} \in \set{Y}}$ are
given, \ie, they cannot be chosen by Alice and Bob, respectively.

With this, the probability of receiving $y_{\ell}$, given that $x_{\ell}$ was
sent and given that the channel state at time $\ell-1$ is known to be
$\rhoSt{\ell-1}$, equals\footnote{Note that what $\rhoSt{\ell-1}$ is, depends
  on the knowledge/ignorance of components of $\vx_1^{\ell-1}$ and
  $\vy_1^{\ell-1}$. For an example, see later parts of this appendix.}
\begin{align}
  \Tr
    \left(
      M_{y_{\ell}} 
      \,
      \TrSt{\ell}
      \biggl(
        \sum_{k_{\ell}}
          E_{k_{\ell}}
          \,
          \bigl( \rhoAxell \otimes \rhoSt{\ell-1} \bigr)
          \,
          E^\Herm_{k_{\ell}}
      \biggr)
      \, 
      M^\Herm_{y_{\ell}}
    \right) ,
      \label{eq:quantum:state:channel:prob:y:1}
\end{align}
which can also be written as
\begin{align}
  \Tr
    \left(
      (M_{y_{\ell}} \! \otimes \! \ISt{\ell})
      \,
      \biggl(
        \sum_{k_{\ell}}
          E_{k_{\ell}}
          \,
          \bigl( \rhoAxell \! \otimes \! \rhoSt{\ell-1} \bigr)
          \,
          E^\Herm_{k_{\ell}}
      \biggr)
      \, 
      (M^\Herm_{y_{\ell}} \otimes \ISt{\ell})
    \right) . 
      \label{eq:quantum:state:channel:prob:y:2}
\end{align}
Moreover, assuming that $x_{\ell}$ was sent, that $y_{\ell}$ was observed, and
that the channel state at time $\ell-1$ is known to be $\rhoSt{\ell-1}$, the
channel state at time index $\ell$ is given by
\begin{align}
  \frac
  {
    \TrBt{\ell}
      \Bigl(
        (M_{y_{\ell}} \! \otimes \! \ISt{\ell})
        \,
        \left(
          \sum_{k_{\ell}}
            E_{k_{\ell}}
            \,
            \bigl( \rhoAxell \! \otimes \! \rhoSt{\ell-1} \bigr)
            \,
            E^\Herm_{k_{\ell}}
        \right)
        \, 
        (M^\Herm_{y_{\ell}} \! \otimes \! \ISt{\ell})
      \Bigr)
  }
  {
    \Tr
      \Bigl(
        (M_{y_{\ell}} \! \otimes \! \ISt{\ell})
        \,
        \left(
          \sum_{k_{\ell}}
            E_{k_{\ell}}
            \,
            \bigl( \rhoAxell \! \otimes \! \rhoSt{\ell-1} \bigr)
            \,
            E^\Herm_{k_{\ell}}
        \right)
        \, 
        (M^\Herm_{y_{\ell}} \! \otimes \! \ISt{\ell})
      \Bigr)
  } .
    \label{eq:quantum:state:channel:next:state:1}
\end{align}
Note that the denominator in~\eqref{eq:quantum:state:channel:next:state:1}
equals the expressions in~\eqref{eq:quantum:state:channel:prob:y:1}
and~\eqref{eq:quantum:state:channel:prob:y:2}.

In order to obtain an NFG representation of the setup in this section, we
introduce
\begin{align}
  &
  W^{(y_{\ell}|x_{\ell})}(s_{\ell-1},s_{\ell}; s'_{\ell-1},s'_{\ell})
    \nonumber \\
    &\quad
     \defeq
     \sum_{k_{\ell}}
         \sum_{a_{\ell}, \, a'_{\ell}}
           \sum_{b_{\ell}, \, b'_{\ell}}
             \sum_{c_{\ell}, \, c'_{\ell}}
               M_{y_{\ell}}(c_{\ell}, b_{\ell})
               \cdot
               E_{k_{\ell}}\bigl( (b_{\ell}, s_{\ell}), (a_{\ell}, s_{\ell-1}) \bigr) 
                 \nonumber \\
    &\quad
     \hspace{1.0cm}
               \cdot
               \rhoAxell(a_{\ell},a'_{\ell})
               \cdot
               E^\Herm_{k_{\ell}}
                 \bigl( (a'_{\ell}, s'_{\ell-1}), (b'_{\ell}, s'_{\ell}) \bigr)
               \cdot
               M^\Herm_{y_{\ell}}(b'_{\ell}, c'_{\ell})
                 \nonumber \\
    &\quad
     \hspace{1.0cm}
               \cdot
               \delta(c'_{\ell}, c_{\ell}) \ .
       \label{eq:quantum:state:channel:W:section:1}
\end{align}
Note that this function is obtained by a suitably closing-the-box operation
(see the red box in Fig.~\ref{fig:quantum:state:channel:1}), where we sum over
all variables associated with edges that are completely inside the box. (Note
that the edge labels inside the box are analogous to the edge labels in
Fig.~\ref{fig:memoryless:quantum:channel:1}. However, for simplicity, we have
omitted most of them.)

On the side, we note that the corresponding NFG for the Quantum
Gilbert--Elliott Channel in Fig.~\ref{fig:quantum:GE:channel:1} contains
additional $U$-boxes representing a unitary evolution of the channel state. By
redefining $E_{k_{\ell}}$ to $(I^{\mathrm{B}_{\ell}} \otimes U) \,
E_{k_{\ell}}$, the NFG in Fig.~\ref{fig:quantum:GE:channel:1} could be brought
into the form of the NFG in
Fig.~\ref{fig:quantum:state:channel:1}.\footnote{Observe that the redefined
  $E_{k_{\ell}}$ still satisfies $\sum_{k_{\ell}} E^\Herm_{k_{\ell}} \,
  E_{k_{\ell}} = I$.}

With the function $W^{(y_{\ell}|x_{\ell})}(s_{\ell-1},s_{\ell};
s'_{\ell-1},s'_{\ell})$ in hand, along with the corresponding (partial) NFG in
Fig.~\ref{fig:quantum:state:channel:1}, we can define an NFG for the overall
setup as in Fig.~\ref{fig:QFMSC:high:level:1}. Note that the boxes labeled $W$
represent the function $W^{(y_{\ell}|x_{\ell})}(s_{\ell-1},s_{\ell};
s'_{\ell-1},s'_{\ell})$ for $\ell = 1, \ldots, n$. Because the function
$W^{(y_{\ell}|x_{\ell})}(s_{\ell-1},s_{\ell}; s'_{\ell-1},s'_{\ell})$
satisfies ~\eqref{eq:quantum:channel:nfg:property:11}
and~\eqref{eq:quantum:channel:nfg:property:12}, the global function $g$ of the
NFG in Fig.~\ref{fig:QFMSC:high:level:1}
satisfies~\eqref{eq:quantum:channel:nfg:property:1}--\eqref{eq:quantum:channel:nfg:property:5}. (See
Appendix~\ref{sec:app:2} for details.)

All probabilities and density operators of interest can be obtained by
suitably summing over variables of the global function of the NFG in
Fig.~\ref{fig:QFMSC:high:level:1}. For example, for fixed $\vy_1^{\ell-1} =
\cvy_1^{\ell-1}$ and $x_{\ell} = \cx_{\ell}$, the probability
$p_{Y_{\ell}|X_{\ell},\vY_1^{\ell-1}} (y_{\ell} | \cx_{\ell},
\cvy_1^{\ell-1})$ can be obtained as follows:
\begin{align}
  &
  p_{Y_{\ell}|X_{\ell},\vY_1^{\ell-1}}(y_{\ell}|\cx_{\ell}, \cvy_1^{\ell-1})
    \nonumber \\
    &\quad\quad
     \propto
       p_{X_{\ell}, \vY_1^{\ell-1}, Y_{\ell}}
         (\cx_{\ell}, \cvy_1^{\ell-1}, y_{\ell})
           \nonumber \\
    &\quad\quad
     = \sum_{\vx_1^{\ell-1}, \vx_{\ell+1}^n, \, \vy_{\ell+1}^n, \, \vs_0^n, \, {\vs'}_0^n}
         \!\!\!\!\!\!\!\!\!\!\!\!\!
         g(\vx_1^{\ell-1}, \cx_{\ell}, \vx_{\ell+1}^n, 
           \cvy_1^{\ell-1}, \vy_{\ell}^n,
           \vs_0^n, {\vs'}_0^n) \ .
             \label{eq:quantum:state:channel:marginal:example:1}
\end{align}
Here, the proportionality constant is chosen such that the left-hand side is a
valid conditional pmf. The computation of this function via closing-the-box
operations is visualized in
Fig.~\ref{fig:CFSM:channel:simulation:Y}. Some comments:
\begin{itemize}

\item Applying the closing-the-box operation to the magenta box results in the
  function
  $\sigmaSt{\ell-1}_{|\cvy_{1}^{\ell-1}}(s_{\ell-1},s'_{\ell-1})$.\footnote{The
  subscript ``$|\cvy_{1}^{\ell-1}$'' emphasizes that
  $\sigmaSt{\ell-1}_{|\cvy_{1}^{\ell-1}}(s_{\ell-1},s'_{\ell-1})$ is based on
  the knowledge of $\vy_1^{\ell-1} = \cvy_1^{\ell-1}$.}

\item Applying the closing-the-box operation to the green box results in the
  function $\delta(s'_{\ell},s_{\ell})$, \ie, a degree-$2$ equality function
  node.

\item Applying the closing-the-box operation to the yellow box results in the
  function $p_{X_{\ell}, \vY_1^{\ell-1}, Y_{\ell}} (\cx_{\ell},
  \cvy_1^{\ell-1}, y_{\ell})$, from which the desired function
  $p_{Y_{\ell}|X_{\ell},\vY_1^{\ell-1}}(y_{\ell}|\cx_{\ell}, \cvy_1^{\ell-1})$
  can be easily obtained by normalization. 

  Note that, mathematically, applying the closing-the-box operation to the
  yellow box gives the following function (with argument $y_{\ell}$)
  \begin{align}
    &
    \sum_{s_{\ell-1}, \, s'_{\ell-1}, \, s_{\ell}, \, s'_{\ell}}
      \sigmaSt{\ell-1}_{|\cvy_{1}^{\ell-1}}(s_{\ell-1},s'_{\ell-1}) \cdot
      p_X(\cx_{\ell})
        \nonumber \\
    &\quad\quad\quad\quad
      \cdot
      W^{(y_{\ell}|\cx_{\ell})}(s_{\ell-1},s_{\ell}; s'_{\ell-1},s'_{\ell})
      \cdot
      \delta(s'_{\ell}, s_{\ell}) \ .
        \label{eq:quantum:state:channel:some:prob:1}
  \end{align}

\item If $\rhoSt{\ell-1} = \sigmaSt{\ell-1}_{|\cvy_{1}^{\ell-1}}$, then the
  expressions in~\eqref{eq:quantum:state:channel:prob:y:1}
  and~\eqref{eq:quantum:state:channel:prob:y:2} equal
  $p_{Y_{\ell}|X_{\ell},\vY_1^{\ell-1}} (y_{\ell} | \cx_{\ell},
  \cvy_1^{\ell-1})$. This connection between the NFG approach and the standard
  quantum information processing notation can be established by
  inserting~\eqref{eq:quantum:state:channel:W:section:1}
  into~\eqref{eq:quantum:state:channel:some:prob:1}.

\item Functions like $\sigmaSt{\ell-1}_{|\cvy_{1}^{\ell-1}}$ can be computed
  efficiently by recursive computations. For more details, see
  Appendix~\ref{sec:app:4}.

\item The analogous NFG for the classical setup is shown in
  Fig.~\ref{fig:QFSM:channel:simulation:Y}.

\end{itemize}
Let us also point out that, very often, the desired functions and quantities
are based on the same partial results. The NFG framework is very helpful to
visualize these partial results and to show how these partial results are
combined to obtain the desired functions and quantities.

\section{Supplementary Notes for 
               Section~\ref{sec:information:rate:estimation:1}}
\label{sec:app:4}

In this appendix, we first present a brief summary of the approach for
estimating the information rate of FSMCs as developed
in~\cite{Arnold:Loeliger:Vontobel:Kavcic:Zeng:06:1}, and then extend this
method to quantum channels with memory.

\subsection{Estimation of $I(\sX;\sY)$ for Classical Channels with Memory}
\label{sec:estimation:information:rate:classical:1}

We make the following assumptions.
\begin{itemize}

\item As already mentioned, the derivations in this paper are for the case
  where the input process $\sX = (X_1, X_2, \ldots)$ is an i.i.d.\
  process. The results can be generalized to other stationary ergodic input
  processes that can be represented by a finite-state-machine source
  (FSMS). Technically, this is done by defining a new state that combines the
  FSMS state and the FSMC state.

\item We assume that the FSMC is indecomposable, which roughly means that in
  the long term the behavior of the channel is independent of the initial
  channel state distribution $p_{\tilde S_0}$ (see~\cite{Gallager:68} for the
  exact definition). For such channels and stationary ergodic input processes,
  the information rate $I(\sX;\sY)$ is well defined.

\end{itemize}

\begin{definition}
  The information rate is defined to be
  \begin{align}
    I(\sX; \sY)
      &\defeq 
         \lim_{n \to \infty}
           \frac{1}{n} I(X_1, \ldots, X_n; Y_1, \ldots, Y_n) \ .
             \label{eq:def:fsmc:ir:1}
  \end{align}
  Equivalently, it can be defined as
  \begin{align}
    I(\sX; \sY)
      &= H(\sX) + H(\sY) - H(\sXY) \ ,
          \label{eq:def:fsmc:ir:2}
  \end{align}
  where
  \begin{align*}
    H(\sX)
      &= \lim_{n \to \infty}
           \frac{1}{n} H(\vX_1^n) \ , \\
    H(\sY)
      &= \lim_{n \to \infty}
           \frac{1}{n} H(\vY_1^n) \ , \\
    H(\sXY)
      &= \lim_{n \to \infty}
           \frac{1}{n} H(\vX_1^n, \vY_1^n) \ .
  \end{align*}
\end{definition}

We proceed as in~\cite{Arnold:Loeliger:Vontobel:Kavcic:Zeng:06:1}. (For more
background information, see the references
in~\cite{Arnold:Loeliger:Vontobel:Kavcic:Zeng:06:1}, in
particular~\cite{Ephraim:Merhav:02:1}.) Namely, because
of~\eqref{eq:def:fsmc:ir:2} and because
\begin{alignat}{3}
  &&-
    \frac{1}{n}\log{p(\vX_1^n)}
   &\rightarrow
      H(\sX) \quad &&\text{w.p.\ $1$} \ ,
        \label{eq:converge:x:1} \\
  &&-
    \frac{1}{n}\log{p(\vY_1^n)}
   &\rightarrow
      H(\sY) \quad &&\text{w.p.\ $1$} \ ,
        \label{eq:converge:y:1} \\
  &&-
    \frac{1}{n}\log{p(\vX_1^n, \vY_1^n)}
   &\rightarrow
      H(\sXY) \quad &&\text{w.p.\ $1$} \ ,
        \label{eq:converge:xy:1}
\end{alignat}
we can choose some finite positive integer $n$ and approximate $I(\sX; \sY)$
as follows
\begin{align}
  I(\sX; \sY)
    &\approx
       \hat I(\sX; \sY) \ ,
\end{align}
where
\begin{align}
  \hat I(\sX; \sY)
    &\defeq
       -
       \frac{1}{n}
         \log{p(\cvx_1^n)}
       -
       \frac{1}{n}
         \log{p(\cvy_1^n)}
       +
       \frac{1}{n}
         \log{p(\cvx_1^n,\cvy_1^n)}
           \label{eq:classical:channel:information:rate:estimate:1}
\end{align}
and where $\cvx_1^n$ and $\cvy_1^n$ are some input and output sequences,
respectively, randomly generated according to
\begin{align}
  \label{eq:fsmc:joint:xy:distribution}
  p_{\vX_1^n,\vY_1^n}(\vx_1^n,\vy_1^n) 
    &= \sum_{\tvs_0^n} 
         p_{\tilde S_0}(\tilde s_0)
         \cdot
         Q(\vx_1^n) 
         \cdot 
         W(\vy_1^n, \tvs_1^n | \vx_1^n, \tilde s_0) \ .
\end{align}
Note that $\cvx$ can be obtained by simulating the input process and $\cvy$
can be obtained by simulating the channel for the given input process
realization $\cvx$.

We continue by showing how the three terms appearing on the right-hand side
of~\eqref{eq:classical:channel:information:rate:estimate:1} can be computed
efficiently. We show it explicitly for the second term, and then outline it
for the first and the third term.

In order to efficiently compute the second term on the right-hand side
of~\eqref{eq:classical:channel:information:rate:estimate:1}, \ie,
$-\frac{1}{n}\log{p(\vY_1^n)}$, we consider the \emph{state metric} defined
in~\cite{Arnold:Loeliger:Vontobel:Kavcic:Zeng:06:1} as
\begin{align}
  \muY_{\ell}(\tilde{s}_{\ell})
    &\defeq
       \sum_{\vx_1^{\ell}} 
         \sum_{\tvs_0^{\ell-1}}       
           p_{\tilde S_0}(\tilde s_0)
           \cdot
           Q(\vx_1^{\ell}) 
           \cdot 
           W(\cvy_1^{\ell}, \tvs_1^{\ell} | \vx_1^{\ell}, \tilde s_0) \ .
             \label{eq:def:classical:channel:state:metric:Y:1}
\end{align}
Note that 
\begin{align}
  \label{eq:calculate:p:y:1}
  p_{\vY_1^n}(\cvy_1^n) 
    &= \sum_{\tilde{s}_n} 
         \muY_n(\tilde{s}_n)
\end{align}
and that $\muY_{\ell}(\tilde{s}_{\ell})$ can be calculated recursively via
\begin{align}
  &
  \muY_{\ell}(\tilde{s}_{\ell})
    \nonumber \\
    &\quad
     = \sum_{x_{\ell}}
         \sum_{\tilde{s}_{\ell-1}}
           \muY_{\ell-1}(\tilde{s}_{\ell-1})
           \cdot 
           Q(x_{\ell} | \vx_1^{\ell-1})
           \cdot
           W(\tilde{s}_{\ell}, \cy_{\ell} | \tilde{s}_{\ell-1}, x_{\ell})
             \nonumber \\
    &\quad
     = \sum_{x_{\ell}}
         \sum_{\tilde{s}_{\ell-1}}
           \muY_{\ell-1}(\tilde{s}_{\ell-1})
           \cdot 
           p_X(x_{\ell})
           \cdot
           W(\tilde{s}_{\ell}, \cy_{\ell}| \tilde{s}_{\ell-1}, x_{\ell}) \ .
  \label{eq:recursive:state:metric:Y:1}
\end{align}
These definitions are visualized in Fig.~\ref{fig:CFSM:estimate:hY} by
applying suitable closing-the-box operations to the NFG in
Fig.~\ref{fig:FMSC:high:level:1}.

\bigformulatop{53}{%

  \begin{itemize}
  
  \item Replace~\eqref{eq:fsmc:joint:xy:distribution} by
    \begin{align}
      p_{\vX_1^n,\vY_1^n}(\vx_1^n,\vy_1^n)
        &= \sum_{\vs_0^n, \, {\vs'}_0^n} 
             \rhoSt{0}(s_0, s'_0)
             \cdot
             Q(\vx_1^n)
             \cdot
             \prod_{\ell=1}^{n} 
               \left(
                 W^{(y_{\ell}|x_{\ell})}(s_{\ell-1}, s_{\ell}; s'_{\ell-1}, s'_{\ell})
               \right)
             \cdot 
             \delta(s'_n,s_n) \ .
               \label{eq:quantum:fsmc:joint:xy:distribution}
    \end{align}
    
  \item Replace the state metric $\muY_{\ell}$
    in~\eqref{eq:def:classical:channel:state:metric:Y:1} by the \emph{state
      operator} $\sigmaY_{\ell}$, where
    \begin{align}
      \sigmaY_{\ell}(s_{\ell}, s'_{\ell})
        &\defeq
           \sum_{\vx_1^{\ell}, \, \vs_0^{\ell-1}, \, {\vs'}_0^{\ell-1}}
             \rhoSt{0}(s_0, s'_0)
             \cdot
             Q(\vx_1^{\ell})
             \cdot \,
             \prod_{h=1}^{\ell} 
               W^{(\check y_{h}|x_{h})}(s_{h-1}, s_{h}; s'_{h-1}, s'_{h}) \ .
    \end{align}
  
  \item Replace~\eqref{eq:calculate:p:y:1} by
    \begin{align}
      p_{\vY_1^n}(\vy_1^n)
        &= \sum_{s_n, \, s'_n}
             \sigmaY_n(s_n, s'_n)
             \cdot
             \delta(s'_n,s_n) \ .
    \end{align}
  
  \item Replace~\eqref{eq:recursive:state:metric:Y:1} by
    \begin{align}
      \sigmaY_{\ell}(s_{\ell}, s'_{\ell})
        &\defeq
           \sum_{x_{\ell}}
             \sum_{s_{\ell-1}, \, s'_{\ell-1}}
               \sigmaY_{\ell-1}(s_{\ell-1}, s'_{\ell-1})
               \cdot
               p_X(x_{\ell})
               \cdot
               W^{(\check y_{\ell}|x_{\ell})}(s_{\ell-1}, s_{\ell}; s'_{\ell-1}, s'_{\ell})
                 \ .
                   \label{eq:recursive:quantum:state:metric:Y:1}
    \end{align}
  
  \item Replace~\eqref{eq:recursive:state:metric:Y:2} by
    \begin{align}
      \bsigmaY_{\ell}(s_{\ell}, s'_{\ell})
        &\defeq
           \lambdaY_{\ell}
           \cdot
             \sum_{x_{\ell}}
             \sum_{s_{\ell-1}, \, s'_{\ell-1}}
               \sigmaY_{\ell-1}(s_{\ell-1}, s'_{\ell-1})
               \cdot
               p_X(x_{\ell})
               \cdot
               W^{(\check y_{\ell}|x_{\ell})}(s_{\ell-1}, s_{\ell}; s'_{\ell-1}, s'_{\ell})
                 \ ,
    \end{align}
    where the scaling factor $\lambdaY_{\ell} > 0$ is defined such that
    $\sum_{s_{\ell}, \, s'_{\ell}} \sigmaY_{\ell}(s_{\ell}, s'_{\ell}) \cdot
    \delta(s'_n,s_n) = 1$, \ie, $\tr(\bsigmaY_{\ell}) = 1$.
  
  \item Replace the state metric $\muXY_{\ell}$
    in~\eqref{eq:def:classical:channel:state:metric:XY:1} by the \emph{state
      operator} $\sigmaXY_{\ell}$, where
    \begin{align}
      \sigmaXY_{\ell}(s_{\ell}, s'_{\ell})
        &\defeq
           \sum_{\vs_0^{\ell-1}, \, {\vs'}_0^{\ell-1}}
             \rhoSt{0}(s_0, s'_0)
             \cdot
             Q(\cvx_1^{\ell})
             \cdot
             \prod_{h=1}^{\ell} 
               W^{(\check y_{h}|\check x_{h})}
                 (s_{h-1}, s_{h}; s'_{h-1}, s'_{h}) \ . 
               \label{eq:def:quantum:channel:state:metric:XY:1}
    \end{align}

  \end{itemize}

}

However, since the value of $\muY_{\ell}(\tilde{s}_{\ell})$ tends to zero as
$\ell$ grows, such recursive calculations are numerically unstable. A solution
is to \emph{normalize} $\muY_{\ell}(\tilde s_{\ell})$ during such recursive
calculations and to keep track of the scaling coefficients. Namely,
\begin{align}
  \bmuY_{\ell}(\tilde{s}_{\ell}) 
    &\defeq
       \lambdaY_{\ell}
       \cdot
       \sum_{x_{\ell}}
         \sum_{\tilde{s}_{\ell-1}}
           \bmuY_{\ell-1}(\tilde{s}_{\ell})
           \cdot 
           p_X(x_{\ell})
           \cdot
           W(\tilde{s}_{\ell}, \cy_{\ell}| \tilde{s}_{\ell-1}, x_{\ell}),
             \label{eq:recursive:state:metric:Y:2}
\end{align}
where the scaling factor $\lambdaY_{\ell} > 0$ is defined such that
$\sum_{\tilde{s}_{\ell}} \bmuY_{\ell}(\tilde{s}_{\ell}) = 1$. With this,
Eq.~\eqref{eq:calculate:p:y:1} can be rewritten as
\begin{align}
  \label{eq:calculate:p:y:2}
  p_{\vY_1^n}(\cvy_1^n) 
    &= \prod_{\ell=1}^n 
         (\lambdaY_{\ell})^{-1} \ .
\end{align}
Finally, we arrive at the following efficient procedure for computing
$-\frac{1}{n}\log{p(\cvy_1^n)}$:
\begin{itemize}

\item For $\ell = 1, \ldots, n$, iteratively compute the normalized state
  metric and with that the scaling factors $\lambdaY_{\ell}$.

\item Conclude with the result
  \begin{align}
    \label{eq:calculate:p:y:3}
    -
    \frac{1}{n}
      \log
        p_{\vY_1^n}(\cvy_1^n)
      &= \frac{1}{n}
         \sum_{\ell=1}^n 
           \log(\lambdaY_{\ell}) \ .
  \end{align}

\end{itemize}

The third term on the right-hand side
of~\eqref{eq:classical:channel:information:rate:estimate:1} can be evaluated
by an analogous procedure, where the state metric $\muY_{\ell}(\tilde
s_{\ell})$ is replaced by the state metric
\begin{align}
  \muXY_{\ell}(\tilde{s}_{\ell}) 
    &\defeq
       \sum_{\tvs_0^{\ell-1}} 
         p_{\tilde S_0}(\tilde s_0)
         \cdot
         Q(\cvx_1^{\ell}) 
         \cdot 
         W(\cvy_1^{\ell}, \tvs_1^{\ell} | \cvx_1^{\ell}) \ .
           \label{eq:def:classical:channel:state:metric:XY:1}
\end{align}
The iterative calculation of $\muXY_{\ell}(\tilde{s}_{\ell})$ is visualized in
Fig.~\ref{fig:CFSM:estimate:hXY} by applying suitable closing-the-box
operations to the NFG in Fig.~\ref{fig:FMSC:high:level:1}.

Finally, the first term on the right-hand side
of~\eqref{eq:classical:channel:information:rate:estimate:1} can be trivially
evaluated if $\sX$ is an i.i.d.\ process, and with a similar approach as above
if it is described by a finite-state process.

\subsection{Estimation of $I(\sX; \sY)$ for Quantum Channels with Memory}
\label{sec:estimation:information:rate:quantum:1}

The development in this section is very similar to the development in
Section~\ref{sec:estimation:information:rate:classical:1}. This similarity
stems from the similarity of Figs.~\ref{fig:FMSC:high:level:1}
and~\ref{fig:QFMSC:high:level:1}, and highlights one of the benefits of the
factor-graph approach that we take to estimate information rates of
quantum channels with memory.

We make the following assumptions.
\begin{itemize}

\item As already mentioned, the derivations in this paper are for the case
  where the input process $\sX = (X_1, X_2, \ldots)$ is an i.i.d.\
  process. The results can be generalized to other stationary ergodic input
  processes that can be represented by a finite-state-machine source
  (FSMS). Technically, this is done by defining a new state that combines the
  FSMS state and the channel state.

\item We assume that the quantum channel with memory is
  indecomposable/forgetful, which roughly means that in the long term the
  behavior of the channel is independent of $\rhoSt{0}$ (see
  \cite{Kretschmann:Werner:05:1, Caruso:Giovannetti:Lupo:Mancini:14:1} for
  more details).

\end{itemize}

The changes that are necessary compared to
Appendix~\ref{sec:estimation:information:rate:classical:1} in order to
estimate $I(\sX; \sY)$, are shown in
Eqs.~\eqref{eq:quantum:fsmc:joint:xy:distribution}--\eqref{eq:def:quantum:channel:state:metric:XY:1}
at the top of this page. The corresponding calculations are visualized in
Figs.~\ref{fig:QFSM:estimate:hY} and~\ref{fig:QFSM:estimate:hXY}.

\setcounter{equation}{59}

\section{Supplementary Notes for 
               Section~\ref{sec:numerical:examples:1}}
\label{sec:app:5}

In this appendix we comment on
Figs.~\ref{fig:QGEC:plot:1}--\ref{fig:QGEC:plot:4}. We start by commenting on
the estimated information rate curves.
\begin{itemize}

\item Fig.~\ref{fig:QGEC:plot:1}: as is to be expected, the estimated
  information rate decreases for increasing $\pbad$ in the range $0 \leq \pbad
  \leq 1/2$. This behavior continues for increasing $\pbad$ in the range $1/2
  \leq \pbad \leq 1$. We conclude from this that the receiver has problems
  tracking the state for large $\pbad$.

\item Fig.~\ref{fig:QGEC:plot:2}: as is to be expected, the estimated
  information rate decreases for increasing $\pbad$ in the range $0 \leq \pbad
  \leq 1/2$. This behavior continues only partly for increasing $\pbad$ in the
  range $1/2 \leq \pbad \leq 1$. We conclude from this that when $\pbad$
  approaches $1$, the capabilities of the receiver to track the state improve
  again.

\item Fig.~\ref{fig:QGEC:plot:3}: the larger $\alpha$ is in magnitude, the
  faster the channel state changes, thereby making it often more difficult for
  the receiver to track the state. (Note that the estimated information rate
  is not plotted for $\alpha$ of small magnitude because the channel is only
  slowly mixing for such $\alpha$.)

\item Fig.~\ref{fig:QGEC:plot:4}: similar comments apply here as for
  Fig.~\ref{fig:QGEC:plot:3}.

\end{itemize}

The estimated information rate lower bounds based on mismatched decoders
nicely show the trade-off between computational complexity at the receiver
side and achievable information rates. Interestingly, for some cases the
(classical) $4$-state-auxiliary-channel-based lower bound is rather close to
the estimate information rate.

On the side, note that the estimation of the information rate lower bound
based on a classical auxiliary channel with memory needs only typical input
and output sequences $\cvx_1^n$ and $\cvy_1^n$ of the quantum channel with
memory. The calculations are then done on an NFG representation of the
classical auxiliary channel with memory
(see~\cite{Arnold:Loeliger:Vontobel:Kavcic:Zeng:06:1,
  Sadeghi:Vontobel:Shams:09:1} for details). This is particularly interesting
for scenarios where the simulation of the quantum channel with memory is too
complicated on a classical computer, yet a physical realization of the quantum
channel with memory is available.

Finally, let us point out that, from a practical point of view, we think that
mismatched decoders based on classical auxiliary channels with memory will be
even more important for quantum channels with memory than for classical
channels with memory.

%
%

%
%
%
%
%
%
%
%
%

%

%
%
%
%
%
%
%
%
%
%
%
%
%
%
%
%
%
%
%
%
%
%
%
%
%
%
%
%
%
%
%
%
%
%
%
%
%
%
%
%
%
%
%
%
%
%
%
%
%
%
%
%
%
%
%
%
%
%
%
%
%
%
%
%
%
%
%
%
%
%
%
%
%
%
%
%
%
%
%
%
%
%
%
%
%
%
%
%
%
%
%
%
%
%
%
%
%
%
%
%
%
%
%
%
%
%
%
%
%
%
%
%
%
%
%
%
%
%
%
%
%
%
%
%
%
%
%
%
%
%
%
%
%
%
%
%
%
%
%
%
%
%
%
%
%
%
%
%
%
%
%
%
%
%
%
%
%
%
%
%
%
%
%
%
%
%
%
%
%
%
%
%
%
%
%
%
%
%
%
%
%
%
%
%
%
%
%
%
%
%
%
%
%
%
%
%
%
%
%
%
%
%
%
%
%
%
%
%
%
%
%
%
%
%
%
%
%
%
%
%
%
%
%
%
%
%
%
%
%
%
%
%
%
%
%
%
%
%
%
%
%
%
%
%
%
%
%
%
%
%
%
%
%
%
%
%
%
%
%
%
%
%
%
%
%
%
%
%
%
%
%
%
%
%
%
%
%
%
%
%
%
%
%
%
%
%
%
%
%
%
%
%
%
%
%
%
%
%
%
%
%
%
%
%
%
%
%
%
%
%
%
%
%
%
%
%
%
%
%
%
%
%
%
%
%
%
%
%
%
%
%
%
%
%
%
%
%
%
%
%
%
%
%
%
%
%
%
%
%
%
%
%
%
%
%
%
%
%
%
%
%
%
%
%
%
%
%
%
%
%
%
%
%
%
%
%
%
%
%
%
%
%
%
%
%
%
%
%
%
%
%
%
%
%
%
%
%
%
%
%
%
%
%
%
%
%
%
%
%
%
%
%
%
%
%
%
%
%
%
%
%
%
%
%
%
%
%
%
%
%
%
%
%
%
%
%
%
%
%
%
%
%
%
%
%
%
%
%
%
%
%
%
%
%
%
%
%
%
%
%
%
%
%
%
%
%
%
%
%
%
%
%
%
%
%
%
%
%
%
%
%
%
%
%
%
%
%
%
%
%
%
%
%
%
%
%
%
%
%
%
%
%
%
%
%
%
%
%
%
%
%
%
%
%
%
%
%
%
%
%
%
%
%
%
%
%
%
%
%
%
%
%
%
%
%
%

%
%

%
%
%
%
%
%
%
%
%
%
%
%
%
%
%
%
%
%
%
%
%
%
%
%
%
%
%
%
%
%
%
%
%
%
%
%
%
%
%
%
%
%
%
%
%
%
%
%
%
%
%
%
%
%
%
%
%
%
%
%
%
%
%
%
%
%
%
%
%
%
%
%
%
%
%
%
%
%
%
%
%
%
%
%
%
%
%
%
%
%
%
%
%
%
%
%
%
%
%
%
%
%
%
%
%
%
%
%
%
%
%
%
%
%
%
%
%
%
%
%
%
%
%
%
%
%
%
%
%
%
%
%
%
%
%
%
%
%
%
%
%
%
%
%
%
%
%
%
%
%
%
%
%
%
%
%
%
%
%
%
%
%
%
%
%
%
%
%
%
%
%
%
%
%
%
%
%
%
%
%
%
%
%
%
%
%
%
%
%
%
%
%
%
%
%
%
%
%
%
%
%
%
%
%
%
%
%
%
%
%
%
%
%
%
%
%
%
%
%
%
%
%
%
%
%
%
%
%
%
%
%
%
%
%
%
%
%
%
%
%
%
%
%
%
%
%
%
%
%
%
%
%
%
%
%
%
%
%
%
%
%
%
%
%
%
%
%
%
%
%
%
%
%
%
%
%
%
%
%
%
%
%
%
%
%
%
%
%
%
%
%
%
%
%
%
%
%
%
%
%
%
%
%
%
%
%
%
%
%
%
%
%
%
%
%
%
%
%
%
%
%
%
%
%
%
%
%
%
%
%


%
%
%

%
%
%

\clearpage

\begin{figure*}
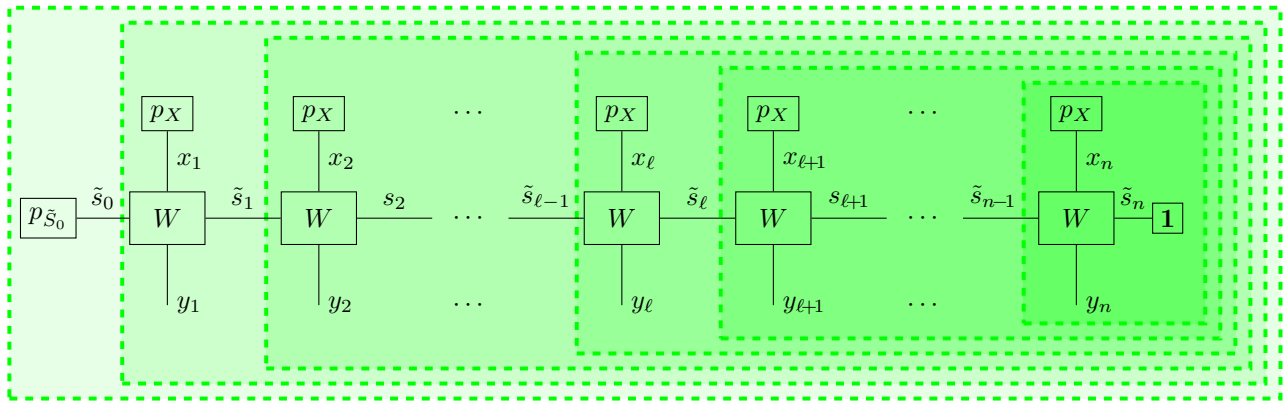

  \centering
  \figCFSMclosingthebox
  \caption{Visualization
    of~\eqref{fig:classical:channel:global:function:sum:1}. Note that every
    closing-the-box operation yields a function node representing the constant
    function $1$.}
  \label{fig:CFSM:closing:the:box}
\end{figure*}

\begin{figure*}
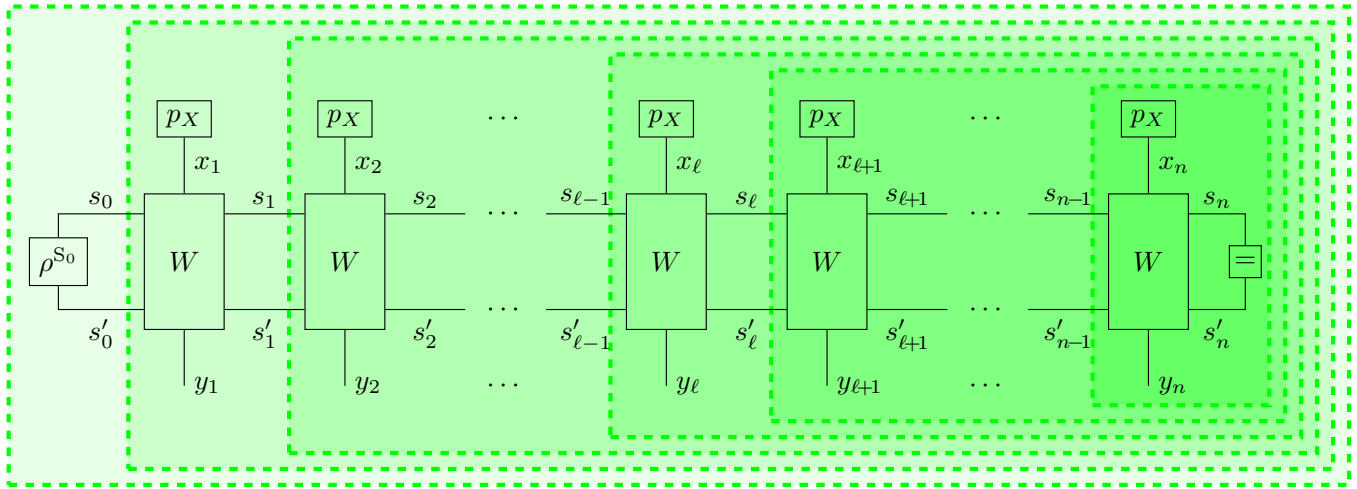

  \centering
  \figQFSMclosingthebox
  \caption{Visualization of~\eqref{eq:quantum:channel:g:sum:property:1}. Note
    that every closing-the-box operation yields a function node representing a
    Kronecker-delta function node, i.e., a degree-two equality function node.}
  \label{fig:QFSM:closing:the:box}
\end{figure*}

\clearpage

\begin{figure*}
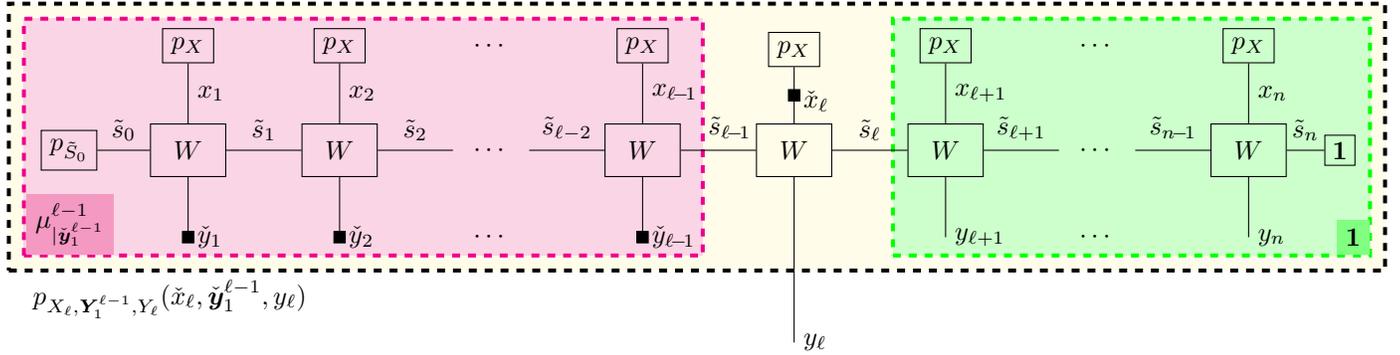

  \centering
  \figCFSMchannelsimulationY
  \caption{Classical-channel-with-memory analog of the NFG in
    Fig.~\ref{fig:CFSM:channel:simulation:Y}.}
\label{fig:QFSM:channel:simulation:Y}
\end{figure*}

\begin{figure*}
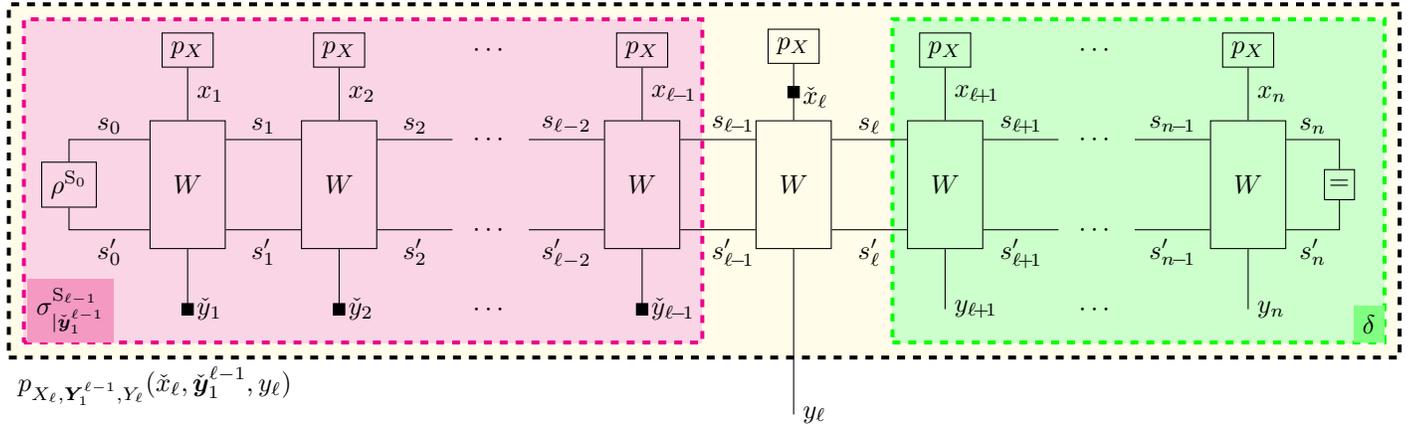

  \centering
  \figQFSMchannelsimulationY
  \caption{Visualization of the computations
    in~\eqref{eq:quantum:state:channel:marginal:example:1} via suitable
    closing-the-box operations. Note that applying the closing-the-box
    operation to the magenta box results in the function
    $\sigmaSt{\ell-1}_{|\cvy_{1}^{\ell-1}}(s_{\ell-1},s'_{\ell-1})$, whereas
    applying the closing-the-box operation to the green box results in a
    Kronecker-delta function, \ie, a degree-$2$ equality function node.}
  \label{fig:CFSM:channel:simulation:Y}
\end{figure*}

\clearpage

\begin{figure*}
  \centering
  \figCFSMestimatehY
  \caption{The iterative computation of $\muY_{\ell}$ as described
    in~\eqref{eq:recursive:state:metric:Y:1} can be understood as a sequence
    of closing-the-box operations as shown above.}
  \label{fig:CFSM:estimate:hY}
\end{figure*}

\begin{figure*}
  \centering
  \figQFSMestimatehY
  \caption{The iterative computation of $\sigmaY_{\ell}$ as described
    in~\eqref{eq:recursive:quantum:state:metric:Y:1} can be understood as a
    sequence of closing-the-box operations as shown above.}
  \label{fig:QFSM:estimate:hY}
\end{figure*}

\clearpage

\begin{figure*}
  \centering
  \figCFSMestimatehXY
  \caption{The iterative computation of $\muXY_{\ell}$ can be understood as a 
    sequence of closing-the-box operations as shown above.}
  \label{fig:CFSM:estimate:hXY}
\end{figure*}

\begin{figure*}
  \centering
  \figQFSMestimatehXY
  \caption{The iterative computation of $\sigmaXY_{\ell}$ as described
    in~\eqref{eq:def:quantum:channel:state:metric:XY:1} can be understood as a
    sequence of closing-the-box operations as shown above.}
  \label{fig:QFSM:estimate:hXY}
\end{figure*}

\fi

\end{document}